\documentclass[vecphys]{svmult}		

\usepackage{amsmath,amssymb}

\usepackage{makeidx}         
\usepackage{graphicx}        
\usepackage{epsfig}          
\usepackage{psfrag}
\usepackage{multicol}        
\usepackage[bottom]{footmisc}

\makeindex             
\begin{document}
\def\rcgindex#1{\index{#1}}
\def\myidxeffect#1{{\bf\large #1}}

\newtheorem{Remark}{Remark}
\renewcommand{\theremark}{\arabic{Remark}}

\title*{Breather solutions of the discrete $p$-Schr\"odinger equation}
\titlerunning{Breather solutions of DpS equation}
\author{
Guillaume James\inst{1,2}
\and
Yuli Starosvetsky\inst{3}
}

\institute{
Laboratoire Jean Kuntzmann, Universit\'e de Grenoble and CNRS,
BP 53, 38041 Grenoble Cedex 9, France \texttt{Guillaume.James@imag.fr}
\and 
INRIA, Bipop Team-Project,
ZIRST Montbonnot, 655 Avenue de l'Europe, 38334 Saint Ismier, France
\and 
Faculty of Mechanical Engineering,
Technion - Israel Institute of Technology,
Technion City, Haifa 32000, Israel
\texttt{staryuli@techunix.technion.ac.il}
}
\maketitle
\abstract
{We consider the discrete $p$-Schr\"odinger (DpS) equation, which approximates
small amplitude oscillations in
chains of oscillators with fully-nonlinear nearest-neighbors interactions
of order $\alpha = p-1 >1$.
Using a mapping approach,
we prove the existence of breather solutions of the DpS equation
with even- or odd-parity reflectional symmetries.
We derive in addition analytical approximations for the breather profiles and the
corresponding intersecting stable and unstable manifolds, valid on a whole
range of nonlinearity orders $\alpha$. 
In the limit of weak nonlinearity ($\alpha \rightarrow 1^+$), we introduce a 
continuum limit connecting the stationary DpS and logarithmic
nonlinear Schr\"odinger equations. In this limit, breathers correspond
asymptotically to Gaussian homoclinic solutions.
We numerically analyze the 
stability properties of breather solutions 
depending on their even- or odd-parity symmetry.
A perturbation of an
unstable breather generally results in a translational motion
(traveling breather) when $\alpha$ is close to unity, whereas 
pinning becomes predominant for larger values of $\alpha$.
}


\section{\label{intro}Introduction and main results}
In this paper
we consider the discrete $p$-Schr\"odinger (DpS) equation 
\rcgindex{\myidxeffect{D}!Discrete $p$-Schr\"odinger (DpS) equation}
defined by
\begin{equation}
\label{dpsa}
i  \frac{d}{d t}\phi_n
=(\Delta_{p}\phi)_n, \ \
n\in \mathbb{Z},
\end{equation}
where $\phi (t)=(\phi_n (t))_{n\in \mathbb{Z}}$ denotes a time-dependent
complex sequence and 
$$
(\Delta_{p}\phi)_n=(\phi_{n+1}-\phi_n)\, |\phi_{n+1}-\phi_n |^{p-2} -
(\phi_{n}-\phi_{n-1})\, |\phi_{n}-\phi_{n-1} |^{p-2} 
$$
the discrete $p$-Laplacian with $p>2$. 
This model has a Hamiltonian structure associated with the energy
\begin{equation}
\label{defh}
H(\phi )= \frac{2}{p}\, {\sum_{n \in \mathbb{Z}}{|\phi_{n+1}-\phi_n|^{p}}}.
\end{equation}
It is reminiscent of the
so-called discrete nonlinear Schr{\"o}dinger (DNLS) equation \cite{pgk,eilbeckj},
with the difference that it contains a fully nonlinear inter-site coupling 
\rcgindex{\myidxeffect{F}!Fully nonlinear inter-site coupling}
term.
The DpS equation was recently introduced to describe small amplitude oscillations in 
a class of mechanical systems consisting of a chain of touching beads confined in 
smooth local potentials 
\cite{James11,JKC11,sahvm},
the most well known example of such systems
being Newton's cradle \cite{hutzler}. In this context, the $p$-Laplacian involved
in (\ref{dpsa}) accounts for the fully-nonlinear character of Hertzian interactions
between beads (with $p=5/2$ in the case of contacting spheres). More generally, the DpS equation 
arises in the study of 
chains of oscillators involving fully nonlinear nearest-neighbors interactions.
It can be derived as a modulation equation when nonlinear interactions dominate on-site anharmonicity,
and allows to describe the slow evolution of the envelope of small amplitude oscillations
over long (but finite) times \cite{bdj}. A variant of the DpS equation combining
an on-site cubic nonlinearity and a cubic intersite nonlinearity (case $p=4$)
was also introduced in \cite{ros} in a case when nonlinear interactions and 
on-site anharmonicity have equal orders. 

Numerical simulations of the DpS equation with $p=5/2$
have revealed the existence 
of breather 
\rcgindex{\myidxeffect{B}!Breather}
solutions (i.e. spatially localized oscillations) of different types, 
either static (i.e. pinned to some lattice sites and time-periodic) or
traveling along the lattice. Such solutions
can be generated from localized initial conditions or may arise from modulational
instabilities of periodic waves \cite{James11,JKC11,sahvm}. Although such
properties are classical in the context of anharmonic Hamiltonian lattices \cite{flg},
energy localization is particularly strong in the DpS equation. Indeed, 
it is proved in \cite{bdj} that the solution of
(\ref{dpsa}) does not disperse for any nonzero initial condition in $\ell_2 (\mathbb{Z})$
(i.e. with finite power). Moreover, 
static breather solutions decay doubly-exponentially in space, as numerically
illustrated in \cite{James11} (see also \cite{nolta12} for an analytical decay estimate).

In this paper, we prove an existence theorem for
static breather solutions of the DpS equation valid for all $p>2$
and introduce 
different methods to approach their profiles. The
breather solutions are searched in the form 
\begin{equation}
\label{solper}
\phi_n (t ) = {|\Omega |}^{\frac{1}{p -2}}  a_n\, e^{i\, \Omega \, t},
\end{equation}
with $a=(a_n)_{n \in \mathbb{Z}}$ a real sequence and $\Omega$
an arbitrary nonvanishing constant.
Equation (\ref{dpsa}) admits time-periodic solutions of the form (\ref{solper})
if and only if $a$ satisfies the stationary DpS equation
\begin{equation}
\label{dpsstat}
-s\, {a_n}= (\Delta_{p}a)_n, \ \
n\in \mathbb{Z},
\end{equation} 
where $s=\mbox{Sign}(\Omega)=1$ or $-1$.
We shall prove the following theorem, which ensures
the existence of nontrivial solutions of (\ref{dpsstat}) homoclinic 
\rcgindex{\myidxeffect{H}!Homoclinic solution}
to $0$ 
(i.e. satisfying $\lim_{n\rightarrow \pm\infty}{a_n}=0$) for $s=1$,
and the nonexistence of (nonzero) bounded solutions for $s=-1$. 

\begin{theorem}
\label{mainthm}
The stationary DpS equation (\ref{dpsstat}) with $s=1$ admits
solutions $a_n^i$ ($i=1,2$) satisfying 
$$
\lim_{n\rightarrow \pm\infty}a_n^i=0,
$$ 
$(-1)^n\, a_{n}^i >0$, 
$|a_n^i | > |a_{n-1}^i | $ for all $n\leq 0$, and
$$
a_n^1 =a_{-n}^1,  \ \ \ 
a_n^2 = -a_{-n+1}^2, 
 \ \ \  \mbox{ for all } n \in \mathbb{Z}.
$$
Moreover, for $s=-1$ the only bounded solution of (\ref{dpsstat}) is $a_n=0$.
\end{theorem}

For $s=1$, the solutions $a_n^i$ ($i=1,2$) of (\ref{dpsstat}) homoclinic to $0$ 
correspond to breathers solutions of (\ref{dpsa}) taking the form (\ref{solper}) with 
arbitrary $\Omega >0$.
The case $i=1$ yields the so-called even-parity modes (or site-centered solutions),
and the case $i=2$ corresponds to odd-parity modes (bond-centered solutions).
These solutions were numerically computed in \cite{James11,sahvm,JKC11} for $p = 5/2$
(see also \cite{flg,ros} and references therein in the case $p=4$ with
an additional on-site cubic nonlinearity).

Theorem \ref{mainthm} is proved using a map approach reminiscent
of previous works of Flach \cite{flach2} and Qin and Xiao \cite{qin} (see also \cite{ros}). More precisely,
a proof of existence of homoclinic solutions has been proposed in \cite{flach2}
(when $p \geq 4$ is an even integer), relying on a reformulation of (\ref{dpsstat}) as a two-dimensional
reversible mapping. 
\rcgindex{\myidxeffect{R}!Reversible mapping} 
However, the analysis was not complete because the 
existence of stable and unstable manifolds 
\rcgindex{\myidxeffect{S}!Stable manifold}
\rcgindex{\myidxeffect{U}!Unstable manifold}
of the ``hyperbolic" fixed point at the
origin was not analytically justified. This point requires
a particular attention because the corresponding map is not differentiable at the origin.
This problem was addressed in \cite{flach2} by adding a small linear Laplacian 
$\epsilon (\Delta_{2}a)_n$ in (\ref{dpsstat}), and studying the limit $\epsilon \rightarrow 0$ 
numerically. In the same spirit,
equation (\ref{dpsstat}) can be recast into the form of a generalized stationary
discrete nonlinear Schr\"odinger (DNLS) equation reminiscent of a model studied in \cite{qin}, 
except the resulting nonlinearity is not differentiable at the origin.
The work \cite{qin}
establishes the existence of homoclinic solutions for this class of models when
$C^1$ stable and unstable manifolds exist, a point requiring special justification in our context.
In this paper, we give equivalent reformulations of (\ref{dpsstat}) as two-dimensional
mappings (section \ref{maps}), and 
prove the existence of smooth stable and unstable manifolds of the origin (section \ref{sum}). 
This allows us to show the 
existence of homoclinic solutions of (\ref{dpsstat}) for $s=1$ using Qin and Xiao's result (section \ref{homsol}).
The nonexistence of nontrivial bounded solutions for $s=-1$ follows from
elementary arguments (see the appendix).

Apart from existence results, several works have introduced
analytical approximations of homoclinic orbits of the stationary DpS equation
having compact supports, relying on two different approaches. The first one is 
reminiscent of the method of successive approximations \cite{page,sievpage}, 
but presents some limitations since a convergence analysis is not available.
The second one relies on formal continuum approximations
which provide quantitatively correct results at least for $p=5/2$ and $p=4$
\cite{kivcomp,sahvm,JKC11}. However, in these works the continuum 
approximation is performed with a finite mesh size, so that 
the continuum problem is not properly justified.

In this paper, we proceed differently and
provide analytical approximations of the stable and unstable manifolds of the origin
(section \ref{sectapprox}).
This allow us to approximate the intersections of the invariant
manifolds and thus the corresponding homoclinic orbits. A similar 
approach was previously introduced in \cite{jesus} to approximate homoclinic
solutions of the DNLS equation, but the technique employed to
approximate the invariant manifolds was different.  
In our case, two different methods are used in this context.
The first one employs a leading order approximation of the local stable manifold
in conjunction with backward iterations, a method especially efficient when
$p$ is far from $2$. The second method is based on a continuum limit
\rcgindex{\myidxeffect{C}!Continuum limit}
obtained when $p$ is close to $2$, where one recovers a logarithmic stationary
nonlinear Schr\"odinger equation 
\rcgindex{\myidxeffect{L}!Logarithmic stationary nonlinear Schr\"odinger equation}
with a Gaussian homoclinic solution. 
This limiting procedure is a first step towards the rigorous justification of previous
formal continuum approximations \cite{kivcomp,sahvm,JKC11}.
Moreover, it is interesting to note that 
the case $p\approx 2$ is physically sound in the context of 
granular chain models \cite{khatri,sun,daraio}.

We complete the above existence and approximation results by a numerical study
of the stability of site- and bond-centered breathers (section \ref{stab}). Similarly to the case
$p=5/2$ considered in \cite{JKC11}, we find that bond-centered breathers are spectrally stable
and site-centered breathers are unstable, at least for all $p\in [2.2,4.4]$. 
When $p$ is close enough to $2$ (e.g. for $p \leq 5/2$), 
site-centered breather instability is very weak, and perturbations 
along a marginal mode \cite{aubryC} lead to traveling breathers
\rcgindex{\myidxeffect{T}!Traveling breathers}
propagating at an almost constant velocity. When $p$ becomes larger (e.g. for $p \geq 4$),
suitable perturbations of a site-centered breather can induce
a translational motion over a few sites, but trapping of the localized solution occurs subsequently. 
We relate this pinning 
\rcgindex{\myidxeffect{P}!Pinning}
effect to a sharp increase of the Peierls-Nabarro energy barrier 
\rcgindex{\myidxeffect{P}!Peierls-Nabarro energy barrier}
\cite{kivc} separating site- and bond-centered breathers.

Another interest of the analytical approximations of stable and unstable manifolds
obtained in section \ref{sectapprox}
concerns the analysis of breather bifurcations induced by a localized defect.
Indeed, as shown in \cite{james1} (see also \cite{palmero}), such bifurcations
possess a geometrical interpretation for DNLS-type systems. In the presence of
an isolated defect, an homoclinic orbit exists when the image of the unstable
manifold by some linear transformation (which depends on the defect strength) 
intersects the stable manifold. Consequently, homoclinic bifurcations occur
when the set of these intersections changes topologically. This typically corresponds
to tangent or pitchfork bifurcations 
of some intersections for critical defect values.
Obviously, 
using this approach to analyze defect-induced homoclinic bifurcations 
requires a good knowledge
of the geometrical structure of the stable and unstable manifolds, whereas
their geometry is often hard to establish rigorously. In that case, a good strategy consists in approximating
the stable and unstable manifolds and perform the above analysis 
on the approximate invariant manifolds, which provides approximations of critical
defect values \cite{james1}. This approach is briefly discussed in
section \ref{sectdiscuss}, and will be used in a forthcoming paper to analyze
defect-induced breather bifurcations in the DpS equation.

\section{\label{maps} Two-dimensional mappings equivalent to the stationary DpS equation}

In this section, we introduce some reformulations of the stationary DpS equation
as two-dimensional reversible mappings. 
The corresponding variables will
be denoted by force variables and mixed (amplitude-force) variables,
in analogy with the Newton's cradle system \cite{James11}.
For spatially homogeneous systems,
the force variables described in section \ref{fov} 
allow to convert the stationary DpS equation
into a generalized DNLS equation with an on-site (fractional power) nonlinearity,
which allows to prove the existence of homoclinic orbits by classical arguments
(see section \ref{homsol}). However, this formulation
would lead to some complications when considering
DpS equations with a local inhomogeneity (see section \ref{sectdiscuss}). 
The mixed variables described in section \ref{miv} 
will turn out to be more convenient in that case.
They will be used in section \ref{sectapprox} when approximating
the stable and unstable manifolds of the origin.

\vspace{1ex}

In this section we restrict to the case $s=1$ of (\ref{dpsstat}).
The even- and odd-parity localized solutions of (\ref{dpsstat})
described in theorem \ref{mainthm}
consist of spatial modulations of binary oscillations, 
i.e. $a_n$ and $a_{n+1}$ have opposite signs. This leads us to perform
the so-called staggering transformation $a_n = (-1)^n\, u_n$, which yields
the equivalent form of (\ref{dpsstat}) given by
\begin{equation}
\label{dpsstatbis}
s\, {u_n}= P_\alpha (u_{n+1}+u_{n})+P_\alpha (u_{n}+u_{n-1}), \ \
n\in \mathbb{Z},
\end{equation} 
where 
$$P_\alpha (x)=x\, {|x|}^{\alpha -1}
$$ 
and $\alpha =p-1 >1$. 
With this change of variables, the even- and odd-parity localized solutions of (\ref{dpsstat})
correspond to one-sign localized solutions $u_n$ of (\ref{dpsstatbis}).

The nonlinear system (\ref{dpsstatbis}) with $s=1$ can be rewritten 
as a two-dimensional mapping in different ways. 
The most direct one is obtained by 
introducing $w_n = u_{n+1}$ and using the variable $(u_n, w_n)$.
Since the map $P_\alpha$ defines an homeomorphism of 
$\mathbb{R}$ and $P_\alpha^{-1}=P_{1/\alpha}$, (\ref{dpsstatbis})
can be rewritten 
\begin{equation}
\label{mappingm}
\begin{pmatrix} u_{n+1} \\ w_{n+1}\end{pmatrix}
=M\, 
\begin{pmatrix} u_{n} \\ w_{n}\end{pmatrix} ,
\end{equation} 
where $M$ is the nonlinear map defined by
\begin{equation}
\label{mapm}
M\, 
\begin{pmatrix} u \\ w \end{pmatrix}
= \begin{pmatrix} w \\ P_{1/\alpha}[w-P_\alpha (u+w)] -w \end{pmatrix} .
\end{equation} 
Notice that $P_{1/\alpha}$ and $M$ are
not differentiable at the origin since $\alpha >1$.
Because the expression of the map $M$ is quite cumbersome, we shall replace
(\ref{mappingm}) by two different formulations, which involve the force and mixed
variables described below. 

\subsection{\label{fov}Force variables}

Defining 
\begin{equation}
\label{defforvar}
x_n = P_\alpha (u_{n+1}+u_{n})
\end{equation}
and summing (\ref{dpsstatbis}) (with $s=1$) at ranks 
$n$ and $n+1$ yields the following equivalent problem,
\begin{equation}
\label{dnlsstat2}
(\Delta x )_n = q(x_n), \ \ \
x=(x_n)_{n \in \mathbb{Z}},
\end{equation} 
where $\Delta$ is the usual discrete Laplacian
$$
(\Delta x )_n = x_{n+1} -2x_n + x_{n-1}, \ \ \
x=(x_n)_{n \in \mathbb{Z}},
$$
and for all $x \in \mathbb{R}$
\begin{equation}
\label{defq}
q(x)=P_{1/\alpha}(x )-4x
\end{equation} 
(we omit the $\alpha$-dependency of $q$ in notations for simplicity).
System (\ref{dnlsstat2}) falls within the class of generalized stationary DNLS equations studied by
Qin and Xiao in \cite{qin}, except $q$ is not differentiable at the origin. 
Introducing $z_n = x_{n+1}$, one can reformulate (\ref{dnlsstat2})
as a two-dimensional mapping as in \cite{qin}
\begin{equation}
\label{mappingt}
\begin{pmatrix} x_{n+1} \\ z_{n+1}\end{pmatrix}
=T\, 
\begin{pmatrix} x_{n} \\ z_{n}\end{pmatrix} ,
\end{equation} 
where
\begin{equation}
\label{mapt}
T\, 
\begin{pmatrix} x \\ z \end{pmatrix}
= \begin{pmatrix} z \\ q(z)+2z-x \end{pmatrix} .
\end{equation} 
The map $T$ is an homeomorphism of $\mathbb{R}^2$ with inverse
$$
T^{-1}\, 
\begin{pmatrix} x \\ z \end{pmatrix}
= \begin{pmatrix} q(x)+2x-z \\ x \end{pmatrix} .
$$ 
Moreover, $T$ is
reversible with respect to the symmetry 
$$
R_1 = \begin{pmatrix} 0 & 1 \\ 1 &0 \end{pmatrix} 
$$
and the nonlinear symmetry $R_2 = R_1 \circ T$,
i.e. we have $R_i \circ T=T^{-1}\circ R_i$ for $i=1,2$. 
In other words, any orbit $X_n = (x_n , z_n)^T$ of (\ref{mappingt})
yields by symmetry other orbits $\tilde{X}_n^i$ ($i=1,2$) taking the form
$\tilde{X}_n^i = R_i (X_{-n}) $. 

\vspace{1ex}

In order to stress the equivalence between formulations (\ref{mappingm}) and (\ref{mappingt}), 
let us consider the homeomorphism of $\mathbb{R}^2$
\begin{equation}
\label{maph}
h\, 
\begin{pmatrix} u \\ w \end{pmatrix}
= \begin{pmatrix} P_\alpha (u+w) \\ w-P_\alpha (u+w) \end{pmatrix} ,
\end{equation} 
whose inverse reads
\begin{equation}
\label{maphm1}
h^{-1}\, 
\begin{pmatrix} x \\ z \end{pmatrix}
= \begin{pmatrix} P_{1/\alpha} (x)-x-z \\ x+z \end{pmatrix} .
\end{equation} 
Note that $h$ defines a diffeomorphism on $\mathbb{R}^2 \setminus \{ 0 \}$.
The following result can be checked by a simple computation.

\begin{lemma}
\label{conju}
The maps $T$ and $M$ are topologically conjugate by $h$, i.e. 
they satisfy $M=h^{-1}\circ T \circ h$.
\end{lemma}

\subsection{\label{miv}Mixed variables}

Let us consider $v_n = u_{n-1}$ and 
\begin{equation}
\label{defyn}
y_n = x_{n-1}=P_\alpha (v_{n+1}+v_{n}).
\end{equation} 
Equation (\ref{dpsstatbis}) reads for $s=1$
\begin{equation}
\label{dpsstattrans}
v_{n+1}=y_{n+1}+y_{n}. 
\end{equation} 
This yields the mapping
\begin{equation}
\label{mappingn}
\begin{pmatrix} v_{n+1} \\ y_{n+1}\end{pmatrix}
=F\, 
\begin{pmatrix} v_{n} \\ y_{n}\end{pmatrix} ,
\end{equation} 
where
\begin{equation}
\label{mapn}
F\, 
\begin{pmatrix} v \\ y \end{pmatrix}
= \begin{pmatrix} P_{1/\alpha} (y)-v \\ 
P_{1/\alpha} (y)-y-v \end{pmatrix} .
\end{equation} 
The first component of (\ref{mappingn}) corresponds to equation (\ref{defyn})
and the second to (\ref{dpsstattrans}).

The maps $T$ and $F$ are
topologically conjugate by the linear transformation
$$
J = \begin{pmatrix} 1 & 1 \\ 0 &1 \end{pmatrix}, 
$$
i.e. one has $F=J\circ T \circ J^{-1}$.
It follows in particular that $F$ is
reversible with respect to the symmetry 
\begin{equation}
\label{reversr1}
\mathcal{R}_1 = J\, R_1 J^{-1} = 
\begin{pmatrix} 1 & 0 \\ 1 &-1 \end{pmatrix} 
\end{equation}
and the nonlinear symmetry $\mathcal{R}_2 = \mathcal{R}_1 \circ F$ taking the form
$$
\mathcal{R}_2 
\begin{pmatrix} v \\ 
y \end{pmatrix}
= \begin{pmatrix} P_{1/\alpha} (y)-v \\ 
y \end{pmatrix} .
$$
 
\section{\label{sum} Stable and unstable manifolds}

In this section we construct the stable and unstable manifolds of the origin for the
different maps introduced in section \ref{maps}. These results will be used in section
\ref{homsol} to obtain homoclinic solutions as intersections of these manifolds,
and in section \ref{sectapprox} to approximate the stable and unstable manifolds
and homoclinic orbits.

\subsection{Heuristics}

Let us first consider the generalized DNLS equation (\ref{dnlsstat2}) in the usual case when
\begin{equation}
\label{cond1}
q\in C^1 (\mathbb{R}),\ \ \  q^\prime (0)>0 
\end{equation}
and the corresponding mapping (\ref{mappingt}).
Condition (\ref{cond1}) implies that the origin is a hyperbolic
fixed point of $T$, since the Jacobian matrix $DT(0)$ has a pair of
real eigenvalues $\Lambda , \Lambda^{-1}$ given by
\begin{equation}
\label{vp}
\Lambda = 1+ \frac{q^\prime (0)}{2}+\left(\,  \frac{q^\prime (0)^2}{4} + q^\prime (0)  \, \right)^{1/2} >1.
\end{equation}
This yields the existence of  
$C^1$ stable and unstable manifolds of the origin, denoted respectively by
$W^{\text{s}}(0)$ and $W^{\text{u}}(0)$.
At the origin, $W^{\text{u}}(0)$ is tangent to the unstable subspace corresponding
to the line $x = \Lambda^{-1}\, z$, and $W^{\text{s}}(0)$ is tangent to the stable subspace
$z = \Lambda^{-1}\, x$.

For system (\ref{dnlsstat2})-(\ref{defq}), 
$q$ is not differentiable at the origin 
and $\lim_{x\rightarrow 0^{\pm}}{q^\prime (x)} =+\infty$. This situation
corresponds formally to the existence of an hyperbolic fixed point of $T$ at the origin,
the eigenvalues of the singular ``Jacobian matrix" being $\Lambda =+\infty$ and $\Lambda^{-1}=0$,
the unstable subspace corresponding
to the axis $x = 0$ and the stable subspace to the axis $z = 0$.
In what follows, we put these arguments onto a rigorous footing, showing
the existence of $C^1$ stable and unstable manifolds of the origin.
As usual, we shall first prove the existence of local stable and unstable
manifolds which can be written as graphs near the origin, and then
construct the global stable and unstable manifolds.

\subsection{Stable manifold theorem for system (\ref{dpsstatbis})}

In order to obtain local stable and unstable manifolds for (\ref{mappingt}), we first 
prove a similar result for system (\ref{dpsstatbis}), which has the advantage of being $C^1$.
For completeness we shall treat the two cases $s=1$ and $s=-1$ simultaneously, but 
the applications considered subsequently will only concern 
the case $s=1$.

Let us consider the usual Banach space $\ell_\infty (\mathbb{N})$ consisting of real bounded sequences,
equiped with the supremum norm.
For all $u_0 \in \mathbb{R}$ and
$u = (u_n )_{n\geq 1} \in \ell_\infty (\mathbb{N})$, we define
$N(u_0 ,u)\in \ell_\infty (\mathbb{N})$ by
$$
(N(u_0 ,u))_n=P_\alpha (u_{n+1}+u_{n})+P_\alpha (u_{n}+u_{n-1}), \ \
n\geq 1 .
$$
System (\ref{dpsstatbis}) restricted to $n\geq 1$ takes the form
\begin{equation}
\label{dpsstatbis2}
g(u_0 , u )=0 
\end{equation} 
where $g(u_0 , u )=s\, u-N(u_0 ,u)$ satisfies $g\in C^1 (\mathbb{R}\times \ell_\infty (\mathbb{N}),\ell_\infty (\mathbb{N}))$,
$g(0,0)=0$, $D_u g(0,0)=s\, \mbox{Id}$. The following result 
is one of the main steps to construct the
local stable manifold of the origin. Below, $B_\infty (R)$ denotes the open ball of radius $R>0$ centered at the origin
in $\ell_\infty (\mathbb{N})$ and we denote $B(\varepsilon )=(-\varepsilon , \varepsilon )$.

\begin{theorem}
\label{unstm}
There exist $\varepsilon , R >0$ and an odd function $\psi \in C^1 (B(\varepsilon ),B_\infty (R))$, such that
$g(u_0 , \psi (u_0) )=0 $ for all $u_0 \in B(\varepsilon )$. Moreover, for all solution of (\ref{dpsstatbis2}) in
$B(\varepsilon ) \times B_\infty (R)$ one has $u=\psi (u_0)$. In addition, 
$\psi = (\psi_k)_{k\geq 1}$ possesses the following properties.
\begin{itemize}
\item[i)]
The function $u_0 \mapsto \psi_1 (u_0)$ is increasing for $s=1$,
decreasing for $s=-1$, and satisfies
\begin{equation}
\label{decpsi1}
| \psi_1 (u_0 ) | <|u_0| \mbox{ for all } u_0 \in B(\varepsilon ),
\end{equation}
\begin{equation}
\label{estpsi1}
\psi_1 (u_0) = s\, P_\alpha (u_0) + \mbox{O}(|u_0|^{2\alpha -1})
\mbox{ when } u_0 \rightarrow 0.
\end{equation}
\item[ii)] 
The functions $u_0 \mapsto \psi_k (u_0)$ ($k\geq 1$) satisfy 
\begin{equation}
\label{formpsi}
\psi_k = \overbrace{\psi_1 \circ \psi_1 \cdots \circ \psi_1}^{k \mbox{ times}},
\end{equation}
\begin{equation}
\label{eqfoncpsi1}
s\, \psi_1 = P_\alpha \circ (
{\rm Id}
+\psi_1 ) \circ \psi_1 + P_\alpha \circ (
{\rm Id}
+\psi_1 ).
\end{equation}
\item[iii)] 
For all $u_0 \in B(\varepsilon )$ we have 
$\lim_{k\rightarrow +\infty}{\psi_k (u_0)}=0$. Moreover,
for all $C>1$, there exists $\varepsilon_C \leq \varepsilon $ such that for all $u_0 \in B(\varepsilon_C )$ we have
\begin{equation}
\label{bounddec}
|\psi_k (u_0) |\leq  C^{\frac{1}{1-\alpha}}\, {\left( C^{\frac{1}{\alpha -1}} |u_0| \right)}^{\alpha^k}.
\end{equation}
\end{itemize}
\end{theorem}

\begin{proof}
By a direct application of the implicit function theorem, 
the solutions of (\ref{dpsstatbis2}) sufficiently close to $0$ in $\mathbb{R}\times \ell_\infty (\mathbb{N})$ have
the form $u=\psi (u_0)$, where $\psi$ is odd due to the invariance $u_n \rightarrow -u_n$ of (\ref{dpsstatbis}).
Since we have
\begin{equation}
\label{eqpsi}
s\, \psi (u_0 )=N(u_0 , \psi(u_0 )),
\end{equation}
it follows that 
\begin{equation}
\label{bigo}
\psi (u_0 )=\mbox{O}(|u_0|^\alpha) \mbox{ in } \ell_\infty (\mathbb{N}) \mbox{ when }
u_0 \rightarrow 0.
\end{equation} 
In particular, by choosing $\varepsilon$ small enough we have 
\begin{equation}
\label{decpsi}
\| \psi (u_0 ) \|_{\infty} <|u_0| \mbox{ for all } u_0 \in B(\varepsilon ),
\end{equation}
hence (\ref{decpsi1}) holds true.
Moreover, by identifying the first terms of the sequences at both sides of (\ref{eqpsi}), we get
\begin{equation}
\label{eqpsi1}
s\, \psi_1 (u_0 ) =P_\alpha (\psi_2(u_0)+\psi_1(u_0))+P_\alpha (\psi_1(u_0)+u_0),
\end{equation}
which yields (\ref{estpsi1}) after elementary computations based on estimate (\ref{bigo}).

System (\ref{dpsstatbis2}) is invariant by index shift, i.e. if $(u_0,(u_{n})_{n\geq 1})$ is a solution then, for all $k\geq 1$,
$(u_k,(u_{k+n})_{n\geq 1})$ is also a solution. Combining this invariance with the above reduction
implies $\psi_{n+k}(u_0)=\psi_n (\psi_k (u_0))$
for all $k, n \geq 1$. In particular, $\psi_{1+k}(u_0)=\psi_1 (\psi_k (u_0))$ yields (\ref{formpsi}),
and (\ref{eqfoncpsi1}) follows by combining (\ref{eqpsi1}) and (\ref{formpsi}).
In addition, differentiating (\ref{eqfoncpsi1}) with respect to $u_0$ yields after lengthy but straightforward computations
$$
\psi_1^\prime (u_0) = s\, P_\alpha^\prime (u_0) + \mbox{O}(|u_0|^{2\alpha -2}),
$$
hence $\psi_1$ is increasing on $B(\varepsilon )$ for $s=1$ and
decreasing on $B(\varepsilon )$ for $s=-1$
provided $\varepsilon$ is chosen small enough.

Properties (\ref{formpsi}) and (\ref{decpsi1}) imply
$\lim_{k\rightarrow +\infty}{\psi_k (u_0)}=0$
for all $u_0 \in B(\varepsilon )$. Moreover,
according to (\ref{estpsi1}), for all $C>1$ there exists $\varepsilon_C >0$ such that for $|u_0 | < \varepsilon_C$
we have 
\begin{equation}
\label{estpsi1bis}
|\psi_1 (u_0) | \leq C\, |u_0|^{\alpha} ,
\end{equation}
which gives in conjunction with (\ref{formpsi})
\begin{equation}
\label{estpsik}
|\psi_{k+1} (u_0) | \leq C\, |\psi_{k} (u_0) |^{\alpha} , \mbox{ for all } k\geq 1.
\end{equation}
Then using both (\ref{estpsi1bis}) and (\ref{estpsik}) yields (\ref{bounddec}) by induction.
{\hfill $\Box$}\end{proof}

\subsection{\label{defsumf}Stable and unstable manifolds for the maps $M,T,F$}

In this section we consider the case $s=1$ of (\ref{dpsstatbis}), and use theorem \ref{unstm}
to construct the stable and unstable manifolds of the origin for the maps $M,T,F$.

Introducing $w_n = u_{n+1}$, one can reformulate (\ref{dpsstatbis})
as a two-dimensional mapping defined by (\ref{mappingm}).
Let us consider $\varepsilon , R ,\psi_1$ as in theorem \ref{unstm}.
Since the statement of theorem \ref{unstm} remains true replacing 
$\varepsilon$ by $\varepsilon^\prime < \varepsilon$, one can assume $\varepsilon < R$ 
without loss of generality. Let us define $\Omega = (-\varepsilon , \varepsilon )^2$ and consider
the $C^1$ one-dimensional submanifold of $\mathbb{R}^2$
\begin{equation}
\label{graphewu}
W^{\text{s}}_{\text{loc}}(0) =
\{ \, (u_0 , \psi_1 (u_0 ))\in \mathbb{R}^2 , |u_0 | < \varepsilon \, \} \subset \Omega.
\end{equation} 

The following result establishes that $W^{\text{s}}_{\text{loc}}(0)$ is the local stable manifold
of the fixed point $0$ of $M$.
In the sequel we call
a curve $\Gamma$ {\em negatively invariant} by an invertible map $M$
if $M^{-1}(\Gamma )\subset \Gamma$, and {\em positively invariant} by $M$
if $M(\Gamma )\subset \Gamma$.

\begin{theorem}
\label{stablemm}
The manifold $W^{\text{s}}_{\text{loc}}(0)$ possesses the following properties.
\begin{itemize}
\item[i)]
$W^{\text{s}}_{\text{loc}}(0)$ is positively invariant by $M$.
\item[ii)]
If $(u_0 , w_0) \in W^{\text{s}}_{\text{loc}}(0)$, then the corresponding solution
of (\ref{mappingm}) satisfies 
$$
\lim_{n\rightarrow +\infty}{(u_n , w_n )}=0.
$$
\item[iii)]
All solution of (\ref{mappingm}) such that $(u_n , w_n )\in \Omega$ for all $n\geq 0$
satisfies $(u_n , w_n) \in W^{\text{s}}_{\text{loc}}(0)$ for all $n\geq 0$.
\end{itemize}
\end{theorem}
\begin{proof}
Let $(u_0 , w_0) \in W^{\text{s}}_{\text{loc}}(0)$ and 
$(u_n , w_{n})=(u_n , u_{n+1})$ the corresponding solution
of (\ref{mappingm}). Since $w_0 = \psi_1(u_0)$, we have
$(u_n , w_{n})=(\psi_n (u_0) , \psi_{n+1}(u_0))$ for all $n\geq 1$, by
uniqueness of the solution of (\ref{dpsstatbis}) given an initial condition.
Then properties {\em i)} and {\em ii)} follow from properties {\em ii)} 
and {\em iii)} of theorem \ref{unstm}.

Let $(u_n , w_{n})=(u_n , u_{n+1})$ denote a solution
of (\ref{mappingm}) staying in $\Omega$ for all $n\geq 0$.
Since $|u_0 |<\varepsilon $ and $\sup_{n\geq1}{|u_n|}\leq \varepsilon < R$,
theorem \ref{unstm} ensures that $u_n = \psi_n (u_0)$ for all $n\geq 1$,
hence $(u_n , w_n) \in W^{\text{s}}_{\text{loc}}(0)$ for all $n\geq 0$.
{\hfill $\Box$}\end{proof}

Now, to obtain a local stable manifold of the origin for the map $T$ defined by (\ref{mapt}), 
we use the fact that $T$ and $M$ are topologically conjugate by the
homeomorphism $h$ defined in (\ref{maph}).
Let us define 
\begin{equation}
\label{defwt}
\tilde{\Omega}=h(\Omega ), \ \ \ 
\tilde{W}^{\text{s}}_{\text{loc}}(0)=h(W^{\text{s}}_{\text{loc}}(0)).
\end{equation}
The following result establishes that $\tilde{W}^{\text{s}}_{\text{loc}}(0)$ is
a $C^1$ local stable manifold for the fixed point $0$ of $T$. The 
smoothness of the stable manifold requires a particular treatment, 
given the fact that $T$ and $h^{-1}$ are not differentiable at the origin. 

\begin{theorem}
\label{stablemt}
The manifold $\tilde{W}^{\text{s}}_{\text{loc}}(0)$ possesses the following properties.
\begin{itemize}
\item[i)]
$\tilde{W}^{\text{s}}_{\text{loc}}(0)$ is positively invariant by $T$.
\item[ii)]
If $(x_0 , z_0) \in \tilde{W}^{\text{s}}_{\text{loc}}(0)$, then the corresponding solution
of (\ref{mappingt}) satisfies 
$$
\lim_{n\rightarrow +\infty}{(x_n , z_n )}=0.
$$
\item[iii)]
All solution of (\ref{mappingt}) such that $(x_n , z_n )\in \tilde{\Omega}$ for all $n\geq 0$
satisfies $(x_n , z_n) \in \tilde{W}^{\text{s}}_{\text{loc}}(0)$ for all $n\geq 0$.
\item[iv)]
There exist $\varepsilon_1 >0$ and a function 
$\gamma \in C^1((-\varepsilon_1,\varepsilon_1),\mathbb{R})$ such that
\begin{equation}
\label{graphewst}
\tilde{W}^{\text{s}}_{\text{loc}}(0) =
\{ \, (x , \gamma (x ))\in \mathbb{R}^2 , |x | < \varepsilon_1 \, \} \subset \tilde\Omega.
\end{equation} 
The function $\gamma$ is odd, increasing and satisfies
$\gamma (x)=P_\alpha (x)+o(|x|^\alpha )$ when $x\rightarrow 0$.
\end{itemize}
\end{theorem}
\begin{proof}
Properties {\em i), ii), iii)} are direct consequences of theorem \ref{stablemm}
and lemma \ref{conju}. Let us prove property {\em iv)}. Using 
(\ref{defwt}) and (\ref{eqfoncpsi1}), one 
obtains after elementary computations
$$
\tilde{W}^{\text{s}}_{\text{loc}}(0) = 
\{ \, ((P_\alpha \circ \varphi )(u_0) , (P_\alpha \circ \varphi )(\psi_1 (u_0 )))
\in \mathbb{R}^2 , |u_0 | < \varepsilon \, \}
$$
where $\varphi = \mbox{Id}+\psi_1$ is odd and
defines a $C^1$-diffeomorphism from $(-\varepsilon , \varepsilon )$ onto its image
since $\varphi^\prime >0$. This yields the parametrization (\ref{graphewst}),
where $x=(P_\alpha \circ \varphi )(u_0)$, $\varepsilon_1 = {\varphi(\varepsilon)}^\alpha$
and $\gamma = P_\alpha \circ \varphi \circ \psi_1 \circ \varphi^{-1} \circ P_{1/\alpha}$.
The function $\gamma$ is odd, increasing and belongs to $C^1 ((0 ,\varepsilon_1),\mathbb{R})$.
Using (\ref{estpsi1}), one obtains 
$$
(\psi_1 \circ \varphi^{-1} \circ P_{1/\alpha})(x)=
x +o(x ) \mbox{ when } 
x\rightarrow 0^+ ,
$$
so that the non-differentiability of $P_{1/\alpha}$ at $x=0$ is compensated by $\psi_1$.
Consequently, one finds
$$
\gamma (x)=x^\alpha +o(x^\alpha ) \mbox{ when } 
x\rightarrow 0^+ ,
$$
which implies $\gamma \in C^1 ((-\varepsilon_1 ,\varepsilon_1),\mathbb{R})$.
{\hfill $\Box$}\end{proof}

Due to the reversibility of $T$ under $R_1$, the fixed point of $T$ at the origin also
admits an unstable manifold, which is simply
\begin{equation}
\label{graphewut}
\tilde{W}^{\text{u}}_{\text{loc}}(0) =R_1 \tilde{W}^{\text{s}}_{\text{loc}}(0) =
\{ \, ( \gamma (z ), z)\in \mathbb{R}^2 , |z | < \varepsilon_1 \, \} ,
\end{equation} 
as stated in the following theorem.

\begin{theorem}
\label{unstablemt}
The manifold $\tilde{W}^{\text{u}}_{\text{loc}}(0)$ possesses the following properties.
\begin{itemize}
\item[i)]
$\tilde{W}^{\text{u}}_{\text{loc}}(0)$ is negatively invariant by $T$.
\item[ii)]
If $(x_0 , z_0) \in \tilde{W}^{\text{u}}_{\text{loc}}(0)$, then the corresponding solution
of (\ref{mappingt}) satisfies 
$$
\lim_{n\rightarrow -\infty}{(x_n , z_n )}=0.
$$
\item[iii)]
All solution of (\ref{mappingt}) such that $(x_n , z_n )\in R_1 \tilde{\Omega}$ for all $n\leq 0$
satisfies $(x_n , z_n) \in \tilde{W}^{\text{u}}_{\text{loc}}(0)$ for all $n\leq 0$.
\end{itemize}
\end{theorem}

As a consequence of theorem \ref{stablemt}, 
\begin{equation}
\label{defstablemt}
\tilde{W}^{\text{s}}(0)=\cup_{n\geq 0}{T^{-n}(\tilde{W}^{\text{s}}_{\text{loc}}(0))}
\end{equation}
defines the (global) stable manifold of the origin, i.e. 
the set of initial conditions $(x_0 , z_0)$ such that $\lim_{n\rightarrow +\infty}{(x_n , z_n )}=0$.
In the same way,
\begin{equation}
\label{defunstablemt}
\tilde{W}^{\text{u}}(0)=\cup_{n\geq 0}{T^{n}(\tilde{W}^{\text{u}}_{\text{loc}}(0))}
\end{equation}
defines the (global) unstable manifold of the origin, i.e. 
the set of initial conditions $(x_0 , z_0)$ such that $\lim_{n\rightarrow -\infty}{(x_n , z_n )}=0$.

In the same way, the stable and unstable manifolds of the origin for the map $F$
are respectively 
\begin{equation}
\label{defmanf}
\mathcal{W}^{\text{s}}(0)=J\, \tilde{W}^{\text{s}}(0), \ \ \
\mathcal{W}^{\text{u}}(0)=J\, \tilde{W}^{\text{u}}(0)=\mathcal{R}_1 \mathcal{W}^{\text{s}}(0).
\end{equation} 
In addition, we shall later refer to the local stable and unstable manifolds defined by
\begin{equation}
\label{defmanfloc}
\mathcal{W}^{\text{s}}_{\text{loc}}(0)=J\, \tilde{W}^{\text{s}}_{\text{loc}}(0), \ \ \
\mathcal{W}^{\text{u}}_{\text{loc}}(0)=J\, \tilde{W}^{\text{u}}_{\text{loc}}(0)
=\mathcal{R}_1 \mathcal{W}^{\text{s}}_{\text{loc}}(0).
\end{equation} 
We have more explicitly
\begin{equation}
\label{smanfloc}
\mathcal{W}^{\text{s}}_{\text{loc}}(0)=
\{ \, ( \gamma (x )+x, \gamma(x))\in \mathbb{R}^2 , |x| < \varepsilon_1 \, \} ,
\end{equation} 
\begin{equation}
\label{umanfloc}
\mathcal{W}^{\text{u}}_{\text{loc}}(0)=
\{ \, ( \gamma (z )+z, z)\in \mathbb{R}^2 , |z | < \varepsilon_1 \, \} .
\end{equation} 

\section{\label{homsol} Homoclinic solutions of the stationary DpS equation}

This section is devoted to the proof of theorem \ref{mainthm}. 
For this purpose, we use the results of section \ref{sum}
combined with a reformulation of the results of \cite{qin} adapted to our context.

\subsection{Homoclinic solutions of generalized DNLS equations}

We consider system (\ref{dnlsstat2}) where
$q\, : \, \mathbb{R} \rightarrow \mathbb{R}$ is an odd function. 
In the sequel we assume $q \in W^{1,1}_{\text{loc}}(\mathbb{R})$, where
$W^{1,1}_{\text{loc}}$ refers to a classical Sobolev space \cite{eg}.
Introducing $z_n = x_{n+1}$, one can reformulate (\ref{dnlsstat2})
as a two-dimensional reversible mapping given by (\ref{mappingt}).
The following existence theorem for homoclinic solutions is 
essentially proved in \cite{qin}, but we shall provide
the whole proof for completeness.

\begin{theorem}
\label{homocl}
Assume the function
$q \in W^{1,1}_{\text{loc}}(\mathbb{R})$ is odd, and 
there exists $x^\ast >0$ such that $q(x^\ast )=0$ and
\begin{equation}
\label{condsigne}
q>0 \mbox{ on } (0,x^\ast ), \ \ \
q^\prime <0 \mbox{ on } (x^\ast , +\infty ).
\end{equation}
Moreover, assume the existence of a curve $\Gamma^{\text{u}}_0$ 
negatively invariant by $T$ and taking the form
\begin{equation}
\label{eqga0u}
\Gamma_0^{\text{u}}=\{ \, (x,z)\in  \mathbb{R}^2 , x=\gamma_{\text{u}}(z), \ 0<z<\epsilon_1   \, \} ,
\end{equation}
where $\epsilon_1 \in (0,x^\ast )$, $\gamma_{\text{u}} \in C^1([0 , \epsilon_1],\mathbb{R})$
and $0<\gamma_{\text{u}}(z)<z$ for all $z\in (0,\epsilon_1 ]$. 

Under the above conditions,
there exist symmetric homoclinic solutions of (\ref{mappingt}) taking the form
$X_n = \pm (x_n^i , x_{n+1}^i)^T$ ($i=1,2$), with 
$\lim_{n\rightarrow \pm\infty}x_n^i=0$, 
$x_n^i > x_{n-1}^i >0$ for all $n\leq 0$, $x_0^i > x^\ast$ and
$$
x_n^1 = x_{-n+1}^1, \ \ \  
x_n^2 =x_{-n}^2, \ \ \  
\mbox{ for all } n \in \mathbb{Z}.
$$
\end{theorem} 

\begin{Remark}
Under the above assumptions,
$q$ has exactly three real zeros, at $x=\pm x^\ast$ and $x=0$,
and $T$ has three fixed points $X=\pm (x^\ast, x^\ast)^T$ and $X=0$.
\end{Remark}

\begin{Remark}
Due to the reversibility of (\ref{mapt}), the assumption made on $\Gamma^{\text{u}}_0$
in theorem \ref{homocl} is equivalent to the existence of a curve
$$
R_1 \Gamma_0^{\text{u}}=\{ \, (x,z)\in  \mathbb{R}^2 , z=\gamma_{\text{u}}(x), \ 0<x<\epsilon_1   \, \} 
$$
positively invariant by $T$.
\end{Remark}

\begin{Remark}
In the work \cite{qin}, the authors have considered 
the smooth situation when condition (\ref{cond1}) is satisfied. In that case,
the assumptions made on the curve 
$\Gamma^{\text{u}}_0$ are automatically satisfied.
Indeed, condition (\ref{cond1}) implies that the origin is a hyperbolic
fixed point of $T$, since $DT(0)$ has a pair of
real eigenvalues $\Lambda , \Lambda^{-1}$ given by (\ref{vp}).
This yields the existence of  
a $C^1$ local unstable manifold of the origin, $W^{\text{u}}_{\text{loc}}(0)=-\Gamma^{\text{u}}_0 \cup \Gamma^{\text{u}}_0$,
with $\Gamma^{\text{u}}_0$ taking the form (\ref{eqga0u})
(at the origin, $W^{\text{u}}_{\text{loc}}(0)$ is tangent to the unstable subspace corresponding
to the line $x = \Lambda^{-1}\, z$).
\end{Remark}

\begin{proof}
The proof given below can be more easily followed using figure \ref{figproof}.
Since $q \in W^{1,1}_{\text{loc}}(\mathbb{R})$ and $\Gamma^{\text{u}}_0$ is negatively invariant by $T$,
$\Gamma^{\text{u}}_n = T^n (\Gamma^{\text{u}}_0) \subset \Gamma^{\text{u}}_{n+1}$
defines an increasing sequence of continuous rectifiable curves. 
Let us consider $\Gamma^{\text{u}} = \cup_{n\geq 0}\Gamma^{\text{u}}_n$.
Since $\Gamma^{\text{u}}_0$ lies in the sector
$$
\Sigma =\{ \, (x,z)\in  (0,+\infty )^2 , \, z>x\, \} ,
$$
and $T^{-1}$ maps the $z$-axis outside $\Sigma$
(onto the $x$-axis), 
the curve $\Gamma^{\text{u}}$ cannot exit $\Sigma$ by crossing the $z$-axis.
In what follows we show that
$\Gamma^{\text{u}}$ exits $\Sigma$ crossing the
line $\mathcal{D}_1=\mbox{Fix}(R_1 )= \mbox{Span}((1,1)^T)$. 

\vspace{1ex}

For this purpose, we first observe that for all $X=(x,z)^T$, one has
$$\mbox{dist}(X , \mathcal{D}_1)=(z-x)/\sqrt{2}$$ 
and thus
\begin{equation}
\label{distance}
\mbox{dist}(T(X) , \mathcal{D}_1) =\frac{q(z)}{\sqrt{2}} + \mbox{dist}(X , \mathcal{D}_1).
\end{equation}
Let us denote by $\Omega_0$ the triangular domain
$$
\Omega_{0}=\{ \, (x,z)\in  (0,+\infty ) \times (0,x^\ast ), \, z>x\, \} 
$$
containing $\Gamma^{\text{u}}_0$.
For all $X=(x,z)^T \in \Omega_0$, $X^\prime =T(X)=(x^\prime ,z^\prime )^T$ satisfies
\begin{equation}
\label{ineg1}
x^\prime =z, \ \ \
x^\ast > x^\prime > x,
\end{equation}
\begin{equation}
\label{ineg2}
\mbox{dist}(X^\prime , \mathcal{D}_1) 
> \mbox{dist}(X , \mathcal{D}_1),
\end{equation}
where we used (\ref{distance}) and the fact that 
$q>0$ on $(0,x^\ast )$.  
By using inequalities (\ref{ineg1})-(\ref{ineg2}), one can rule out the two following
situations (for details see \cite{qin}, theorem 2.3),
\begin{itemize}
\item[i)] 
$\Gamma^{\text{u}}_n \subset \Omega_0$ for all $n\geq 0$,
\item[ii)] 
$\Gamma^{\text{u}}$ exits $\Omega_0$ by intersecting $\mathcal{D}_1$.
\end{itemize}
Consequently, $\Gamma^{\text{u}}$ exits 
$\Omega_0$ by intersecting the segment $\mathcal{S}_0=(0,x^\ast )\times \{ x^\ast \}$ at some point
$X_0^\ast = (x_0 , x^\ast )^T$, thereby entering the domain 
$$
\Omega_1=\{ \, (x,z)\in  (0,+\infty ) \times  ( x^\ast , +\infty ), \, z>x\, \} .
$$
Defining $X_n = T^n (X_0^\ast )=(x_n , z_n)^T = (x_n , x_{n+1})^T$, one can check that 
$X_1 = (x^\ast , z_1 ) \in \Omega_1$. If $X_n \in \Omega_1$ for
$1\leq n \leq n_0$, we find by induction
for all $n =2,\ldots , n_0 $,
\begin{equation}
\label{ineg4}
z_n > z_1 > x^\ast , \ \ \ 
\mbox{dist}(X_{n+1} , \mathcal{D}_1) 
< \frac{q(z_1)}{\sqrt{2}} + \mbox{dist}(X_n , \mathcal{D}_1),
\end{equation}
thanks to equality (\ref{distance}) with
$q^\prime<0$ on $(x^\ast ,+\infty)$. Since $q(z_1)<0$, (\ref{ineg4}) 
provides a nonzero minimal decrease of $\mbox{dist}(X_n , \mathcal{D}_1)$ at each step,
hence $X_n$ must leave $\Omega_1$ for large enough $n$.
Consequently,
$\Gamma^{\text{u}}$ connects two iterates
$X_{n_0}\in ([x^\ast , +\infty ) \times \mathbb{R}) \cap \Omega_1 $ 
and $X_{n_0 +1}\in ([x^\ast , +\infty ) \times \mathbb{R}) \setminus \Omega_1$,
hence $\Gamma^{\text{u}}$ intersects the boundary $\mathcal{D}_1 \cap (x^\ast,+\infty)^2$
of $\Omega_1$. Consequently, 
we have shown that
$\Gamma^{\text{u}}$ exits $\Sigma$ by crossing the
line $\mathcal{D}_1=\mbox{Fix}(R_1 )$.

\vspace{1ex}

Now let us show that for all initial condition $X_k \in \Gamma^{\text{u}}$, the corresponding solution 
$X_n = (x_n , z_n)^T=(x_n , x_{n+1})^T$ of (\ref{mappingt}) satisfies $\lim_{n\rightarrow -\infty}X_n=0$.
Let $X_k \in \Gamma^{\text{u}}_k$ and $X_0 = T^{-k}(X_k)\in \Gamma^{\text{u}}_0$.
We have $x_1 = z_0 \in (0, \epsilon_1 ]$ and
$x_{-n-1}=\gamma_{\text{u}}(x_{-n})\in (0,x_{-n})$ for all $n \geq -1$, 
since $\Gamma^{\text{u}}_0$ is negatively invariant by $T$. Consequently,
there exists $\ell\in [0,\epsilon_1 )$ such that $\lim_{n\rightarrow -\infty}{x_n}=\ell$.
It follows that $\ell =0$, since $0$ is
the only fixed point of $\gamma_{\text{u}}$ in $[0,\epsilon_1 )$.

\vspace{1ex}

Let us consider an intersection $X_0$ of $\Gamma^{\text{u}}$ and $\mathcal{D}_1$
and denote by $X_n =(x_n^1 , x_{n+1}^1)^T$ the corresponding solution
of (\ref{mappingt}). The above result implies that $\lim_{n\rightarrow -\infty}X_n=0$.
Moreover, since $(R_1 X_{-n})_{n\in \mathbb{Z}}$ and $(X_n)_{n\in \mathbb{Z}}$
define two solutions of (\ref{mappingt}) equal at $n=0$, both solutions
coincide for all $n\in \mathbb{Z}$. We have then $R_1 X_{-n}=X_n$, i.e. 
$x_n^1 = x_{-n+1}^1$ for all $n \in \mathbb{Z}$.
This implies $\lim_{n\rightarrow +\infty}X_n=0$, i.e. $(X_n)_{n\in \mathbb{Z}}$
is a solution homoclinic to $0$.
Now we consider the intersection
$X_0^1$ of $\Gamma^{\text{u}}$ and $\mathcal{D}_1 \cap (x^\ast ,+\infty)^2$
which is the closest to the origin along $\Gamma^{\text{u}}$
(i.e. the part of $\Gamma^{\text{u}}$ joining $0$ and $X_0^1$ has minimal arclength). 
Since the arc joining $0$ and $X_0^1$ does not exit $\Sigma$, we have $x_n^1 >0$ and
$x_n^1 > x_{n-1}^1$ for all $n\leq 0$.

\vspace{1ex}

There remains to check that $\Gamma^{\text{u}}$ intersects
$$
\mathcal{D}_2 = \mbox{Fix}(R_2)=\{\, (x,z)\in \mathbb{R}^2, \, x=z+\frac{q(z)}{2}   \, \} .
$$
Since $q<0$ on $(x^\ast ,+\infty)$, the curve $\mathcal{D}_2$ divides
${\Omega_1}$ into two connected components. The boundary of the first
component contains the segment $\mathcal{S}_0$, and the boundary of the second 
the half line $\mathcal{D}_1 \cap (x^\ast ,+\infty)^2$.
Since $\Gamma^{\text{u}}$ intersects $\mathcal{S}_0$ at $X_0^\ast$
and $\mathcal{D}_1 \cap (x^\ast ,+\infty)^2$ at $X_0^1$, there exists
an intersection $X_{-1}^2=(x_{-1}^2,x_{0}^2)^T$ between $\Gamma^{\text{u}}$ and $\mathcal{D}_2$
(with $x_{0}^2 > x^\ast$).
As above, since $R_2 X_{-1}^2 = X_{-1}^2$, the solution of (\ref{mappingt}) given by
$X_n =T^{n+1}\, (X_{-1}^2 )= (x_n^2 , x_{n+1}^2)^T$ 
satisfies $R_2 X_{-n-2}=X_n$, i.e. $x_n^2=x_{-n}^2$
for all $n \in \mathbb{Z}$. 
Consequently, $(X_n)_{n\in \mathbb{Z}}$
is a solution homoclinic to $0$.
Moreover,
since the arc joining $0$ and $X_{-1}^2$ does not exit $\Sigma$, we have $x_n^2 >0$ and
$x_n^2 > x_{n-1}^2$ for all $n\leq 0$.
{\hfill $\Box$}\end{proof}

\begin{figure}[h]
\begin{center}
\includegraphics[scale=0.35]{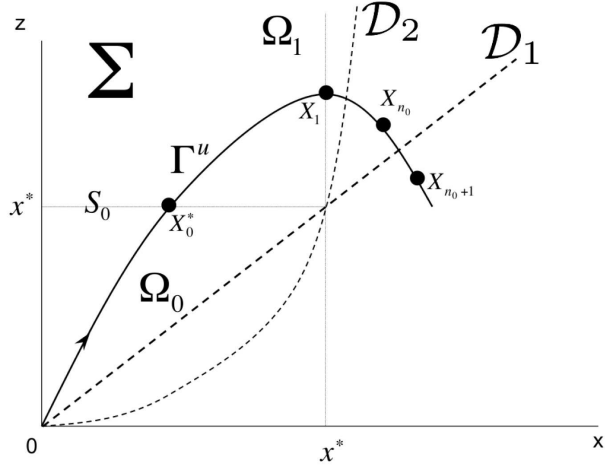}
\end{center}
\caption{\label{figproof}
Sketch of the different subsets of the plane used in
the proof of theorem \ref{homocl}.}
\end{figure}

\subsection{Application to the DpS equation}

Theorem \ref{homocl} can be applied to system (\ref{dnlsstat2})
with nonlinearity $q$ defined by (\ref{defq}), which is equivalent to the
stationary DpS equation. Indeed, $q^\prime$ is locally integrable and
the assumptions made on $q$
in theorem \ref{homocl} are satisfied with $x^\ast = 4^{\frac{\alpha}{1-\alpha}}$.
Moreover, the existence of the invariant curve $\Gamma^{\text{u}}_0$ is
established in theorem \ref{unstablemt}, where 
$\tilde{W}^{\text{u}}_{\text{loc}}(0)=-\Gamma^{\text{u}}_0 \cup \Gamma^{\text{u}}_0$.
This leads to the following result.

\begin{theorem}
\label{homocx}
There exist solutions $x_n^i$ ($i=1,2$)
of (\ref{dnlsstat2})-(\ref{defq}) satisfying 
$$
\lim_{n\rightarrow \pm\infty}x_n^i=0,
$$ 
$x_n^i > x_{n-1}^i >0$ for all $n\leq 0$, $x_0^i > 4^{\frac{\alpha}{1-\alpha}}$ and
$$
x_n^1 = x_{-n+1}^1, \ \ \  
x_n^2 =x_{-n}^2, \ \ \  
\mbox{ for all } n \in \mathbb{Z}.
$$
\end{theorem}

Returning to the DpS equation in its original form (\ref{dpsstat}), we get the following result, 
which establishes theorem \ref{mainthm}
in conjunction with lemma \ref{nonexist} proved in the appendix.

\begin{theorem}
\label{homocsolu}
The stationary DpS equation (\ref{dpsstat}) with $s=1$ admits
solutions $\tilde{a}_n^1$, $a_n^2$ satisfying the properties
$$
\lim_{n\rightarrow \pm\infty}a_n=0, \ \ \
(-1)^n\, a_{n} >0, \ \ \
|a_n | > |a_{n-1} | \ \forall n\leq 0,
$$ 
and
\begin{equation}
\label{syma}
\tilde{a}_n^1 =\tilde{a}_{-n}^1,  \ \ \ 
a_n^2 = -a_{-n+1}^2, 
 \ \ \  \mbox{ for all } n \in \mathbb{Z}.
\end{equation}
\end{theorem}
\begin{proof}
From the localized solutions $x_n^i$ of theorem \ref{homocx},
let us define 
\begin{equation}
\label{defu1}
\tilde{u}^1_n=x^1_n + x^1_{n+1}, \ \ \ \tilde{a}^1_n=(-1)^n\, \tilde{u}^1_n,
\end{equation}
\begin{equation}
\label{defu2}
u^2_n=x^2_n + x^2_{n-1}, \ \ \ a^2_n=(-1)^n\, u^2_n.
\end{equation}
Since $x_n^1$ and $x_n^2$
are solution of (\ref{dnlsstat2})-(\ref{defq}), we get
$$
\tilde{u}^1_n + \tilde{u}^1_{n-1}=P_{1/\alpha}(x^1_n), \ \ \
u^2_n + u^2_{n+1}=P_{1/\alpha}(x^2_n),
$$
hence
$$
x^1_n=P_{\alpha}(\tilde{u}^1_n + \tilde{u}^1_{n-1}), \ \ \
x^2_n=P_{\alpha}(u^2_n + u^2_{n+1}).
$$
Reporting the above identities in definitions (\ref{defu1}) and (\ref{defu2}),
one finds that $\tilde{u}^1_n , u^2_n$ define solutions of (\ref{dpsstatbis}) for $s=1$
(note that ${u}^1_n:=\tilde{u}^1_{n-1}$ and $x^1_n$ are linked by
equality (\ref{defforvar})), and thus
$\tilde{a}^1_n , a^2_n$ are solutions of (\ref{dpsstat}).
The remaining properties of $\tilde{a}^1_n , a^2_n$ directly follow from 
those of $x_n^1$ and $x_n^2$.
{\hfill $\Box$}\end{proof}

Figures \ref{homocorb} and \ref{homocorb2} 
illustrate the profiles of the homoclinic solutions $\tilde{a}_n^1$, $a_n^2$ 
for different values of $\alpha$. The solutions of the stationary DpS equation (\ref{dpsstat}) with $s=1$
are computed with a Newton-type method (we use the MATLAB function fsolve), for
a finite lattice of $21$ particles with zero boundary conditions. When $\alpha $ converges towards unity, 
the homoclinic solutions become more extended and their amplitude goes to $0$. This phenomenon
will be explained in section \ref{sectapprox2}
by introducing a suitable continuum limit of the stationary DpS equation.

\begin{figure}[h]
\begin{center}
\includegraphics[scale=0.2]{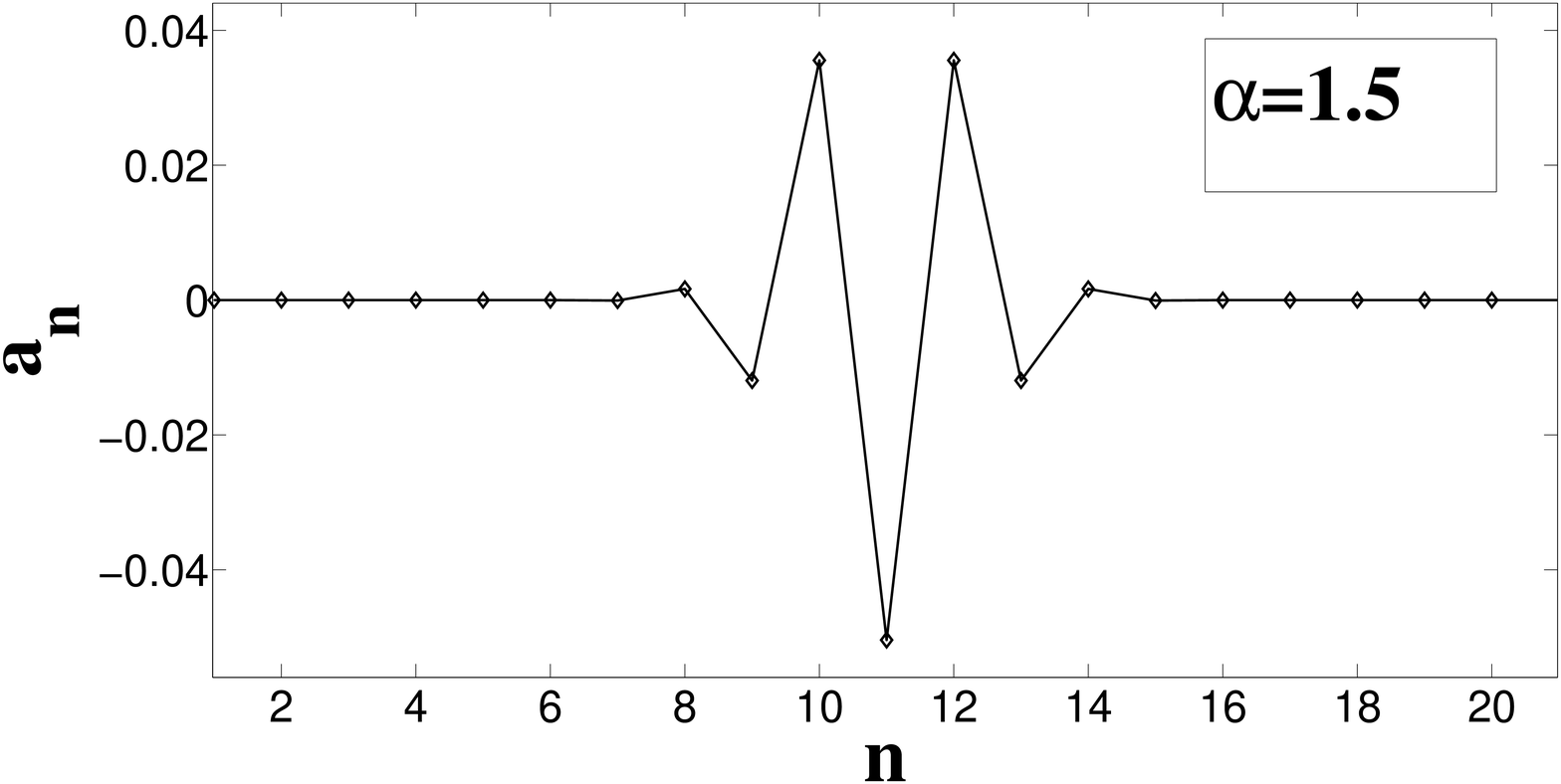}
\includegraphics[scale=0.2]{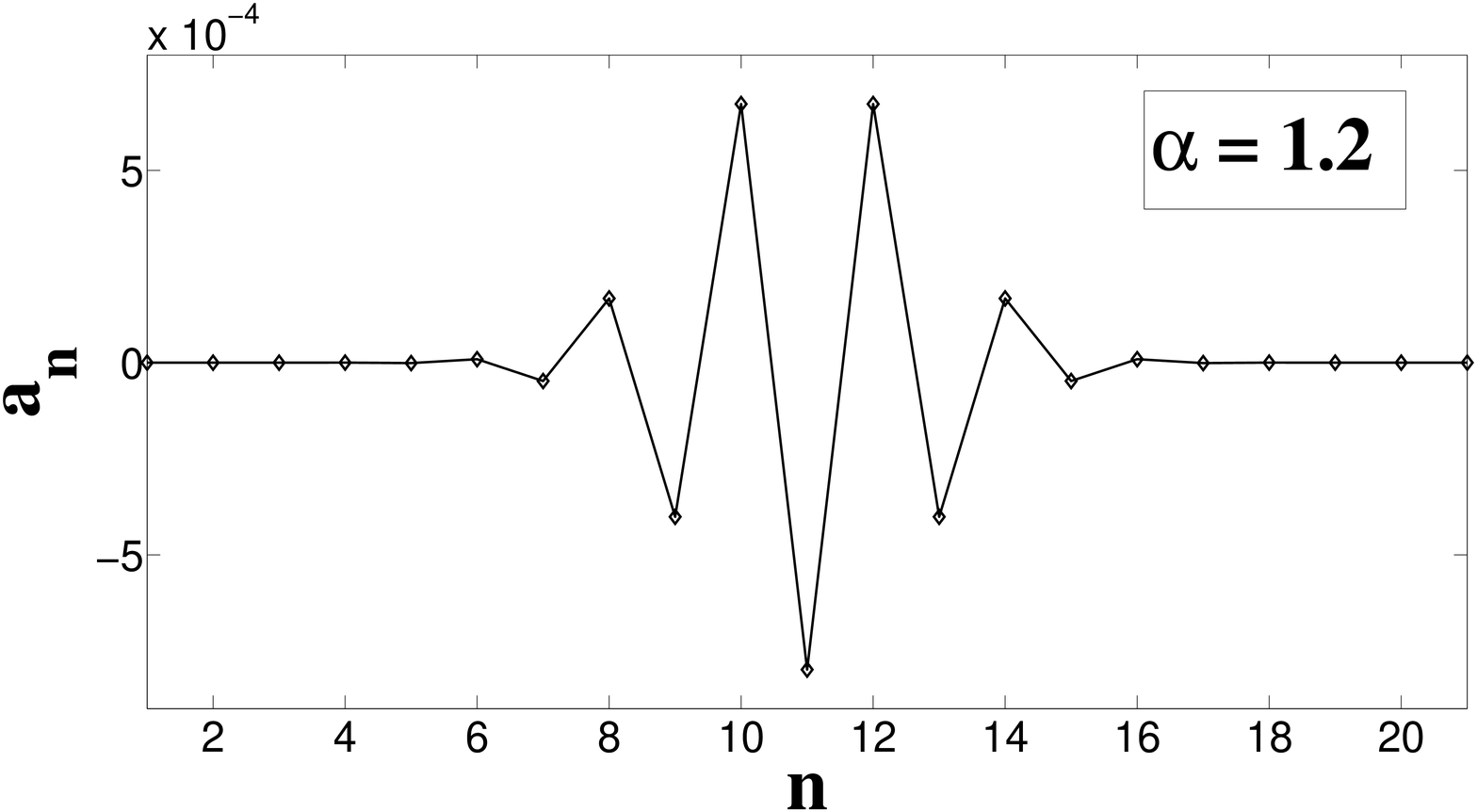}
\end{center}
\caption{\label{homocorb}
Homoclinic solution $\tilde{a}_n^1$
of the stationary DpS equation (\ref{dpsstat}) with $s=1$.
The solution is computed numerically 
for $\alpha =3/2$ (top panel) and $\alpha = 1.2$ (bottom panel).
}
\end{figure}

\begin{figure}[h]
\begin{center}
\includegraphics[scale=0.2]{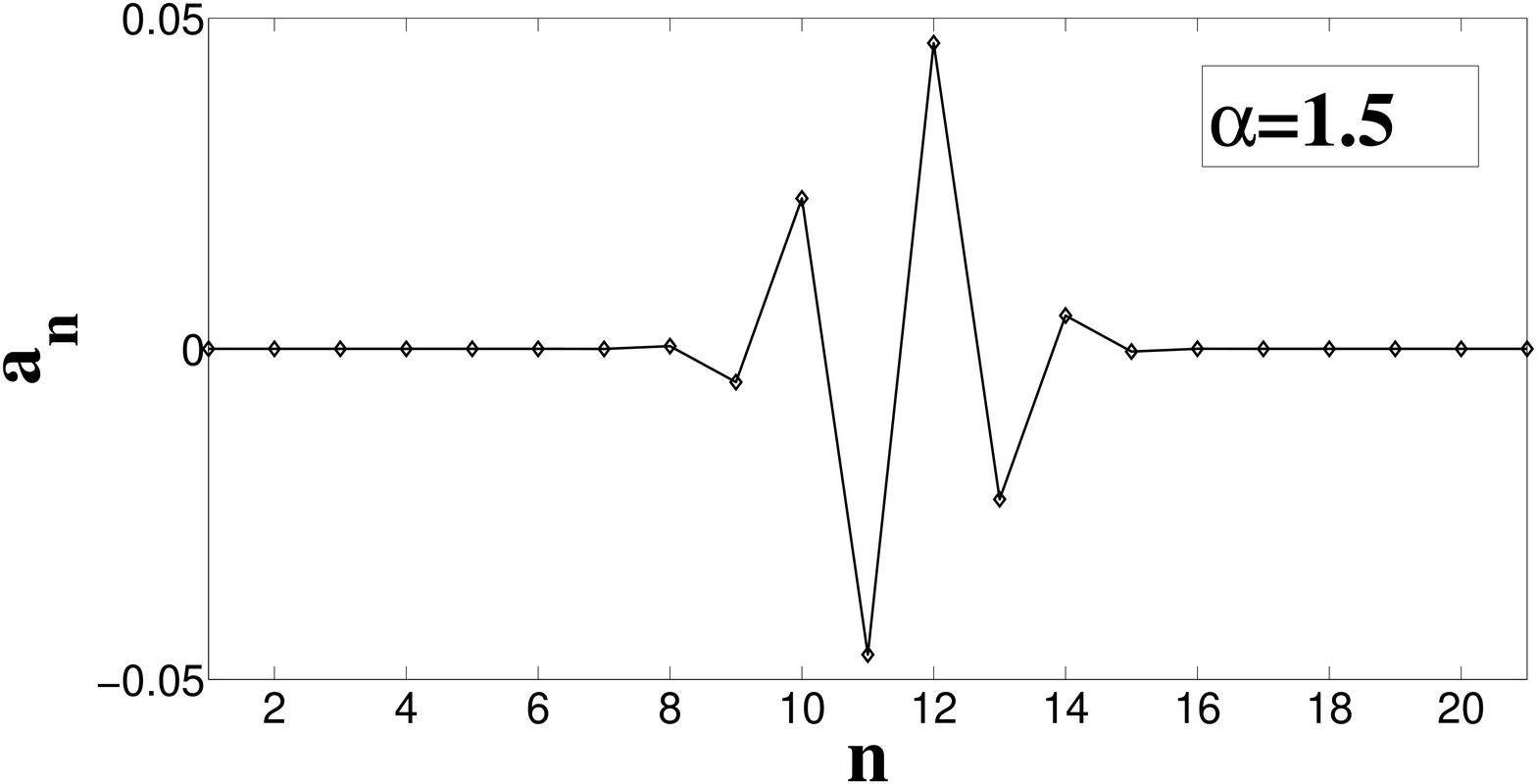}
\includegraphics[scale=0.2]{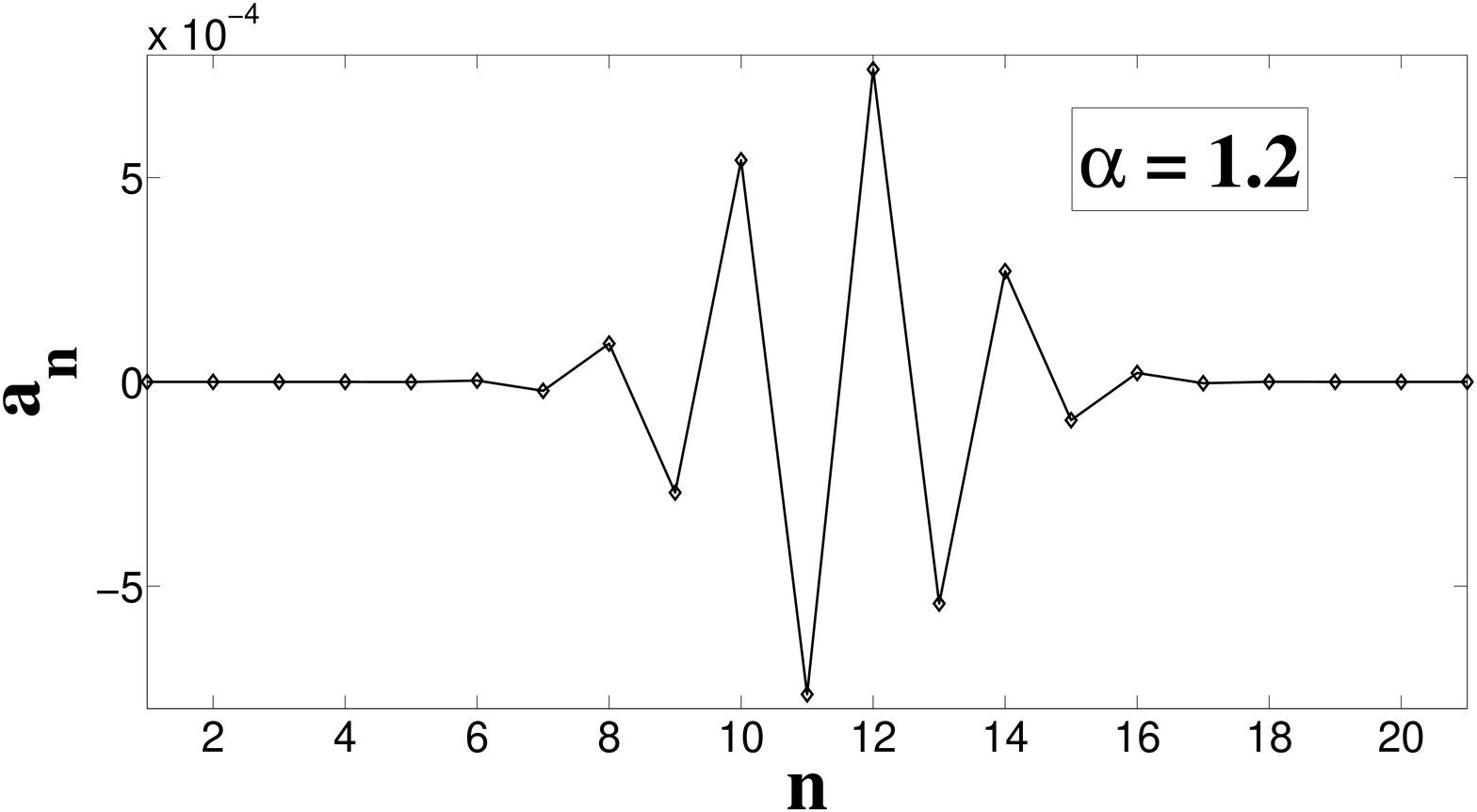}
\end{center}
\caption{\label{homocorb2}
Homoclinic solution $a_n^2$ 
of the stationary DpS equation (\ref{dpsstat}) with $s=1$.
The solution is computed numerically 
for $\alpha =3/2$ (top panel) and $\alpha = 1.2$ (bottom panel).
}
\end{figure}

One can notice that the localized solutions provided by theorem \ref{homocsolu}
are ``staggered", i.e. $a_{n+1}$ and $a_n$ have opposite signs. As shown in the following lemma,
this remains true for all localized solutions $a_n$
of (\ref{dpsstat}) provided $|n|$ is sufficiently large.
Moreover, we prove below that all localized solutions of (\ref{dpsstat})
(in particular $\tilde{a}_n^1$, $a_n^2$) decay at least doubly exponentially at infinity.

\begin{lemma}
\label{homocbehav}
Let $(a_n)_{n\in \mathbb{Z}}$ denote a solution of (\ref{dpsstat}) such that $\lim_{n\rightarrow \pm \infty}{a_n}=0$
and fix $q\in (0,1)$. There exists $n_0 \geq 0$ such that for all $n\in \mathbb{Z}$ with $|n| \geq n_0$,
$a_{n+1}$ and $a_n$ have opposite signs and 
\begin{equation}
\label{decayloc}
|a_n | \leq q^{1+\alpha^{|n|-n_0}}.
\end{equation}
\end{lemma}
\begin{proof}
According to lemma \ref{nonexist} proved in the appendix, nonzero localized
solutions of (\ref{dpsstat}) exist only for $s=1$. 
To prove that localized solutions
are staggered for sufficiently large $n$, let us consider a solution $a_n$ of (\ref{dpsstat}) such that
$\lim_{n\rightarrow +\infty}a_n=0$. One can define a 
solution of (\ref{dpsstatbis}) vanishing as $n\rightarrow +\infty$
through $u_n = (-1)^n a_{n+n_0}$,
with $\sup_{n\geq 0}{|u_n|}$ arbitrarily small provided $n_0$ is sufficiently large.
By theorem \ref{unstm},
it follows that $u_n = \psi_n (u_0 )$ for all $n \geq 1$ provided $n_0$ is large enough.
For $s=1$, $\psi_n (u_0)$ and $u_0$ have the same sign for all $n\geq 1$
when $u_0$ is small enough (due to properties i) and ii) of theorem \ref{unstm}),
hence $(u_n)_{n\geq 0}$ has a sign provided $n_0$ is large enough.
Consequently, $a_{n+1}$ and $a_n$ have opposite signs for
all $n\geq n_0$ if $n_0$ is sufficiently large. In the same way
(using the reflectional symmetry of (\ref{dpsstat})), if $\lim_{n\rightarrow -\infty}a_n=0$ then
$a_{n+1}$ and $a_n$ have opposite signs for
all $n\leq -n_0$ if $n_0 \geq 0$ is sufficiently large.
Inequality (\ref{decayloc})
follows similarly from theorem \ref{unstm},
after elementary computations based on the decay estimate (\ref{bounddec}).
One can fix $C=q^{1-\alpha}$ in (\ref{bounddec}) and observe that $q^{-1}\, |u_0| \leq q$
provided $|n_0|$ is large enough, which yields (\ref{decayloc}).
{\hfill $\Box$}\end{proof}

\section{\label{sectapprox}Analytical approximations of stable and unstable manifolds}

In previous sections, we have 
proved the existence of stable and 
unstable manifolds of the origin for maps equivalent to the
stationary DpS equation (\ref{dpsstat}), and the existence of homoclinic orbits
for $s=1$.
In what follows, we illustrate the shape of these 
manifolds and orbits and provide analytical approximations thereof,
depending on the nonlinearity exponent $\alpha = p-1 >1$.
We shall work with the map $F$ defined in (\ref{mapn}) because
it will be suitable for future extensions to this work taking into account lattice defects,
as briefly discussed in section \ref{sectdiscuss}.

Numerical computation of the exact stable 
manifold is performed in two stages. In the first stage we construct a local stable manifold. To this end we choose $N$ equally distant points $0 < u_{0,1} < u_{0,2} < \cdots < u_{0,N}$ with
$u_{0,N}$ small enough. 
From each point $u_{0,i}$ 
we compute an orbit on the local stable manifold from the corresponding
solution $(u_n)_{n \geq 0}$ of equation (\ref{dpsstatbis2})
with $u_{0} =u_{0,i}$. This problem
is solved with a Newton-type method, for
a finite lattice of size $L$ sufficiently large and the
fixed boundary condition $u_{L} =0$.
This yields a set of orbits $(v_{n,i},y_{n,i})_{n \geq 1}$ on the stable manifold, where
$v_{n,i}=u_{n-1}$ and 
$y_{n,i} =P_{\alpha } (u_{n} +u_{n-1} )$.
In a second stage, to construct the global stable manifold,
we recursively apply the inverse mapping $F^{-1}$
to the set of points $(v_{1,i},y_{1,i})$ 
computed previously on the local stable manifold. The
construction of the unstable manifold is performed in a similar manner.

In addition, 
we shall resort to two different methods to obtain 
analytical approximations of stable and unstable manifolds. 
The first one is based on a leading order approximation of the local stable manifold
and the computation of some backward iterates (section \ref{sectapprox1}), which turns out to be 
efficient when $\alpha$ is far from unity. The second method is based on a continuum limit
obtained when $\alpha$ is close to unity, where one recovers a logarithmic stationary
nonlinear Schr\"odinger equation (section \ref{sectapprox2}). 
These approximations will be computed for the stable manifold, and
their analogues for the unstable manifold can be obtained using
the reversibility symmetry  $\mathcal{R}_1$.
For the approximation of homoclinic orbits we shall restrict to the site-centered 
solution $\tilde{a}_n^1$ described in theorem \ref{homocsolu}, but the bond-centered
solution $a^2_n$ could be approximated similarly.

\subsection{\label{sectapprox1}Method of local approximation and backward iterates}

Let us consider the map $F$ defined in (\ref{mapn}) and 
the stable and unstable manifolds of the origin 
$\mathcal{W}^{\text{u}}(0)$ and $\mathcal{W}^{\text{s}}(0)$ 
defined in section \ref{defsumf}. 
Using parametrization (\ref{smanfloc}) and the fact that
$\gamma (x)=P_\alpha (x)+o(|x|^\alpha )$ when $x\rightarrow 0$,
we derive the following approximations of the
local stable and unstable manifolds
\begin{equation}
\label{smanflocapp}
\mathcal{W}^{\text{s}}_{\text{app}}=
\{ \, (v,y)=( P_\alpha (x )+x, P_\alpha(x))\in \mathbb{R}^2 , x \in \mathbb{R} \, \} ,
\end{equation} 
\begin{equation}
\label{umanflocapp}
\mathcal{W}^{\text{u}}_{\text{app}}=
\{ \, (v,y)=( P_\alpha (z)+z, z)\in \mathbb{R}^2 , z \in \mathbb{R}  \, \}, 
\end{equation} 
which are valid close enough to the origin. 
To improve their validity domain, we shall consider the
backward iterates of the approximate stable manifold
\begin{equation}
\label{smanfback}
\mathcal{W}^{\text{s},(k)}_{\text{app}}=F^{-k}(\mathcal{W}^{\text{s}}_{\text{app}}),
\end{equation} 
where
$$
F^{-1}\, 
\left( \begin{array}{c}
{v}\\
{y}
\end{array} \right) = \left( \begin{array}{c}
{P_{1/\alpha }}\left( {{v} - {y}} \right) - {v}\\
{v} - {y}
\end{array} \right) .
$$
We define by symmetry
$$
\mathcal{W}^{\text{u},(k)}_{\text{app}}=\mathcal{R}_1 \mathcal{W}^{\text{s},(k)}_{\text{app}}.
$$
One can use the parametrization
\begin{equation}
\label{paramk}
\mathcal{W}^{\text{s},(k)}_{\text{app}}=
\{ \, (v,y)=( v^{(k)}(x), y^{(k)}(x)) \in \mathbb{R}^2 , x \in \mathbb{R} \, \} ,
\end{equation} 
where $v^{(0)}=P_\alpha +\mbox{Id}$, $y^{(0)}=P_\alpha$ and
the functions $( v^{(k)}, y^{(k)})$ are defined by induction for $k\geq 1$, with
\begin{equation}
\left( \begin{array}{c}
{v^{(k+1)}}\\
{y^{(k+1)}}
\end{array} \right) = \left( \begin{array}{c}
{P_{1/\alpha }}\left( {{v^{(k)}} - {y^{(k)}}} \right) - {v^{(k)}}\\
{v^{(k)}} - {y^{(k)}}
\end{array} \right) .
\label{eq:3}
\end{equation}
We have in particular
\begin{equation}
\left( \begin{array}{c}
{v^{(1)}(x)}\\
{y^{(1)}(x)}
\end{array} \right) = \left( \begin{array}{c}
{P_{1/\alpha }}(x) - {P_\alpha }\left( x \right) - x\\
x
\end{array} \right) .
\label{eq:1}
\end{equation}
When $\alpha$ is sufficiently far from unity,
one observes  numerically that 
$\mathcal{W}^{\text{s},(1)}_{\text{app}}$
and $\mathcal{W}^{\text{u},(1)}_{\text{app}}$ provide good approximations
of $\mathcal{W}^{\text{s}}(0)$ and $\mathcal{W}^{\text{u}}(0)$ close
to the reversible homoclinics. This result is illustrated by figure \ref{compal3}
for $\alpha=3$. The quality of the approximation decreases when $\alpha$ becomes close to
unity, as shown in figure \ref{compal1_5} for $\alpha=3/2$. However, the precision can be improved
by considering $\mathcal{W}^{\text{s},(k)}_{\text{app}}$
and $\mathcal{W}^{\text{u},(k)}_{\text{app}}$ with larger values of $k$
(at the expense of working with more complex parametrizations of the manifolds).
It seems that the precision can be improved arbitrarily 
(in some fixed neighborhood of $(v,y)=0$) by increasing $k$, but a convergence proof
is not yet available.

\begin{figure}[h]
\psfrag{y_n}[1][Bl]{ $y_n$}
\psfrag{v_n}[1][Bl]{ $v_n$}
\begin{center}
\includegraphics[scale=0.25]{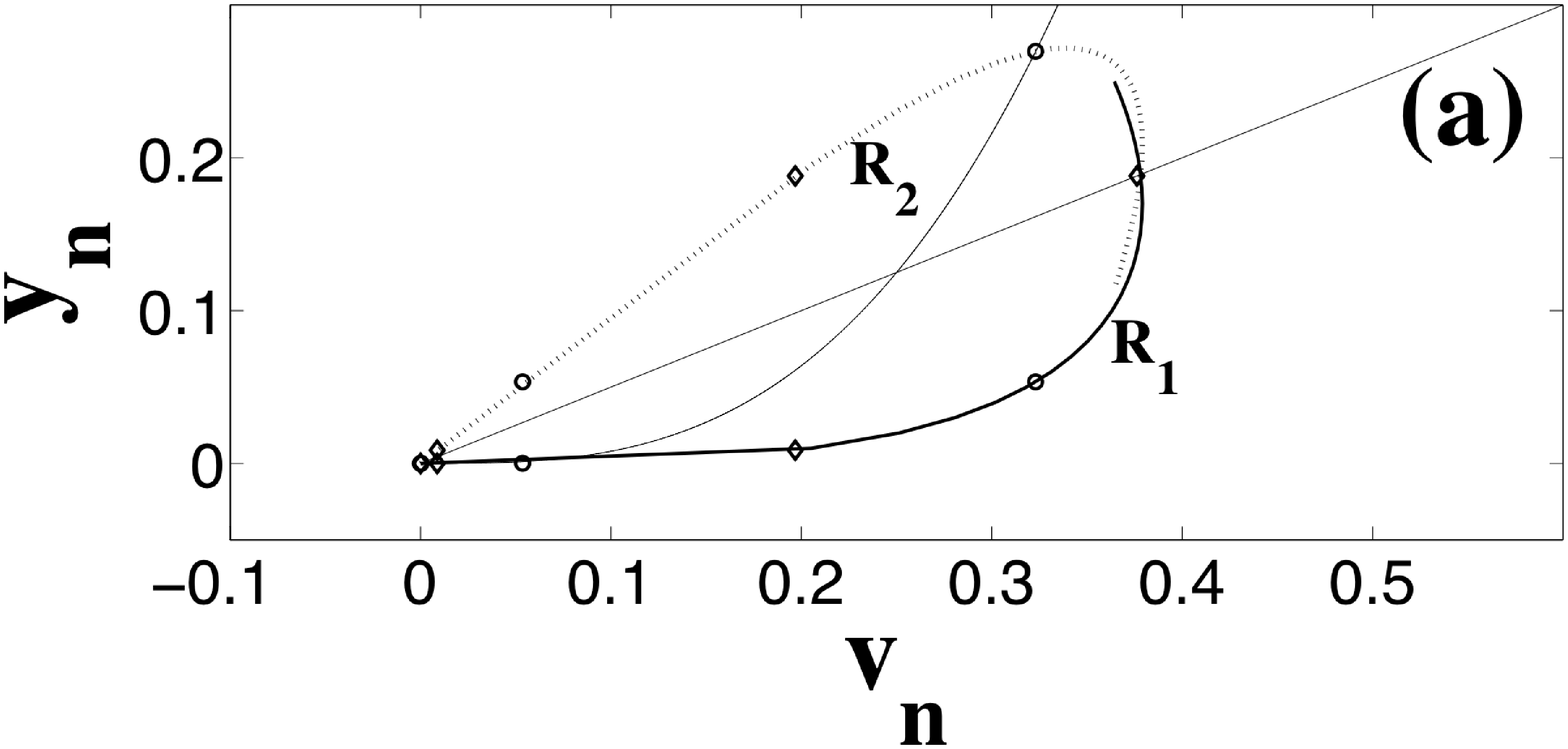}
\includegraphics[scale=0.25]{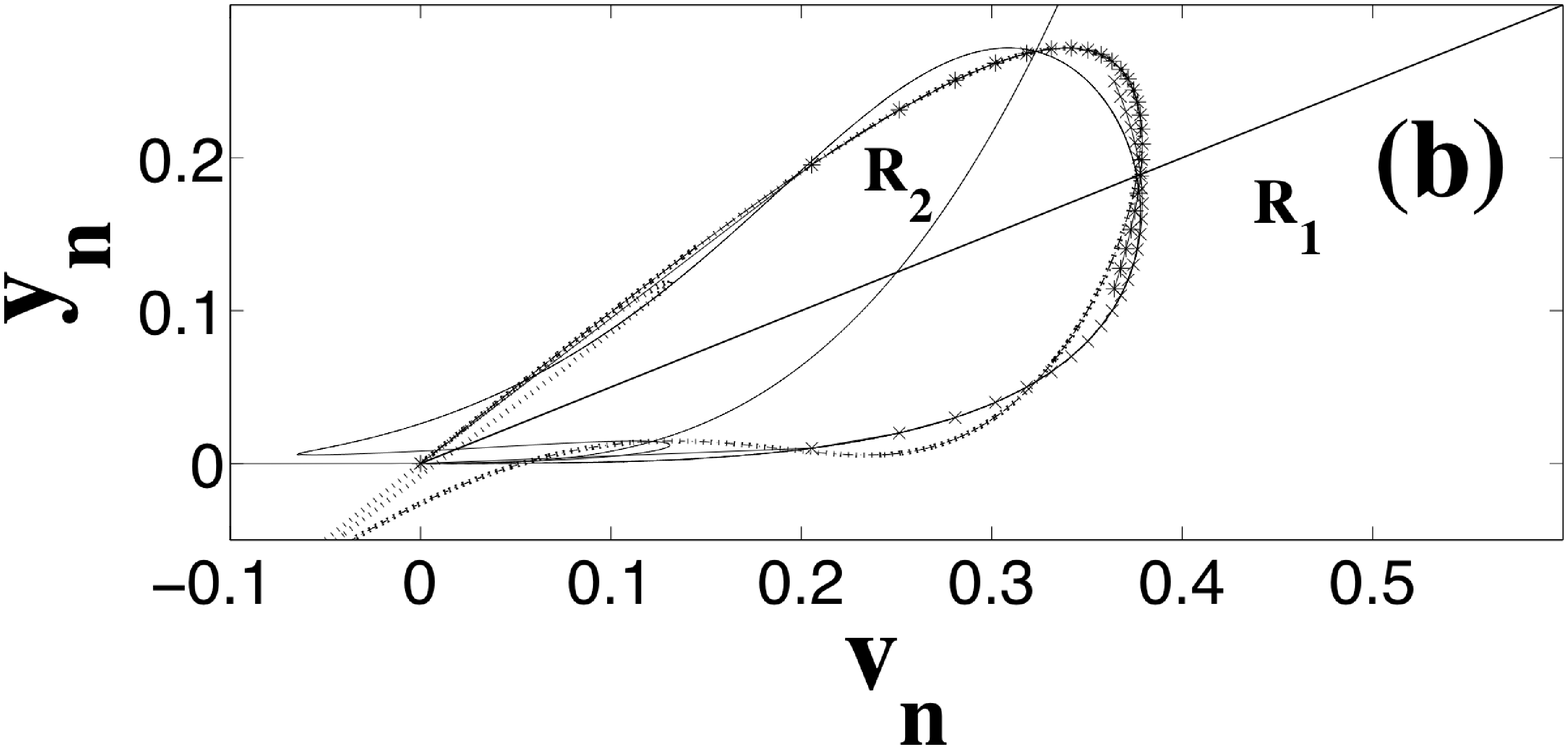}
\end{center}
\caption{\label{compal3}
(a):
approximate stable manifold $\mathcal{W}^{\text{s},(1)}_{\text{app}}$
(bold solid line) and approximate unstable manifold $\mathcal{W}^{\text{u},(1)}_{\text{app}}$
(dashed thin curve) in the case $\alpha=3$ (only a limited part of the curves has been represented). 
Dots correspond to numerically computed homoclinic orbits.
The intersections of $\mbox{Fix}(\mathcal{R}_1)$ and $\mbox{Fix}(\mathcal{R}_2)$
with $\mathcal{W}^{\text{s},(1)}_{\text{app}}$ and $\mathcal{W}^{\text{u},(1)}_{\text{app}}$
are very close to exact homoclinic intersections.
(b):
same as above, except the exact stable (solid curve) and unstable (dashed curve) manifolds
$\mathcal{W}^{\text{s}}(0)$ and $\mathcal{W}^{\text{u}}(0)$ are now represented. 
The approximate stable manifold $\mathcal{W}^{\text{s},(1)}_{\text{app}}$
is denoted with a solid curve with (x) marks and $\mathcal{W}^{\text{u},(1)}_{\text{app}}$
is denoted with a solid curve with ($\ast$) marks. 
The exact and approximate stable and unstable manifolds 
almost coincide before crossing
the line $\mbox{Fix}(\mathcal{R}_1)$.
$\mbox{Fix}(\mathcal{R}_1)$ and $\mbox{Fix}(\mathcal{R}_2)$
are denoted on both (a) and (b) with solid lines marked with the appropriate labels 
$R_{1} $ and $R_{2}$. 
}
\end{figure}

\begin{figure}[h]
\begin{center}
\includegraphics[scale=0.2]{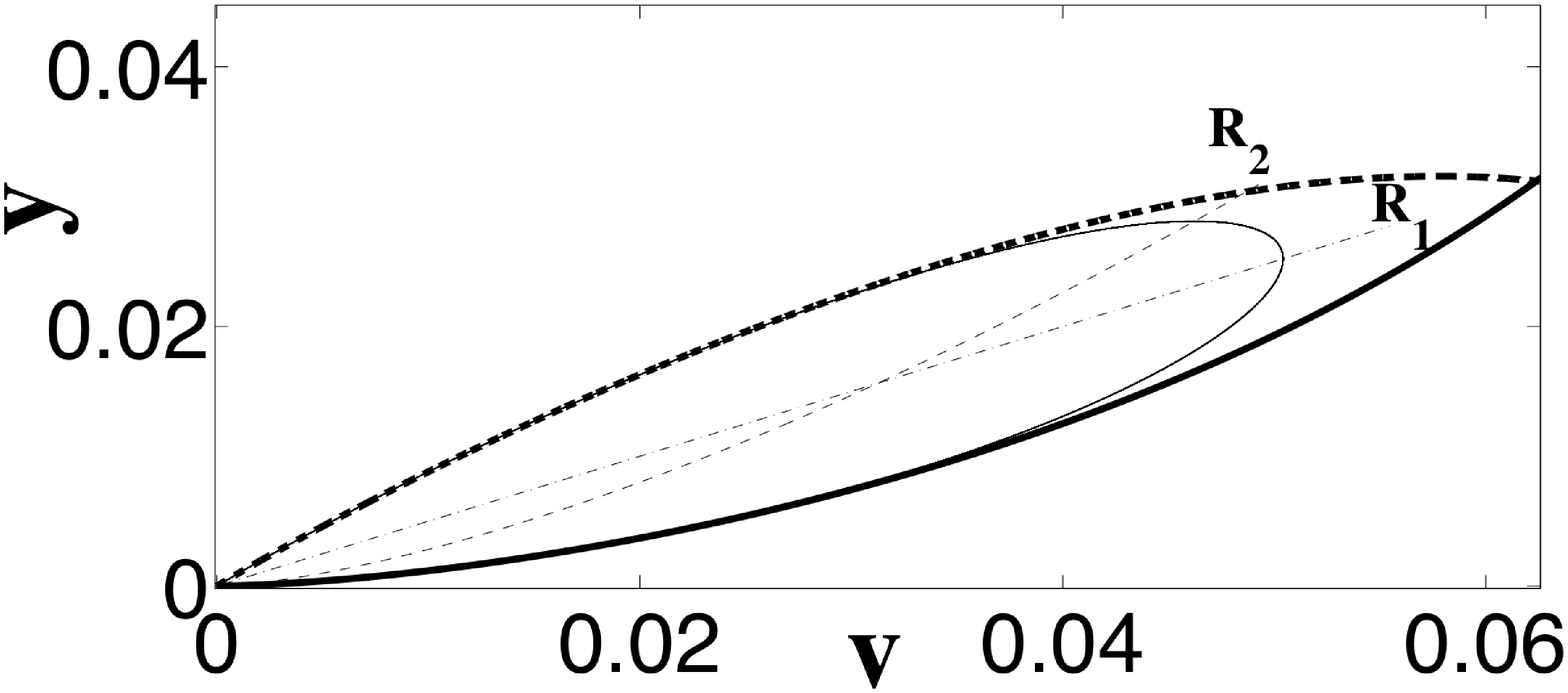}
\includegraphics[scale=0.2]{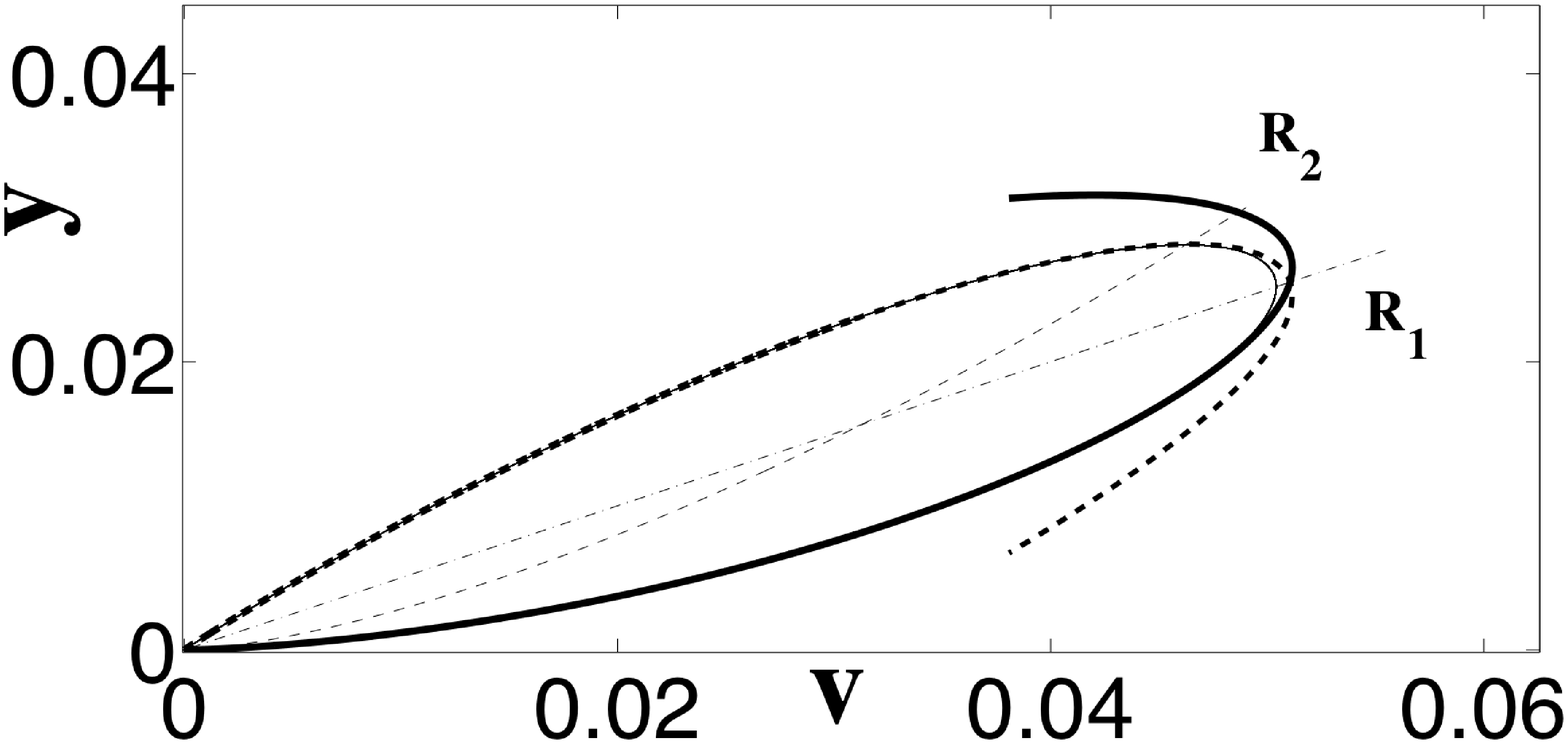}
\hspace*{-3ex}
\includegraphics[scale=0.215]{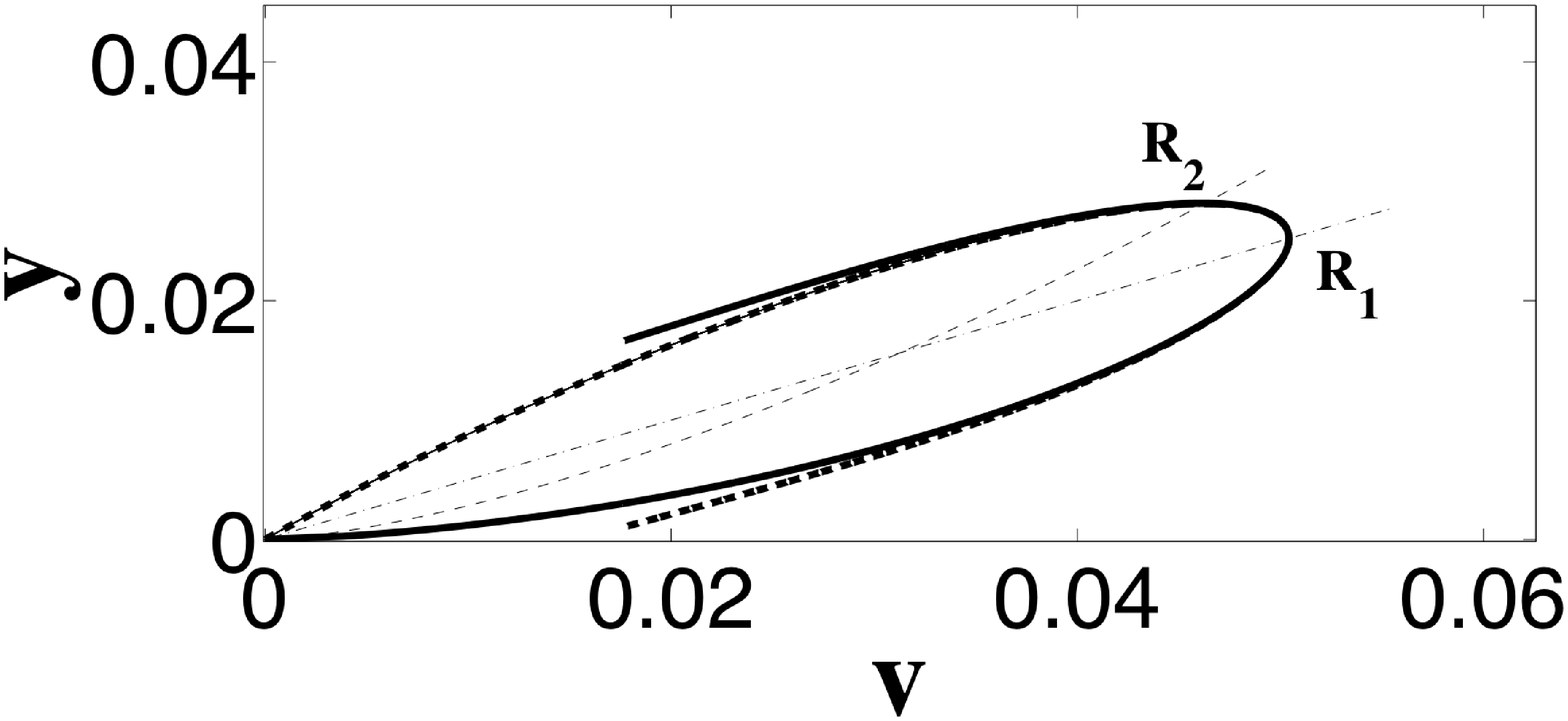}
\end{center}
\caption{\label{compal1_5}
Comparison of the exact invariant manifolds $\mathcal{W}^{\text{s}}(0)$ 
and $\mathcal{W}^{\text{u}}(0)$ (solid lines)
with the approximate ones $\mathcal{W}^{\text{s},(k)}_{\text{app}}$ (bold solid curve)
and $\mathcal{W}^{\text{u},(k)}_{\text{app}}$ (bold dashed curve) 
in the case $\alpha=3/2$. The thin, dotted lines
correspond to $\mbox{Fix}(\mathcal{R}_1)$ and $\mbox{Fix}(\mathcal{R}_2)$.
The top panel corresponds to $k=1$. The
analytic approximation corresponding to $\mathcal{W}^{\text{s},(1)}_{\text{app}}$
and $\mathcal{W}^{\text{u},(1)}_{\text{app}}$ doesn't 
provide a good match with the exact invariant manifolds. To make the 
approximation more precise we perform additional iterations, yielding the improved 
approximations $\mathcal{W}^{\text{s},(2)}_{\text{app}}$, $\mathcal{W}^{\text{u},(2)}_{\text{app}}$
(middle panel), and 
$\mathcal{W}^{\text{s},(3)}_{\text{app}}$, $\mathcal{W}^{\text{u},(3)}_{\text{app}}$
(bottom panel).
}
\end{figure}

As an application, let us use the analytic approximations $\mathcal{W}^{\text{s},(1)}_{\text{app}}$
and $\mathcal{W}^{\text{u},(1)}_{\text{app}}$ to approximate the site-centered
homoclinic solution of the stationary DpS equation. The fixed point
$(v_1,y_1)^T$ of $\mathcal{R}_1$ lying on $\mathcal{W}^{\text{s},(1)}_{\text{app}}$
satisfies
$$
v_1 = 2\, y_1, \ \ \
v_1={P_{1/\alpha }}(y_1) - {P_\alpha }\left( y_1 \right) - y_1,
$$
hence $y_1 >0$ satisfies
\begin{equation}
\label{eqe}
3=y_1^{(1/\alpha) -1}-y_1^{\alpha -1},
\end{equation}
which admits a unique solution $y_1 \in (0,1)$ since the right side of
(\ref{eqe}) is monotone decreasing. 
This solution corresponds to an approximate solution of the 
stationary DpS equation (\ref{dpsstat}) determined by the initial condition
\begin{equation}
\label{icapprox}
a_0=v_1=2\, y_1,
\ \ \
a_1=v_1- y_1^{(1/\alpha)}.
\end{equation} 
Obviously, due to the sensivity of
the map $F$ to initial conditions, the above approximation is only meaningful
for a finite number of sites away from $n=0$.

In order to obtain an approximation of the solution of (\ref{eqe}),
we introduce a new variable
\begin{equation}
\label{eq5}
\eta =y_{1}^{\textstyle{{\alpha -1} \over \alpha }} .
\end{equation}
Introducing (\ref{eq5}) into (\ref{eqe}) yields
\begin{equation}
\label{eq6}
\eta^{\alpha +1}=1-3\eta .
\end{equation}
Using the method of successive iterations to approximate $\eta$, we fix $\eta_{0} =0$ and 
consider the following recurrence relation
\begin{equation}
\label{eq8}
\eta_{i} =\left( {1-\eta_{i-1}^{\alpha +1} } \right)/3 .
\end{equation}
Using $\eta_2$ as an approximation of the solution of (\ref{eq6}) leads to
\begin{equation}
\label{eq11}
y_{1} \approx \eta_2^\frac{\alpha}{\alpha -1}=
3^\frac{\alpha}{1-\alpha }\, \left(1-3^{-\alpha-1} \right)^\frac{\alpha}{\alpha -1}.
\end{equation}
According to the results presented in figure \ref{errorsy},
using $\eta_2$ gives a good approximation of the branch of solutions we 
are looking for, provided $\alpha$ is chosen within 
the range of values where approximation (\ref{eqe}) is valid.

\begin{figure}[h]
\begin{center}
\includegraphics[scale=0.25]{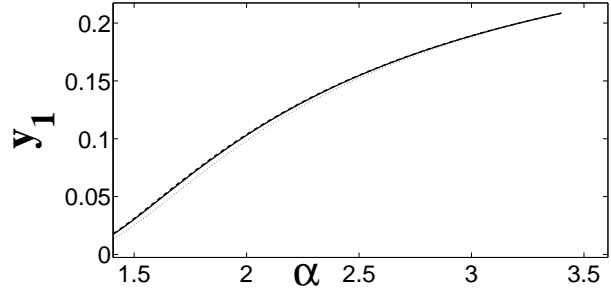}
\end{center}
\caption{\label{errorsy}
Comparison of the exact value of $y_1$ for the site-centered homoclinic orbit (dotted line),
its approximation given by the exact solution of (\ref{eqe}) (dashed line) and the analytical approximation
(\ref{eq11}) (solid line), for different values of $\alpha$. 
The exact site-centered homoclinic orbit is obtained by solving the stationary
DpS equation with a Newton-type method, and the same iterative procedure is employed to solve (\ref{eqe}).
It appears that the solution of (\ref{eqe}) yields
in fact a slightly less precise approximation
compared to analytical approximation (\ref{eq11}).}
\end{figure}

In figure \ref{comphom}, we
compare for different values of $\alpha$
the approximate site-centered homoclinic 
orbit determined by (\ref{eq11}) and (\ref{icapprox}) with the exact 
site-centered homoclinic orbit (computed by a Newton-type method from the stationary DpS equation). 
We plot the amplitude $u_n= (-1)^n a_n$ obtained after the staggering transformation.
An approximation error is computed as
\begin{equation}
\label{eq12}
\mbox{Err}=\frac{\left\| {{\rm {\bf u}}_{\mbox{exact}} -{\rm {\bf 
u}}_{\mbox{approx}} } \right\|_{2} }{\left\| {{\rm {\bf 
u}}_{\mbox{exact}} } \right\|_{2} }
\end{equation}
where ${\rm {\bf u}}=[u_{-2} ,...,u_{2} ]$.
As it comes out from figure \ref{comphom}, analytical approximation 
(\ref{eq11}) is in fairly good agreement with the results of direct
numerical computations for the higher values of $\alpha $, but it 
provides unsatisfactory results for $\alpha $ close to unity 
(e.g. $\alpha =3/2, \,2)$. 

\begin{figure}[h]
\psfrag{Element Index}[1][Bl]{ $n$}
\begin{center}
\hspace*{-4.5ex}
\includegraphics[scale=0.21]{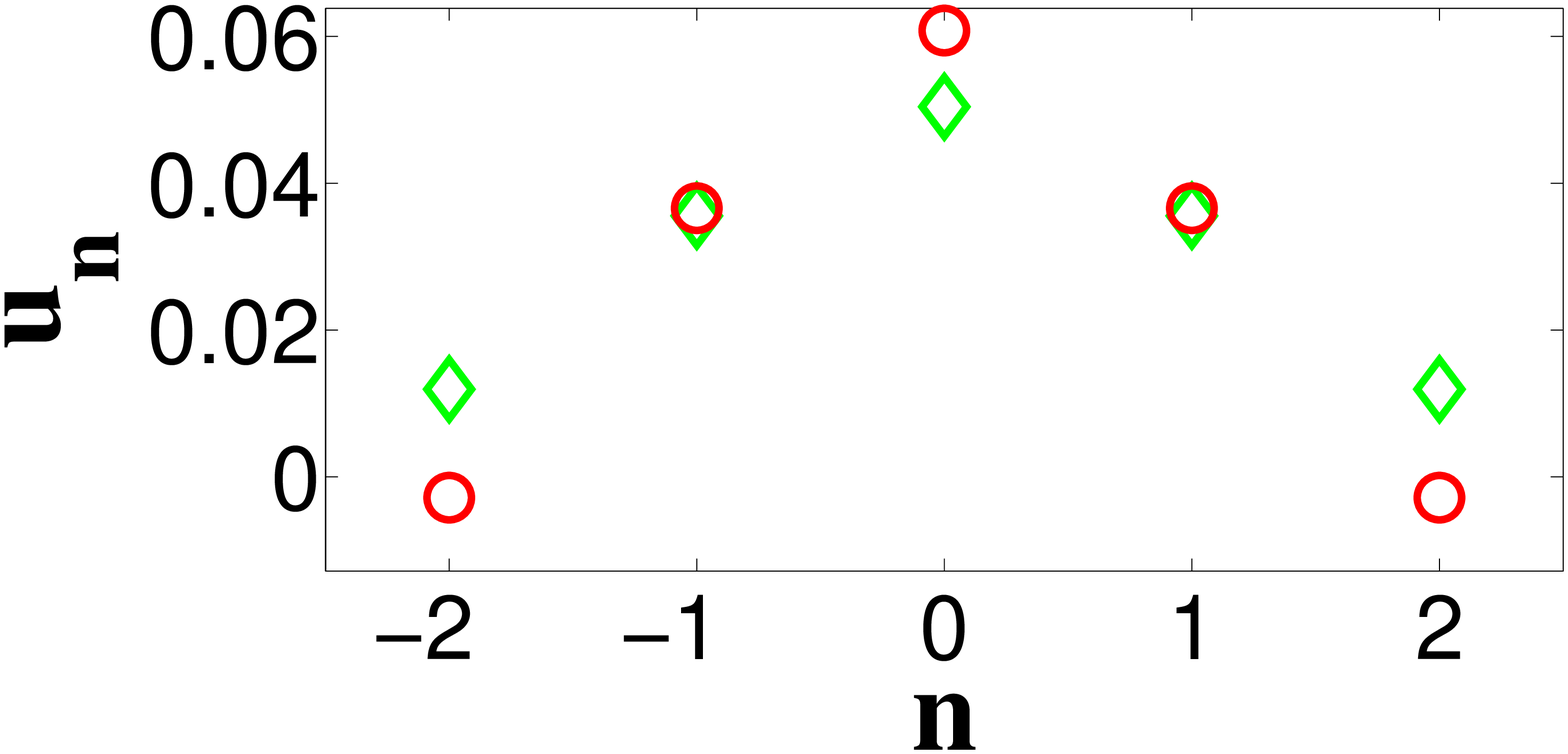}
\hspace*{0.5ex}
\includegraphics[scale=0.19]{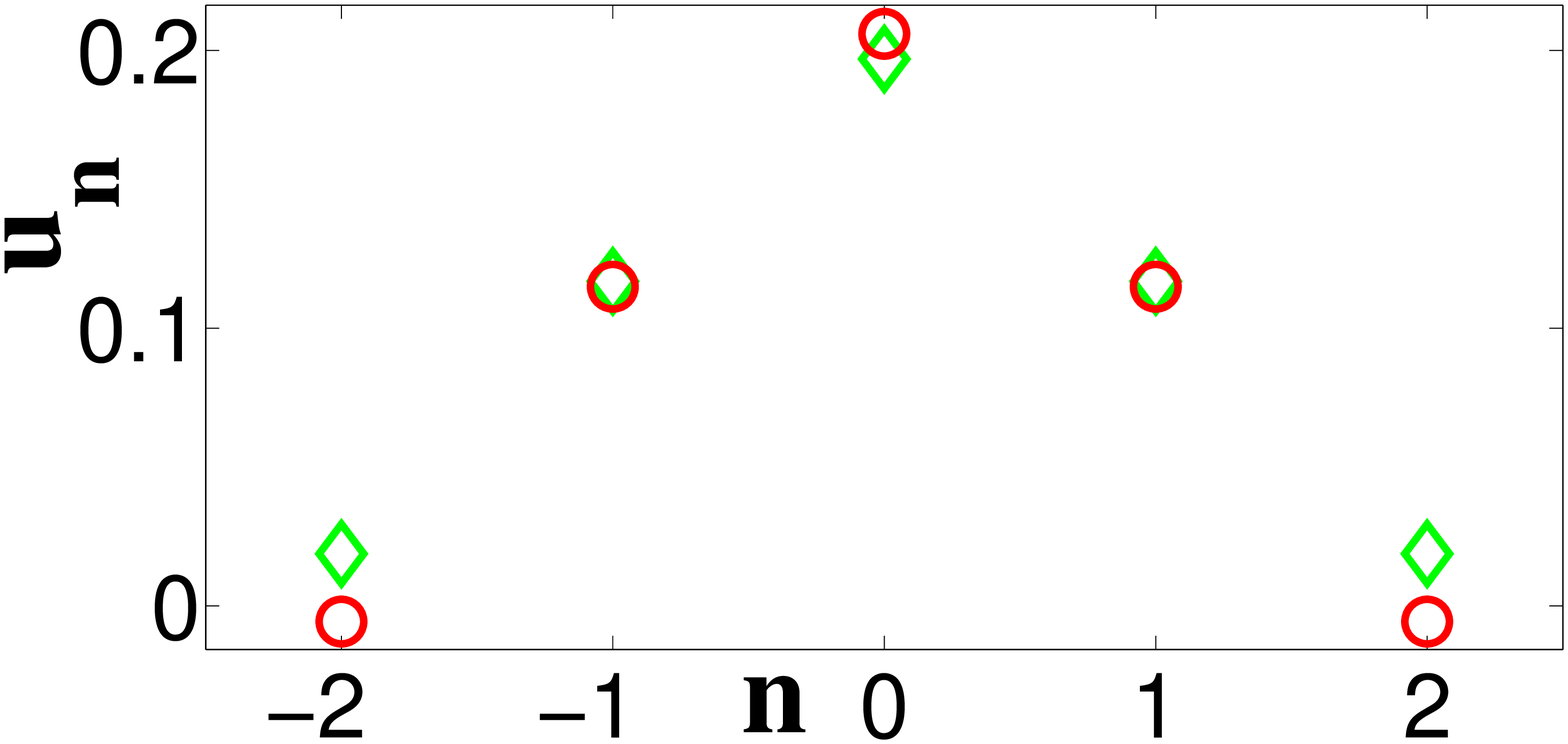}
\includegraphics[scale=0.2]{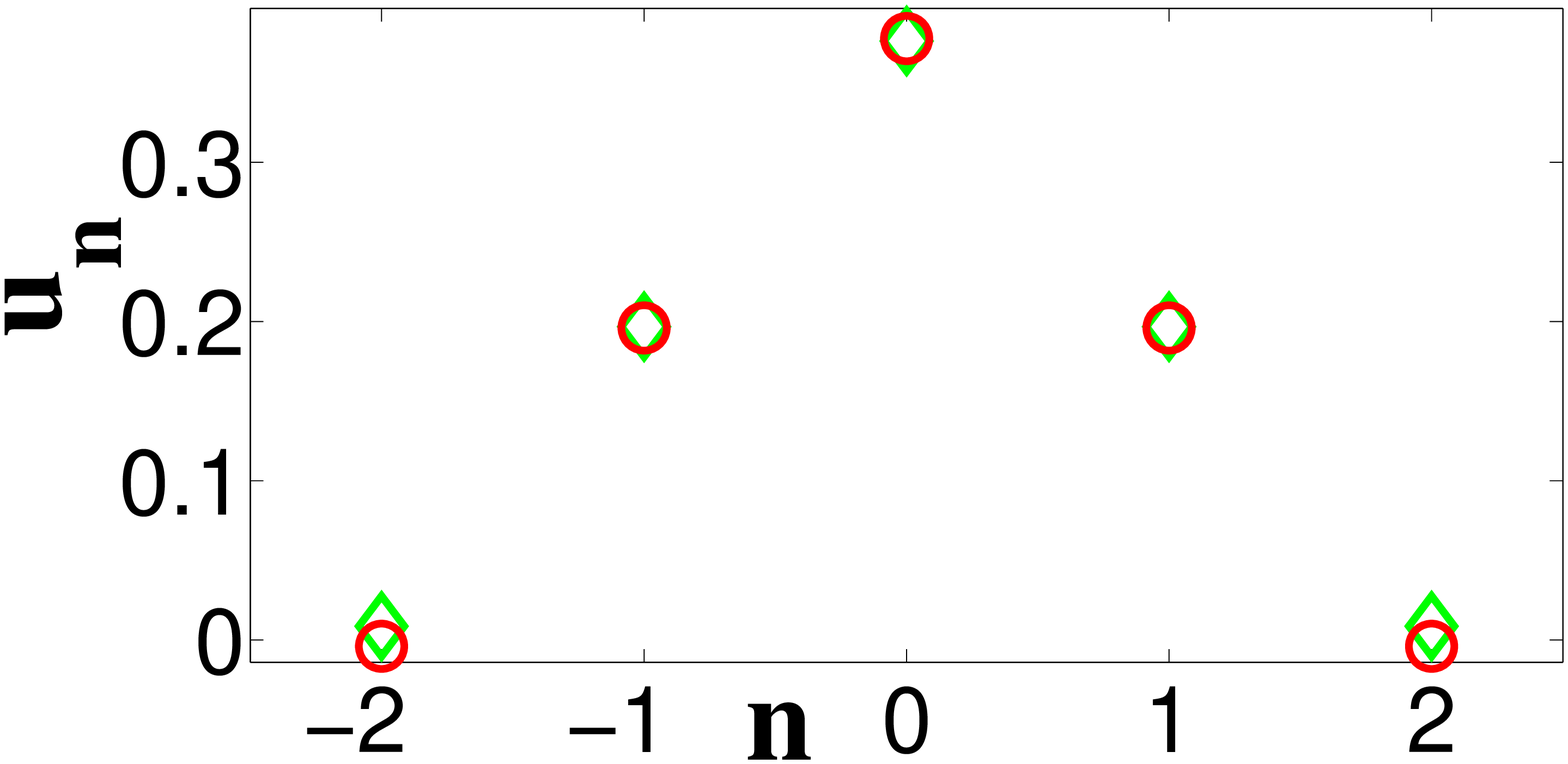}
\end{center}
\caption{\label{comphom}
Comparison of the approximate site-centered homoclinic 
orbit determined by (\ref{eq11}) and (\ref{icapprox}) with the exact 
site-centered homoclinic solution $\tilde{a}_n^1$ described in theorem \ref{homocsolu}
(this solution is computed numerically). The exact homoclinic orbit 
$\tilde{u}_n= (-1)^n \tilde{a}_n^1$ corresponds to green diamonds and the
approximate homoclinic orbit is represented with red circles. 
The site-centered homoclinic orbit is spanned over 
approximately 5 particles of interest (i.e. $n=-2,...,2$).
Top panel~: $\alpha = 3/2$ (relative error: 31.64{\%}),
middle panel~: $\alpha = 2$ (relative error: 13.98{\%}),
bottom panel~: $\alpha = 3$ (relative error: 3.83{\%}).
The relative error drops to 1.19{\%} for $\alpha = 4$
and 0.39{\%} for $\alpha = 5$ (profiles not shown).
}
\end{figure}

To improve the above approximation we resort to the improved 
approximate stable manifold $\mathcal{W}^{\text{s},(2)}_{\text{app}}$,
parametrized in the following way using (\ref{eq:3}),
\begin{equation}
\label{eq13}
\left( {\begin{array}{c}
 v^{(2)}(x) \\ 
 y^{(2)}(x) \\ 
 \end{array}} \right)=\left( {\begin{array}{c}
 P_{1/\alpha } \left( {P_{1/\alpha } (x)-P_{\alpha } \left( x \right)-2x} 
\right)-P_{1/\alpha } \left( x \right)+x+P_{\alpha } \left( x \right) \\ 
 P_{1/\alpha } (x)-P_{\alpha } \left( x \right)-2x \\ 
 \end{array}} \right) .
\end{equation}
Thus, setting $(v_1,y_1)=( v^{(2)}(x),y^{(2)}(x))$ and
$v_1=2\, y_1$, we obtain an approximation of the
homoclinic intersection determined by
\begin{equation}
\label{eq14}
\begin{array}{l}
 P_{1/\alpha } \left( {P_{1/\alpha } (x )-P_{\alpha } \left( {x } 
\right)-2x } \right)-3P_{1/\alpha } \left( {x } \right)+3P_{\alpha } 
\left( {x } \right)+5x =0, \\ 
 y_1 = P_{1/\alpha } (x )-P_{\alpha } \left( {x } \right)-2x. \\ 
 \end{array} 
\end{equation}
Here it is important to emphasize that in contrast to the case of the first order approximation 
(\ref{eq:1}) where the intersection point $y_{1} $ could be analytically approximated and is given explicitly in (\ref{eq11}), the solution of (\ref{eq14}) is calculated numerically by a Newton-type method.

In Table \ref{tab3}, we compare the 
accuracies of the approximate site-centered homoclinic solutions 
computed from the approximate stable manifolds
$\mathcal{W}^{\text{s},(k)}_{\text{app}}$, for $k=1,2,3$
and $\alpha =1.5,\, 2, \, 3$.
The agreement obtained with the improved approximation (\ref{eq14})
is excellent (with a strong improvement compared to approximation (\ref{eqe})),
but at the expense of working with the more
complex formula (\ref{eq14}).

\begin{table}[h]
\begin{center}
\begin{tabular}{|c||c|c||c|c||c|c|}
\hline
$\alpha $& 
Approx.& 
Relative error &
Approx. & 
Relative error&
Approx.& 
Relative error 
\\
 & 
$y_1^{(1)}$& 
on $a_n$ &
$y_1^{(2)}$& 
on $a_n$ &
$y_1^{(3)}$& 
on $a_n$ 
\\
\hline
3/2& 0.0312& 37.54 {\%}& 0.0255 & 1.78{\%} & 0.0252 & 0.0079 {\%}
\\
\hline
2&0.1038 & 16.72 {\%}& 0.0984761& 0.02{\%} &0.0984678 &4.37$\cdot 10^{-7}$ {\%} \\
\hline
3&0.1890611 & 4.69 {\%}&0.1880980 & 8.69$\cdot 10^{-5}${\%} &0.1880981 &3.32$\cdot 10^{-7}$ {\%}  \\
\hline
\end{tabular}
\end{center}
\caption{\label{tab3}
Approximations 
$y_{1}^{(1)} ,y_{1}^{(2)} ,y_{1}^{(3)}$ of the exact value of
$y_1$ for the site-centered homoclinic orbit. The approximations
correspond to the intersection between 
$\mbox{Fix}(R_{1})$
and the approximate stable manifolds
$\mathcal{W}^{\text{s},(1)}_{\text{app}}$, $\mathcal{W}^{\text{s},(2)}_{\text{app}}$, 
$\mathcal{W}^{\text{s},(3)}_{\text{app}}$.
The relative errors have been calculated between the exact site-centered homoclinic solution 
$\tilde{a}_n^1$ described in theorem \ref{homocsolu} and
the approximate ones, according to formula (\ref{eq12}) 
(5 lattice sites have been used to evaluate the error). 
The approximate site-centered homoclinic solutions are computed from the iteration of the map 
starting from the initial conditions
$(v_1,y_1)=(2\, y_{1}^{(i)} ,y_{1}^{(i)} )$ with $i=1,2,3$.}
\end{table}

As a conclusion, 
the approximation method of invariant manifolds and
homoclinic orbits described in this section is
efficient when $\alpha$ is far enough from unity.
For $\alpha$ close to unity, it becomes unpractical because approximate
stable manifolds $\mathcal{W}^{\text{s},(k)}_{\text{app}}$ involving increasingly
complex parametrizations must be used to reach a good precision. 

\subsection{\label{sectapprox2}Continuum limit for weak nonlinearities}

In this section we develop another
approximation method for invariant manifolds and
homoclinic orbits based on a continuum limit, which is
obtained when $\alpha$ is close to unity. This
method completes the one described in section \ref{sectapprox1}, which
was efficient for $\alpha$ far enough from unity.

\paragraph{Continuum limit model with logarithmic nonlinearity}

We consider equation (\ref{dpsstatbis}) for $s=1$. As seen in section \ref{fov},
$y_n =x_{n-1}$ satisfies
\begin{equation}
\label{eqyn}
y_{n+1} -2y_n + y_{n-1} = P_{1/\alpha}(y_n )-4y_n, \ \ \
{n \in \mathbb{Z}}.
\end{equation} 
We renormalize the problem by setting
\begin{equation}
\label{defzn}
y_n = 4^{\frac{\alpha}{1-\alpha}}\, z_n
\end{equation} 
(the prefactor corresponds to a fixed point of (\ref{eqyn})) and 
define $f_\alpha (z)=P_{1/\alpha}(z )-z$. Dividing (\ref{eqyn}) by $\alpha -1$
and using (\ref{defzn}) yields
\begin{equation}
\label{eqzn}
\frac{z_{n+1} -2z_n + z_{n-1}}{\alpha -1} = \frac{4}{\alpha -1}\, f_\alpha (z_n), \ \ \
{n \in \mathbb{Z}}.
\end{equation} 
This system has a well-defined formal limit when $\alpha \rightarrow 1$. Indeed,
we observe that for all $z>0$,
$$
\lim_{\alpha \rightarrow 1}\frac{f_\alpha (z)}{\alpha -1}=-z\, \ln{(z)}.
$$ 
Let us set
\begin{equation}
\label{eq21}
z_{n} =z\left( {\sqrt {\alpha -1} \left( {n-n_0} \right)} \right),
\end{equation} 
where $z$ denotes a sufficiently smooth function and $n_0$ will be specified later.
Letting $\alpha \rightarrow 1$ in (\ref{eqzn}) yields the ODE
\begin{equation}
\label{eq15}
\frac{d^2 z}{dx^2}+4z\ln \left( z \right)=0 .
\end{equation}
This equation can be seen as a
(one-dimensional) stationary logarithmic nonlinear Schr\"odinger equation,
a system originally introduced in the context of nonlinear wave mechanics \cite{bm}.
It admits the first integral of motion
\begin{equation}
\label{eq16}
{\left(\frac{dz}{dx}\right)}^{2}-4z^{2}\left( {\frac{1}{2}-\ln z} \right)=C
\end{equation}
and the Gaussian homoclinic solution
\begin{equation}
\label{homlog}
z(x)=\sqrt{e}\, e^{-x^2}.
\end{equation}
In what follows we deduce an approximate solution of the stationary DpS equation
from the above computations, and we
compare it to the site-centered homoclinic solution of the DpS equation
denoted in theorem \ref{homocsolu} by $\tilde{a}^1_n = (-1)^n \tilde{u}^1_n$. 
The site-centered symmetry of $\tilde{a}^1_n$ and  $\tilde{u}^1_n$ (i.e. $\tilde{u}^1_n=\tilde{u}^1_{-n}$)
corresponds to a bond-centered symmetry of $y_n =x^1_{n-1}=P_\alpha (\tilde{u}^1_{n-1}+\tilde{u}^1_{n-2})$,
i.e. we have $y_{1-n}=y_{2+n}$. In analogy with equation (\ref{defu1}),
the corresponding approximate
solution of (\ref{dpsstatbis}) is thus defined by $u_n^{{\rm app}} = y_{n+1}+ y_{n+2}$, with
\begin{equation}
\label{eq21bis}
y_{n} =4^{\left( {\frac{\alpha }{1-\alpha }} \right)}z_{n}, \ \ \ 
z_{n} =z\left( {\sqrt {\alpha -1} \left( {n-3/2} \right)} \right)=\exp 
\left( {\frac{1}{2}- {(\alpha -1) \left( {n-3/2} \right)^{2}} } \right),
\end{equation}
where we have fixed $n_0=3/2$ in (\ref{eq21}).
After some elementary computations, we get the expression
\begin{equation}
\label{eq22}
u_n^{{\rm app}} =
2^{\frac{1+\alpha}{1-\alpha}}\, \sqrt{e}\, e^{-(\alpha -1)\, (n^2+\frac{1}{4})}\, {\cosh}((\alpha -1)\, n).
\end{equation}
Figure \ref{comphomcont} compares
the exact site-centered homoclinic solution $\tilde{u}^1_n$ of (\ref{dpsstatbis}) 
(computed numerically by a Newton-type method) and 
the continuum approximation (\ref{eq22}). 
The left panel displays the relative error, which decreases to $0$
(almost linearly) when $\alpha$ approaches unity. 
The right panel shows the very good agreement between the approximate and numerical homoclinic
solutions for $\alpha = 1.05$. 
According to figure \ref{comphomcont}, the relative error drops below $10\%$ approximately when
$\alpha \leq 1.4$, but a refined approximation is required to approximate
the homoclinic solution precisely for higher values of $\alpha$. 
We shall address this problem in the sequel.

\begin{figure}[h]
\psfrag{8e-13}[1][Bl]{{\small $8\cdot 10^{-13}~~~~~$}}
\psfrag{7e-13}[1][Bl]{~}
\psfrag{6e-13}[1][Bl]{~}
\psfrag{5e-13}[1][Bl]{~}
\psfrag{4e-13}[1][Bl]{~}
\psfrag{3e-13}[1][Bl]{~}
\psfrag{2e-13}[1][Bl]{~}
\psfrag{1e-13}[1][Bl]{~}
\psfrag{0e+00}[1][Bl]{{\small $0~$}}
\psfrag{0}[1][Bl]{{\small $0$}}
\psfrag{-10}[1][Bl]{~}
\psfrag{10}[1][Bl]{~}
\psfrag{-5}[1][Bl]{~}
\psfrag{5}[1][Bl]{~}
\psfrag{0}[1][Bl]{{\small $0$}}
\psfrag{0}[1][Bl]{{\small $0$}}
\psfrag{0}[1][Bl]{{\small $0$}}
\psfrag{-15}[1][Bl]{{\small $-15$}}
\psfrag{15}[1][Bl]{{\small $15$}}
\psfrag{12}[1][Bl]{{\small $12$}}
\psfrag{11}[1][Bl]{~}
\psfrag{9}[1][Bl]{~}
\psfrag{8}[1][Bl]{~}
\psfrag{7}[1][Bl]{~}
\psfrag{4}[1][Bl]{~}
\psfrag{3}[1][Bl]{~}
\psfrag{1}[1][Bl]{~}
\psfrag{2}[1][Bl]{{\small $2$}}
\psfrag{6}[1][Bl]{{\small $6$}}
\psfrag{1.10}[1][Bl]{~}
\psfrag{1.15}[1][Bl]{~}
\psfrag{1.20}[1][Bl]{~}
\psfrag{1.30}[1][Bl]{~}
\psfrag{1.35}[1][Bl]{~}
\psfrag{1.40}[1][Bl]{~}
\psfrag{1.45}[1][Bl]{~}
\psfrag{1.05}[1][Bl]{{\small $1.05$}}
\psfrag{1.50}[1][Bl]{{\small $1.5$}}
\psfrag{1.25}[1][Bl]{{\small $1.25$}}
\psfrag{n}[1][Bl]{~~~~~~~~~~~~~~~~~~$n$}
\psfrag{u_n}[1][Bl]{\begin{tabular}{c}$u_n$\\ ~ \end{tabular}}
\psfrag{a}[1][Bl]{~~~~~~~~~~~~~~~~~~$\alpha$}
\psfrag{e}[1][Bl]{\begin{tabular}{c}Relative error({\%})\\ ~\\ ~ \end{tabular}}
\begin{center}
\includegraphics[scale=0.29]{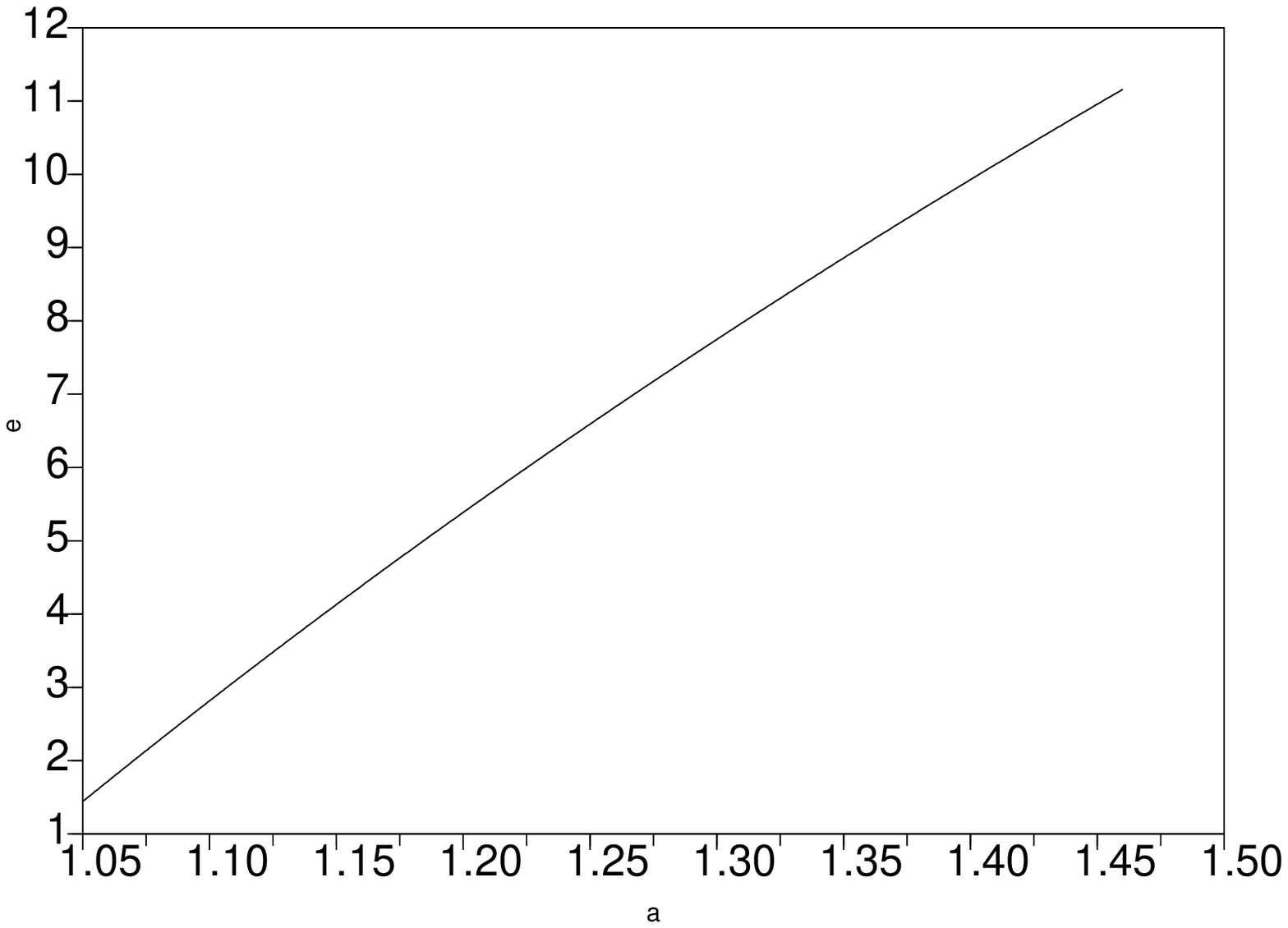}
\includegraphics[scale=0.29]{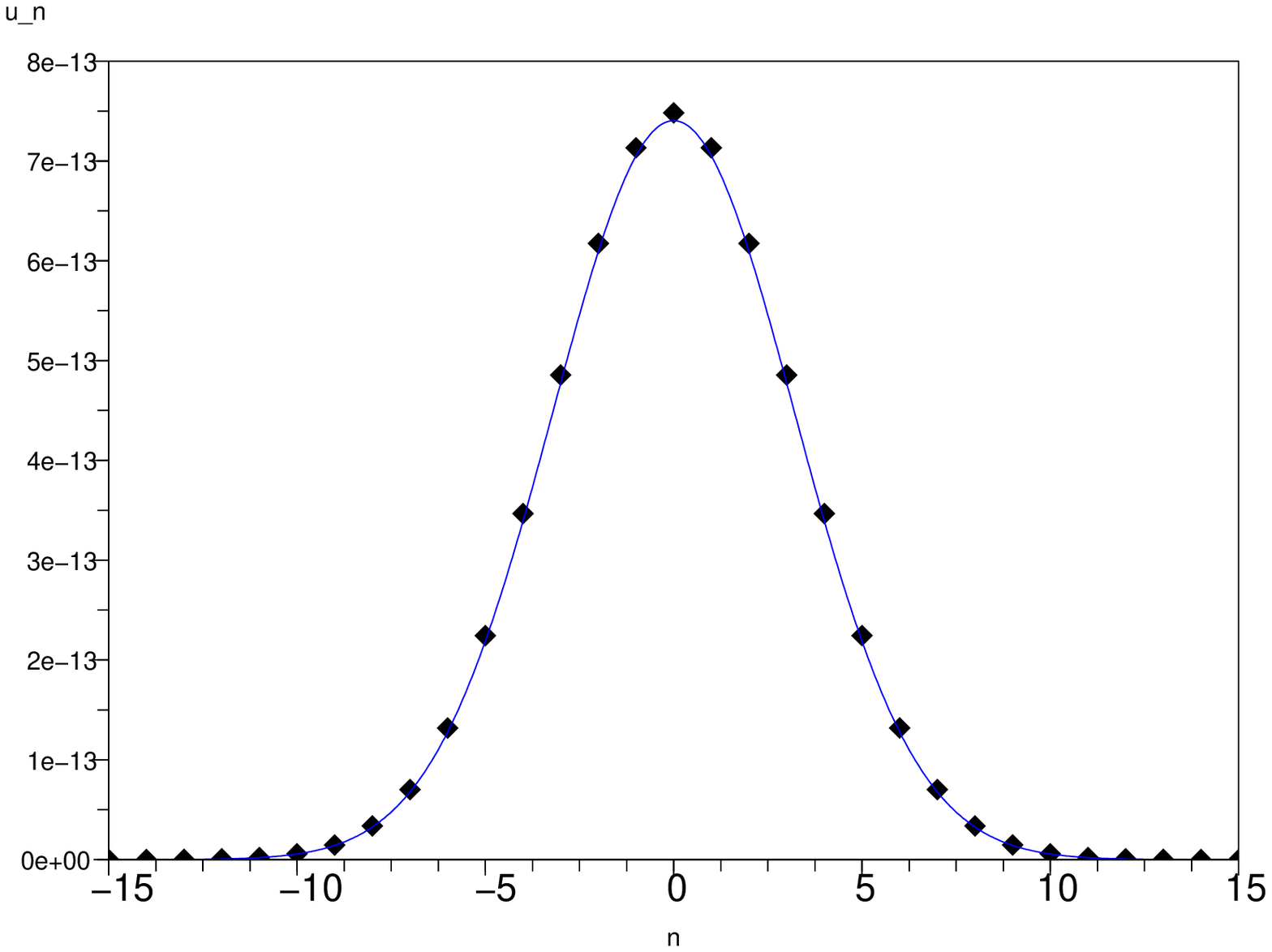}
\end{center}
\caption{\label{comphomcont}
Comparison of the continuum approximation defined by (\ref{eq22}) 
with the exact site-centered homoclinic solution of (\ref{dpsstatbis}) computed numerically. 
The left graph provides the relative error as a function of $\alpha$
(31 lattice sites are used to compute the error).
The right graph displays both profiles for the particular value $\alpha=1.05$.
The dots correspond to the numerically computed homoclinic orbit,
and the blue curve to the continuum approximation.
}
\end{figure}

\paragraph{\label{sectapprox3}Improved continuum limit}

In this section, we show that the continuum limit approximation derived in 
the previous section can be substantially improved.
Instead of letting $\alpha \rightarrow 1$ at both sides of (\ref{eqzn}),
we retain the right side as it is, yielding for $z \geq 0$
\begin{equation}
\label{eq1}
{z}''=\frac{4}{\alpha -1}\left[ {z^{1/\alpha }-z} \right] .
\end{equation}
Equation (\ref{eq1}) possesses the first integral of motion,
\begin{equation}
\label{eq2}
\frac{{z}'^{2}}{2}=\frac{4}{\alpha -1}\left[ {\frac{\alpha }{1+\alpha 
}z^{(1/\alpha ) +1}-\frac{z^{2}}{2}} \right]+C.
\end{equation}
Below we give a method for approximating the stable manifold of the origin $\mathcal{W}^s(0)$ 
for the map $F$ and the site-centered homoclinic orbit using the above continuum limit.
Using (\ref{defzn})-(\ref{eq21}) and
the centered finite difference approximation
$$
z^\prime \left( {\sqrt {\alpha -1} \left( {n-n_0} \right)} \right)
\approx \frac{z_{n+1} - z_{n-1}}{2\sqrt{\alpha -1}}
$$
in the left side of
(\ref{eq2}), one obtains for $y_n \geq 0$
\begin{equation}
\label{eq3}
\left( {y_{n+1} -y_{n-1} } \right)^{2}=8\left[ {\frac{\alpha }{1+\alpha }y_{n}^{(1/\alpha )
+1}-2y_{n}^{2}} \right].
\end{equation}
Let us assume this equality holds true for all orbit of $F$ along $\mathcal{W}^s(0)$,
lying near the origin for all $n\geq 0$.
Using (\ref{eqyn}) and (\ref{dpsstattrans}) at rank $n-1$, we obtain
$$
\left( y_n^{1/\alpha} -2v_n \right)^{2}=8\left[ {\frac{\alpha }{1+\alpha }y_{n}^{(1/\alpha )
+1}-2y_{n}^{2}} \right].
$$
Consequently, we obtain an approximation of (a part of) the stable manifold
of the origin, given for $y \geq 0$ by the parametrization
\begin{equation}
\label{eq4}
v=\frac{1}{2} y^{1/\alpha}+ \sqrt 2 \left[ {\frac{\alpha }{1+\alpha }y^{(1/\alpha )
+1}-2y^{2}} \right]^{1/2}.
\end{equation}
In figure \ref{fig1} we compare the analytic approximation (\ref{eq4})
of the stable manifold with the numerically computed stable manifold,
and find a good agreement provided $\alpha$ is below $1.5$. 

\begin{figure}[h]
\begin{center}
\includegraphics[scale=0.205]{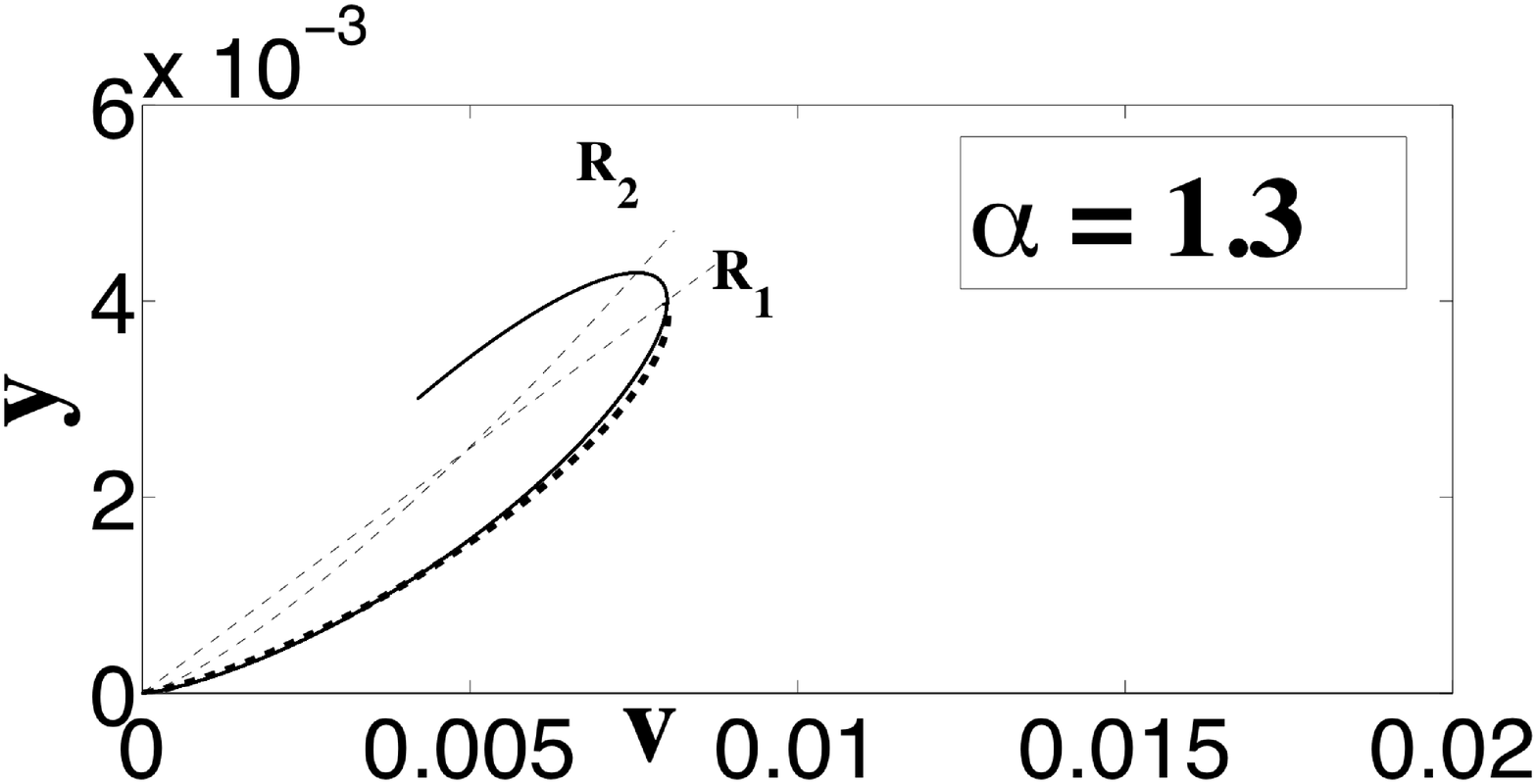}
\includegraphics[scale=0.215]{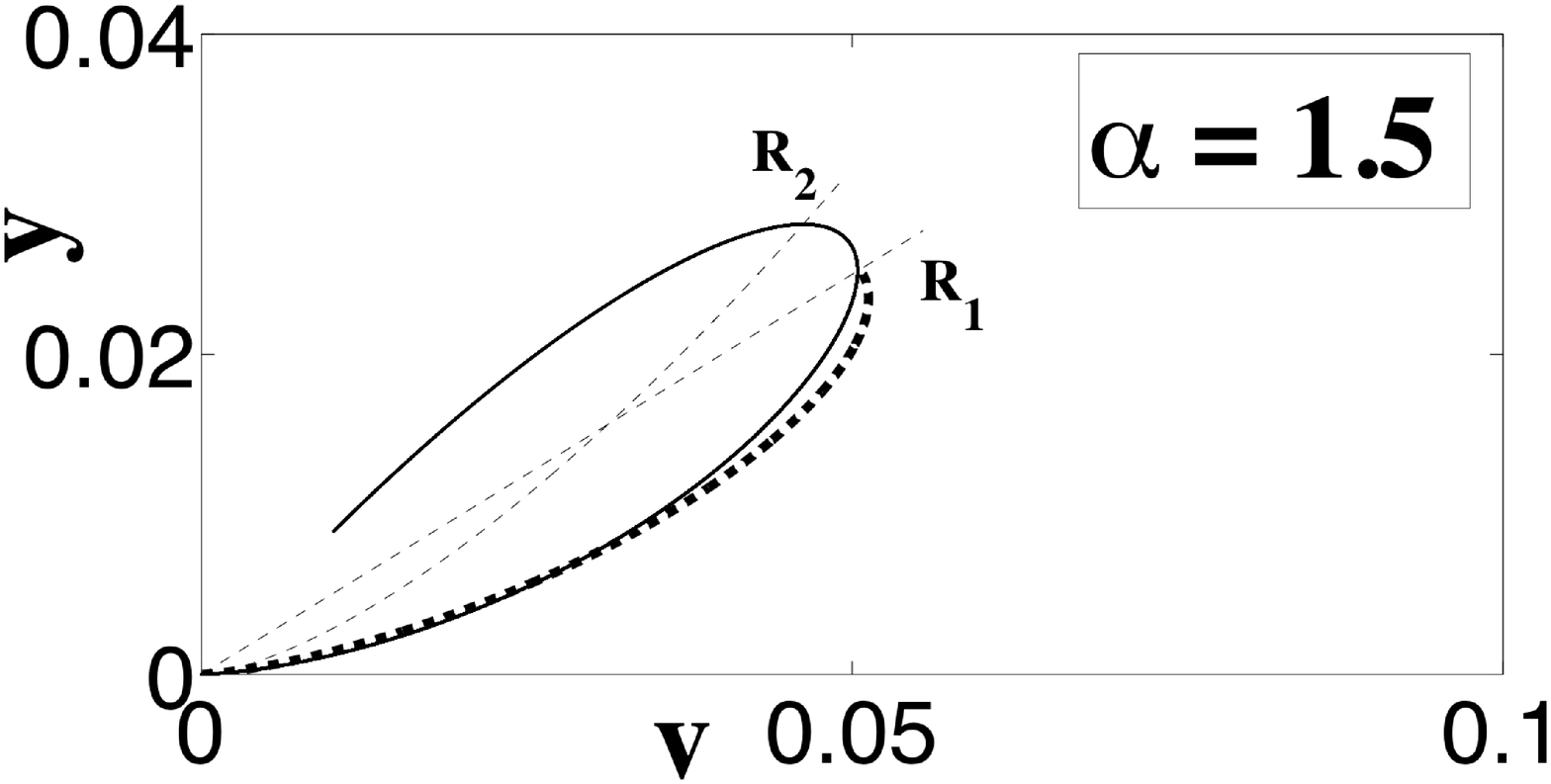}
\includegraphics[scale=0.2]{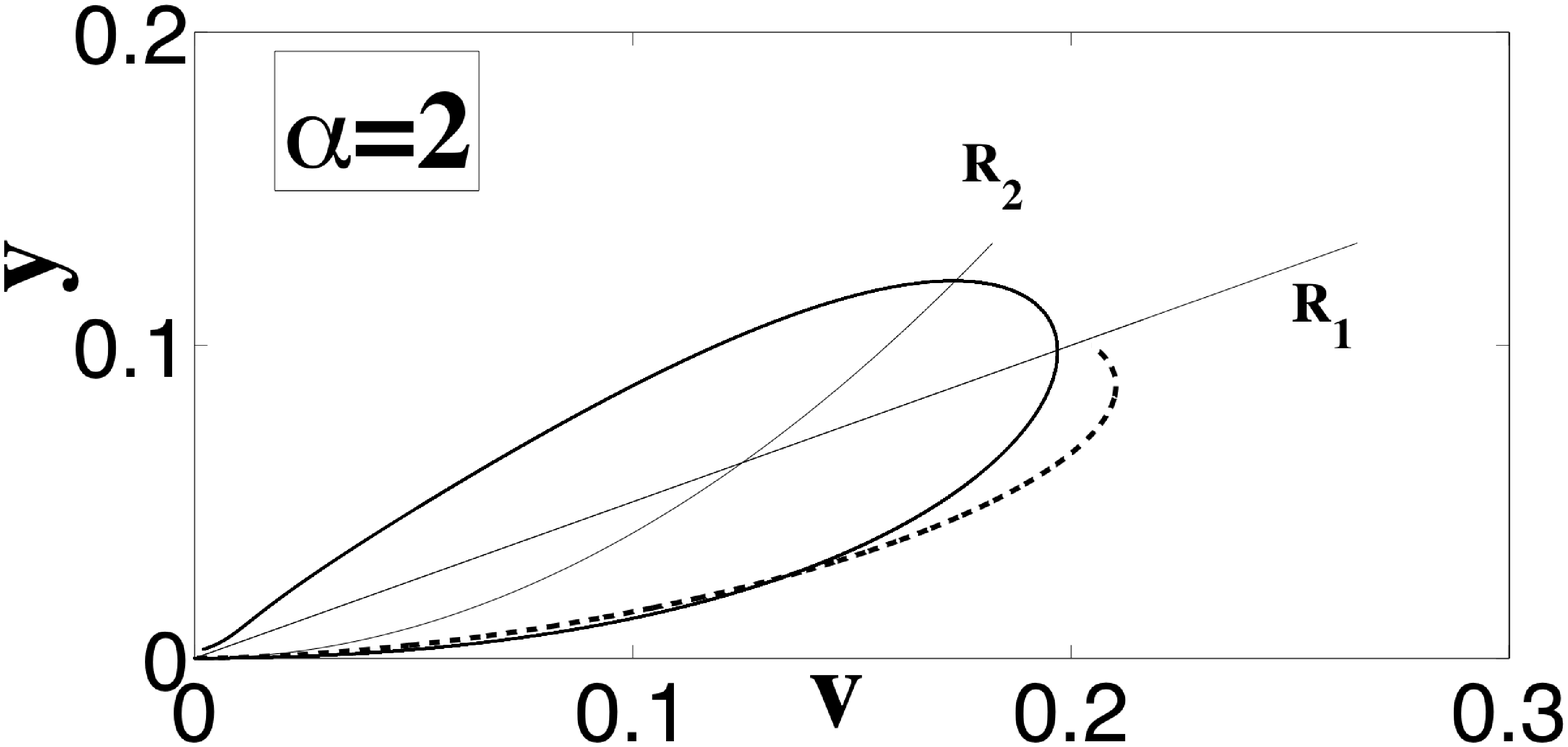}
\end{center}
\caption{\label{fig1}
Comparison of the numerically computed stable manifold of the origin for the map $F$
(bold solid curve) with continuum approximation (\ref{eq4}) (bold dashed curve). 
System parameter: $\alpha =1.3$ (top),
$\alpha =1.5$ (middle),
$\alpha =2$ (bottom).
The fixed sets of reversibility symmetries $\mathcal{R}_{1}$, $\mathcal{R}_{2}$ are
also plotted (thin dashed lines).
}
\end{figure}

To further illustrate the improvement of the continuum limit model, 
we compare the approximate site-centered homoclinic orbit deduced from 
the above computations with the exact solution computed numerically.
The approximate site-centered homoclinics can be obtained from
the point of intersection between the approximate stable manifold and 
$\mbox{Fix}(\mathcal{R}_{1})$. Thus setting $v=2y$ in (\ref{eq4}), we derive 
a scalar nonlinear equation for the 
point of intersection $(v,y)=(2y_{{\rm i}},y_{{\rm i}})$
which belongs to the homoclinic 
orbit, namely
\begin{equation}
\label{eq5cont}
2y_{{\rm i}}=\frac{1}{2} y_{{\rm i}}^{1/\alpha}+ \sqrt 2 \left[ {\frac{\alpha }{1+\alpha }y_{{\rm i}}^{(1/\alpha )
+1}-2y_{{\rm i}}^{2}} \right]^{1/2}.
\end{equation}
The point of intersection can be obtained by solving (\ref{eq5cont}) iteratively as described above.
With a good approximation of $y_{\rm i}$ at hand, we then iterate the map $F$ to obtain the 
approximate homoclinic orbit on a few lattice sites.
We iterate the map until $y_n$ becomes negative or starts to increase (such errors are due to the
sensitivity of iterates of $F$ to initial conditions), after what the values of the approximate solution
are set to $0$. 

In figure \ref{fig4} we compare the approximate and exact homoclinic
orbits. We observe an excellent agreement even for much higher values of 
$\alpha $ than in the previous case. Indeed, for
$\alpha \leq 3$ the relative error doesn't exceed 10{\%}. 
It is interesting to notice that the error does not vary monotonically with respect to $\alpha$.

\begin{figure}[h]
\begin{center}
\includegraphics[scale=0.2]{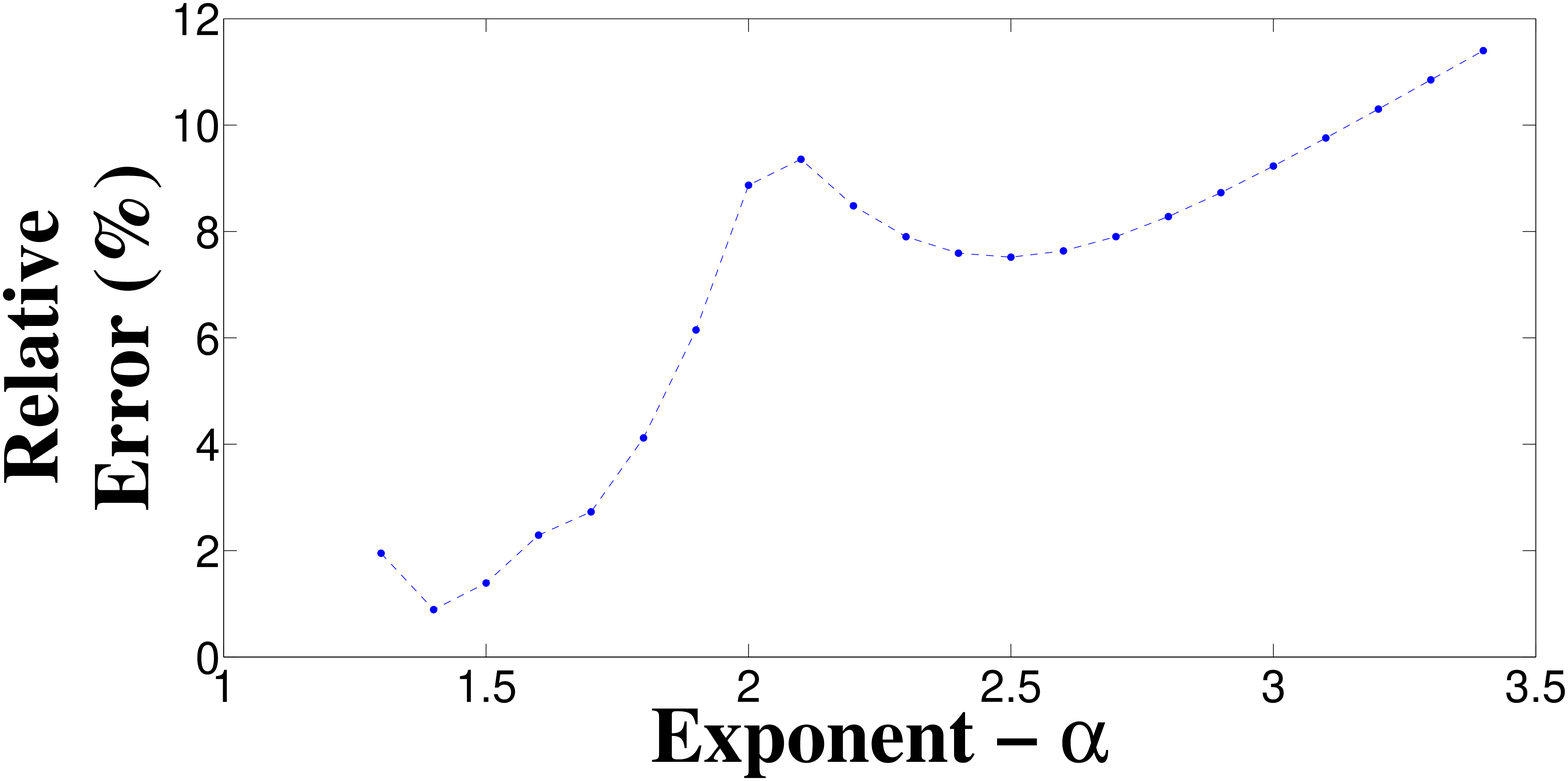}
\includegraphics[scale=0.2]{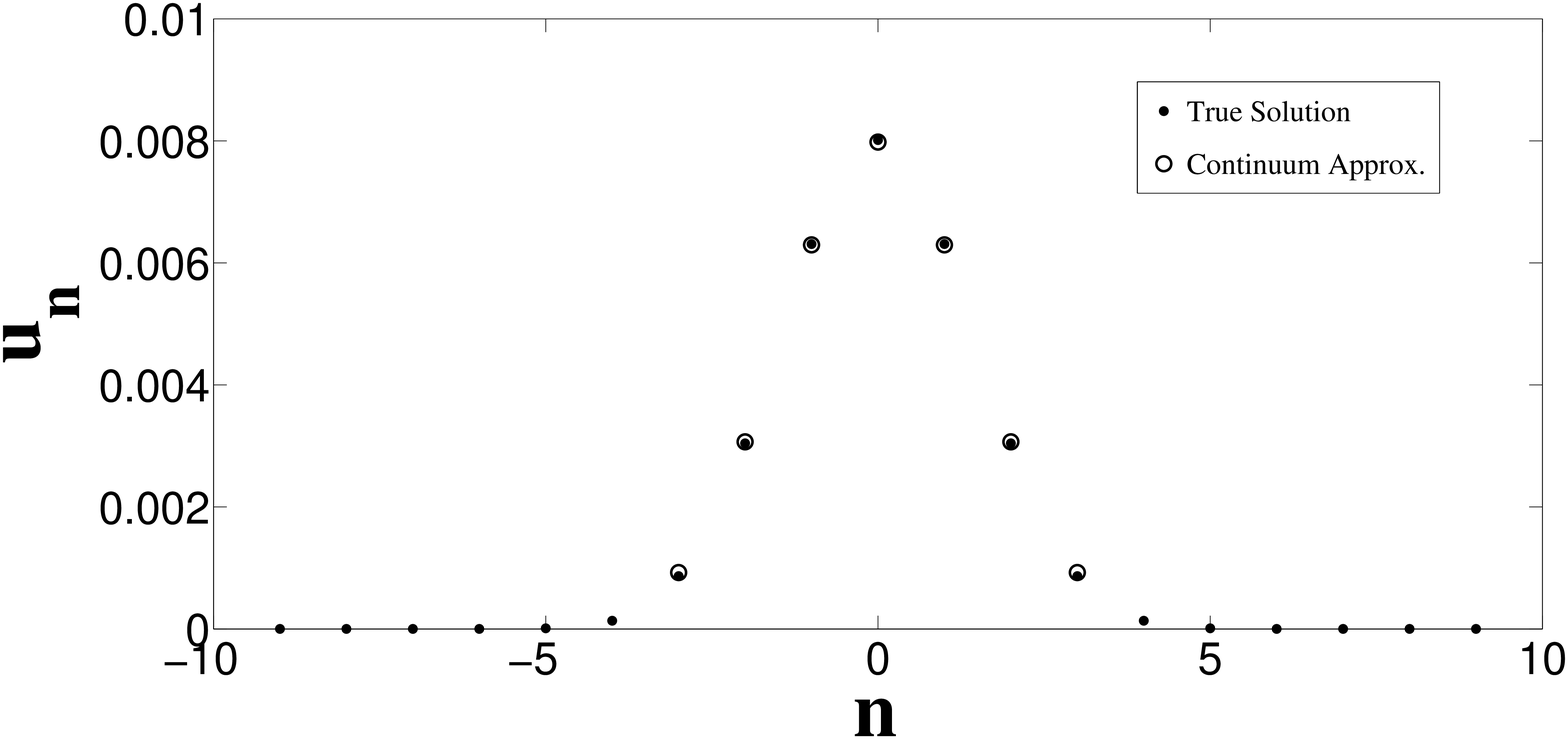}
\end{center}
\caption{\label{fig4}
Top plot~:
relative error between approximate and exact site-centered homoclinic orbits
vs. exponent $\alpha$.
The approximate solution is deduced from (\ref{eq5cont}) (see text) and the exact one
is computed numerically from the stationary DpS equation. The relative 
error is evaluated over the support of the approximate solution 
(the support being computed as indicated in the text). 
Bottom plot~:
comparison of the exact (bold solid dots) and approximate (circles)
homoclinic orbits for $\alpha =1.3$.
}
\end{figure}

\section{\label{stab}Breather stability and mobility}

In this section we numerically analyze the
spectral stability 
\rcgindex{\myidxeffect{S}!Spectral stability}
of site-centered 
and bond-centered breather solutions of (\ref{dpsa}) given by (\ref{solper})
(with $\Omega >0$),
in connection with their possible mobility (translational motion) under suitable
perturbations. These solutions 
correspond respectively to the site-centered 
homoclinic orbit $a^1_n$ and bond-centered homoclinic orbit $a^2_n$
described in theorem \ref{mainthm}. Their profiles
are illustrated on figures \ref{homocorb} and \ref{homocorb2} for $\alpha =3/2$ and $\alpha =1.2$.
This stability analysis was performed in the work \cite{JKC11} in the particular case
$\alpha=3/2$, and we shall extend this study to
some interval of physically meaningful values of $\alpha$.

The stability analysis can be simplified by taking advantage of the scale invariance of (\ref{dpsa}).
More precisely, if $\phi (t)= (\phi_n (t))_{n\in \mathbb{Z}}$ is a solution of (\ref{dpsa}) then so is 
${\Omega}^{\frac{1}{\alpha -1}}  \phi ({\Omega \, t})$ for all $\Omega >0$.
This scaling transformation allows one to generate the full family of breather solutions
(\ref{solper}) from the particular case $\Omega =1$. Consequently, 
the stable or unstable character
of these periodic solutions is independent of $\Omega$,
both for nonlinear, linear and spectral stability. For this reason
we shall restrict to the case $\Omega =1$ of (\ref{solper})
without loss of generality.

The bond- and site-centered
homoclinic solutions of the stationary DpS equation (\ref{dpsstat})
are computed numerically with a Newton-type method, for
a finite lattice of $N=21$ particles with zero boundary conditions.
Their spectral stability can be determined 
in complete analogy with the DNLS equation \cite{pgk},
by adopting a perturbation of the form 
\begin{equation}
\label{eq0stab}
\phi_{n} (t)=\exp (i\, t) \left( a_{n} + \varphi_{n} (t) \right) ,
\end{equation}
where $a=(a_{n} )_{1\leq n \leq N}$ denotes the (real-valued) localized solution of (\ref{dpsa})
corresponding to the unperturbed breather.
Substituting (\ref{eq0stab}) into (\ref{dpsa}) and linearizing with respect to
$\varphi = (\varphi_{n}(t) )_{1\leq n \leq N}$ yields a linear autonomous differential
equation for $(\varphi , \varphi^\star )$. Looking for solutions of the form  
\begin{equation}
\label{eq1stab}
\varphi_{n} (t)= \alpha_{n} \exp \left( {\lambda t} 
\right)+\beta_{n}^{\star } \exp \left( {\lambda^{\star }t} \right) 
\end{equation}
yields a linear eigenvalue problem $M\, V = i\, \lambda\, V$, where
$M\in M_{2N}(\mathbb{R})$ and $V=(\alpha_n ,\beta_n)^T_{1\leq n \leq N}$.
Due to the Hamiltonian character of (\ref{dpsa}), 
whenever $\lambda$ is an eigenvalue, so are $\lambda^{\star}$,
$-\lambda$ and $-\lambda^{\star}$.

Figure \ref{fig9b} displays the resulting eigenvalues in the complex plane, 
for the site- and bond-centered homoclinic solutions and $\alpha =3/2$
(the eigenvalue problem is solved by standard numerical linear algebra solvers). 
For all values of $\alpha$ in the interval $[1.2,3.4]$,
we find that the bond-centered breather is spectrally stable
(i.e. the eigenvalues $\lambda$ are purely imaginary), whereas
the site-centered breather is unstable through a simple real positive eigenvalue. 
In figure \ref{fig9c}
we plot $M^{\ast }= \mbox{Max}\{ \mbox{Re}(\lambda ) \}$,
i.e. the maximal real part 
reached by the eigenvalues 
for the site-centered breather, versus $\alpha \in [1.2,3.4]$.
These results show that the strength of the instability increases with $\alpha$,
and indicate that
the unstable eigenvalue goes to $0$ in the limit $\alpha \rightarrow 1^+$.

The above property can be interpreted intuitively in connection with the continuum limit
derived in section \ref{sectapprox2}, where one obtains asymptotically a
translationally invariant family of Gaussian breather solutions. Such 
families of solutions usually correspond to an extra double (non-semi-simple)
eigenvalue $0$ associated with a translation mode \cite{dmi},
and perturbations along the associated marginal mode (generalized eigenvector)
generate a translational motion \cite{aubryC}. 
However, this analogy is purely heuristic since breather solutions delocalize when $\alpha \rightarrow 1^+$.

\begin{figure}[h]
\begin{center}
\includegraphics[scale=0.25]{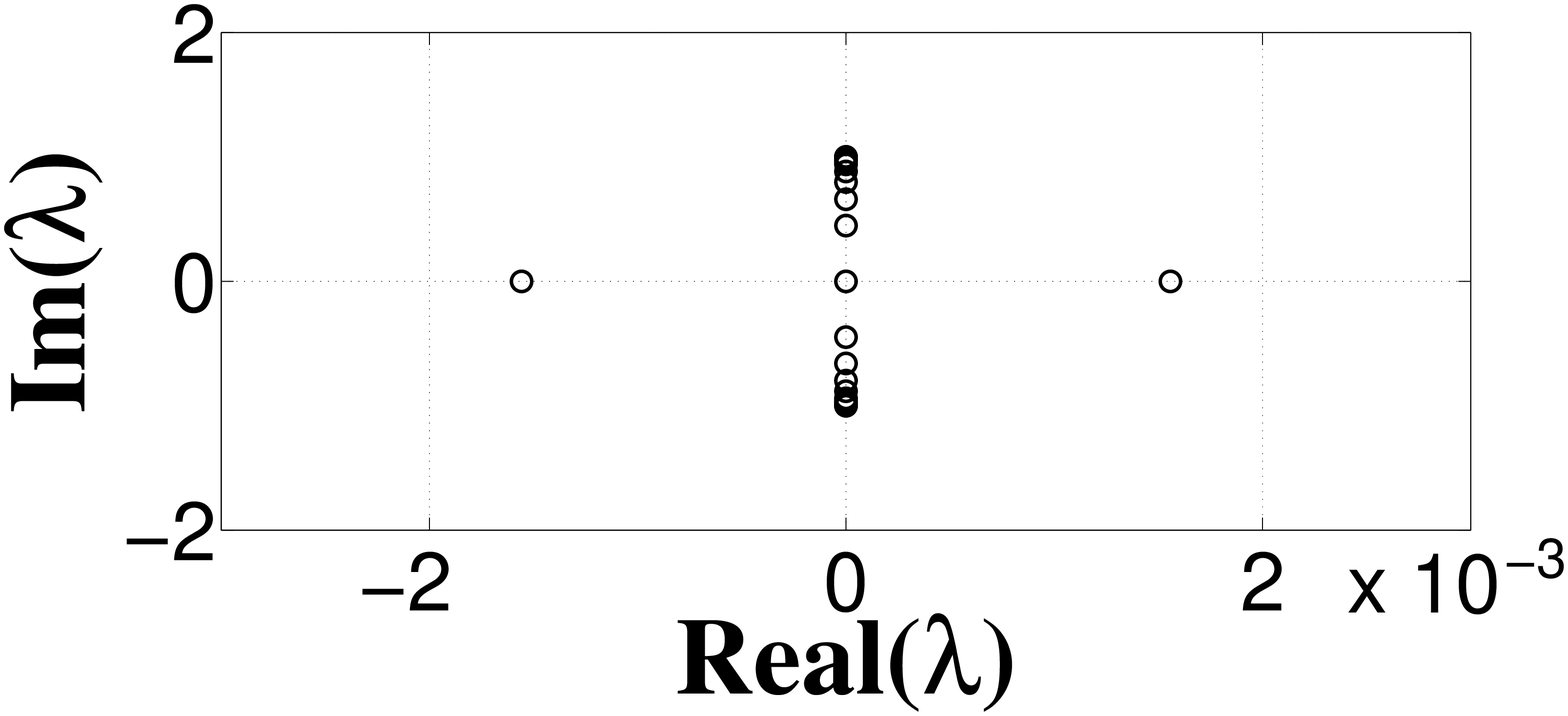}
\includegraphics[scale=0.25]{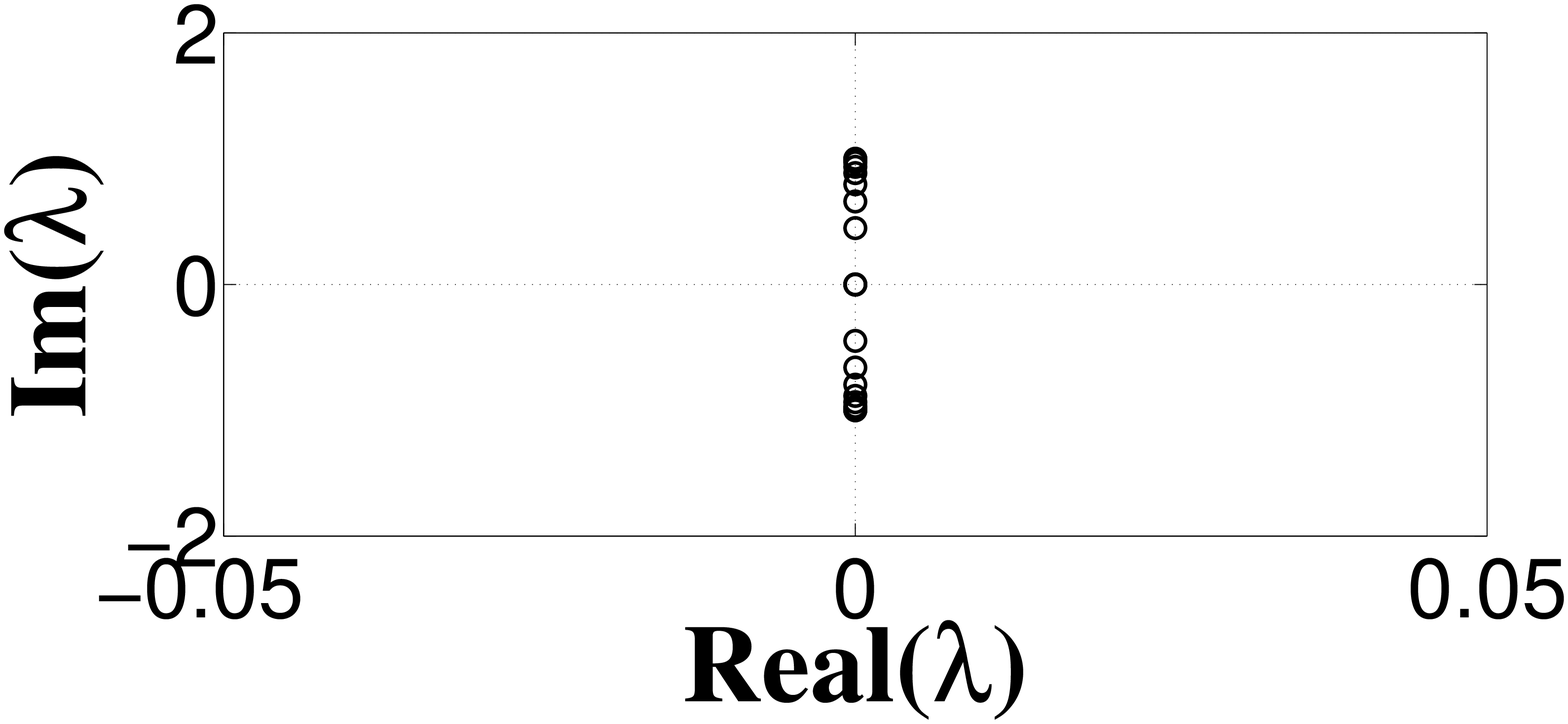}
\end{center}
\caption{\label{fig9b}
Eigenvalues of the linearization of the DpS equation
around the site-centered breather solution (top plot)
and the bond-centered breather solution (bottom plot),
computed for $\alpha =1.5$ and $\Omega =1$.}
\end{figure}

\begin{figure}[h]
\begin{center}
\includegraphics[scale=0.2]{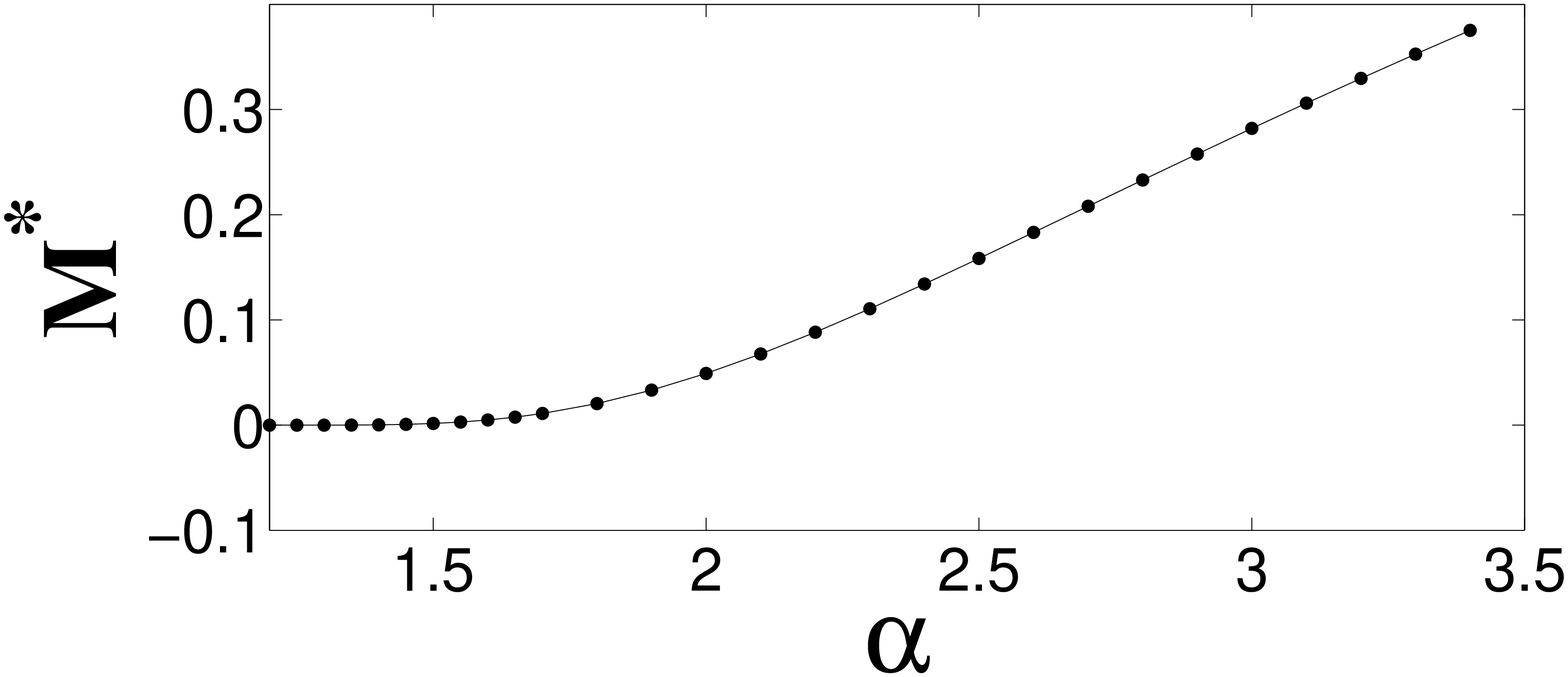}
\end{center}
\caption{\label{fig9c}
Maximal real part of the
eigenvalues of the linearization of the DpS equation
around the site-centered breather solution, plotted for
different values of $\alpha$ and $\Omega =1$.}
\end{figure}

To illustrate the dynamical effects of breather instability we attempt
the perturbation of the site-centered breather with $\Omega =1$
for different values of $\alpha$. We consider an eigenvector
$V^{\rm{u}}=(\alpha_n ,\alpha_n^\ast )_{1\leq n \leq N}$ associated with
the real unstable eigenvalue $\lambda >0$, and
the eigenvector
$V^{\rm{s}}=(\alpha_n^\ast , \alpha_n)_{1\leq n \leq N}$ associated with
the stable eigenvalue $-\lambda$. 
In formal analogy with the work \cite{aubryC},
we perturb the unstable breathers along the {\em approximate} marginal mode
$W=V^{\rm{u}}-V^{\rm{s}}$, which corresponds to fixing
\begin{equation}
\label{defperturb}
\phi_n (0)=a_n + c\, i\ \rm{Im}(\alpha_n),
\end{equation}
with $c \in \mathbb{R}$.
The parameter $c$ is tuned in order to achieve a desired level of
relative energy perturbation $\delta H_{\rm{rel}} = (H( \phi )- H( a )) / H( a )$,
where the energy $H$ is defined through (\ref{defh}).

Figure \ref{tbal1_5} illustrates the dynamics of a perturbed 
site-centered breather for $\alpha=3/2$. For
different strengths of the perturbation, the
instability of the site-centered breather leads to its 
translational motion. One can notice 
that additional breathers with smaller amplitudes are also emitted
when the initial perturbation is sufficiently large.
Note that random perturbations 
(which generally possess a nonzero
component along the marginal mode) also
typically generate a translational motion \cite{JKC11}, 
albeit the propagation velocity is generally smaller compared to a
``pure" perturbation (of similar strength) along the marginal mode. 

In figure \ref{tbal3}, the same type of computation is performed 
for $\alpha=3$. The situation is strikingly different compared to figure \ref{tbal1_5}.
The resulting dynamics is now mainly characterized by a pinning of the localized
excitation alternating with phases of irregular motion when the initial
perturbation is large enough. Such behavior may 
be induced by the interaction of the localized solution with
other localized or extended waves generated by the initial perturbation,
a situation reminiscent of numerical observations made in \cite{ssen}.

\begin{figure}[h]
\begin{center}
\includegraphics[scale=0.2]{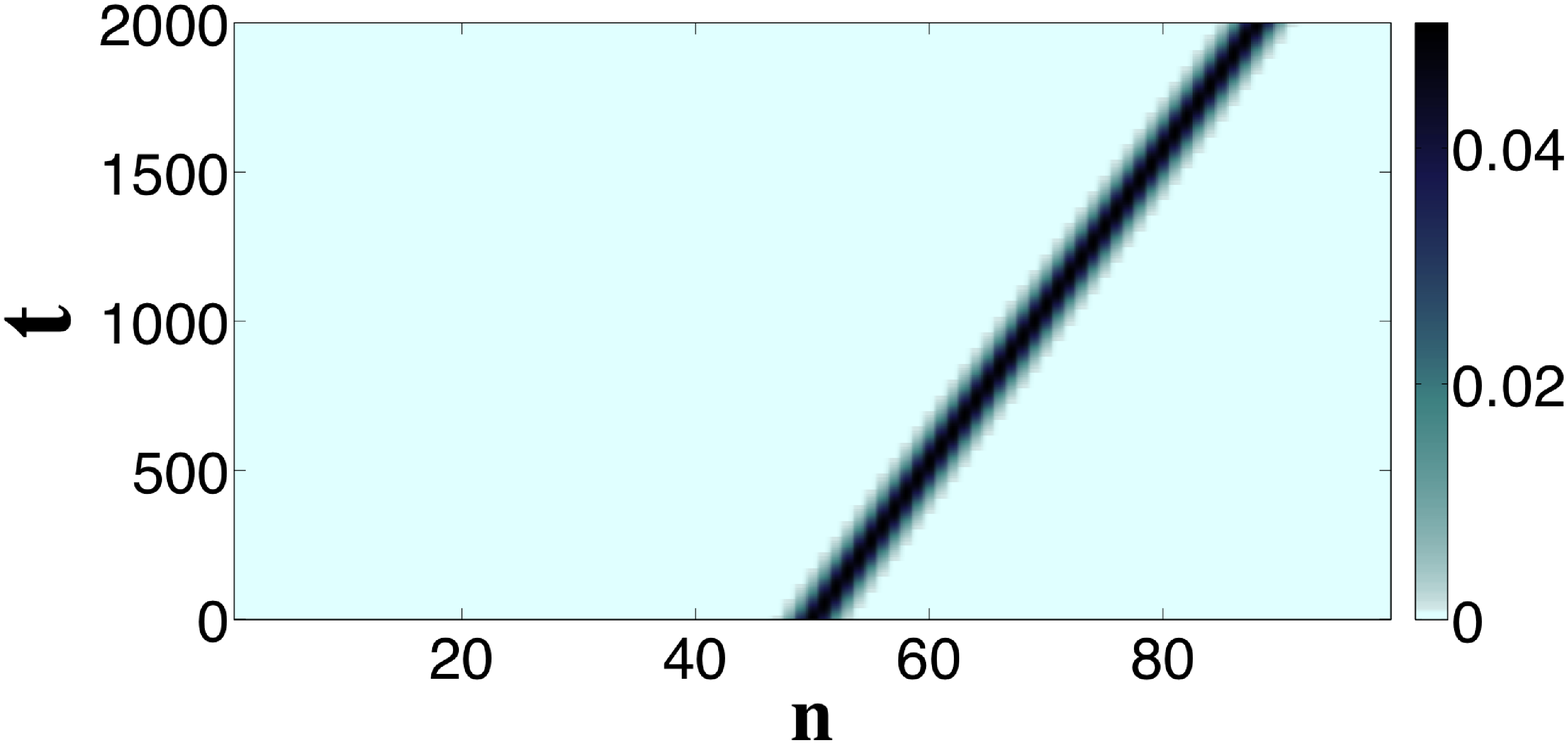}
\includegraphics[scale=0.2]{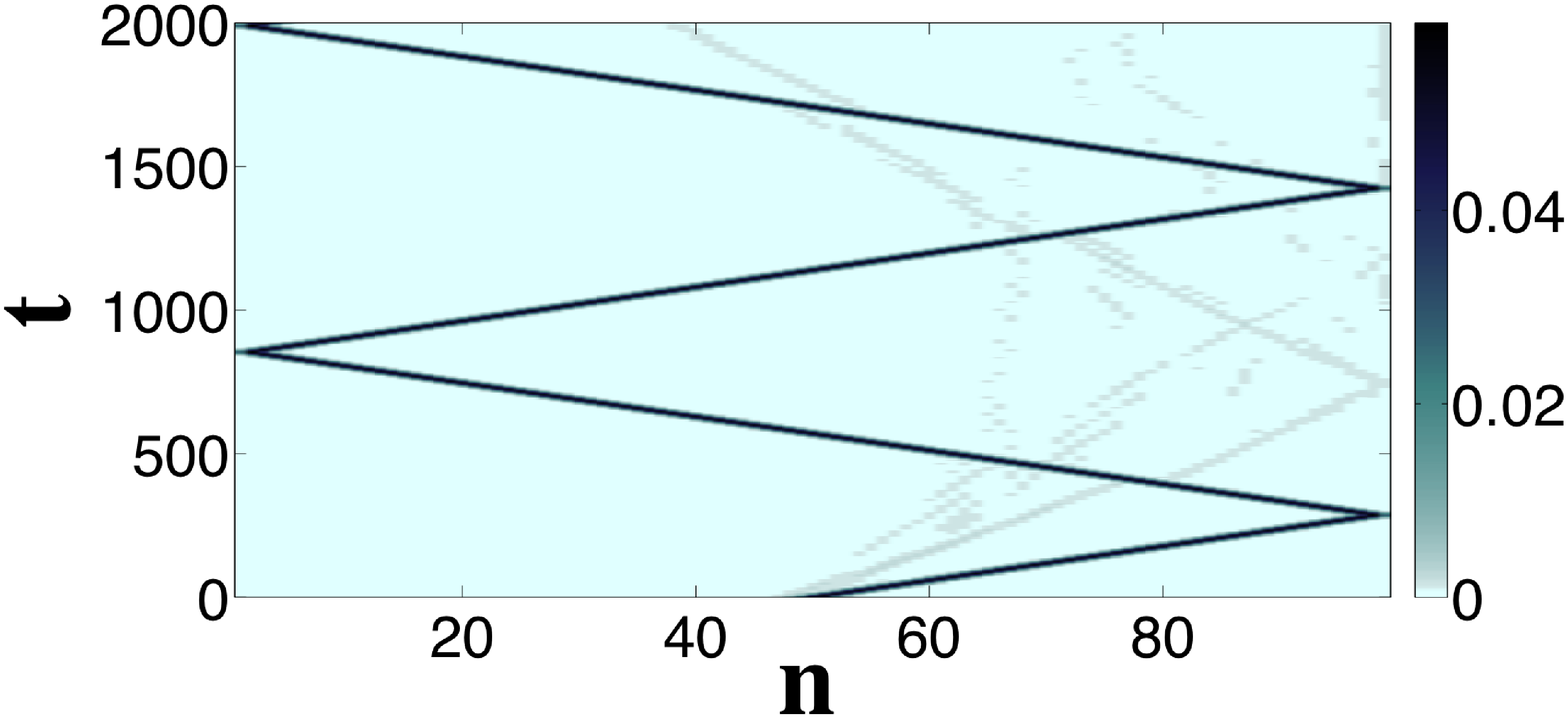}
\includegraphics[scale=0.2]{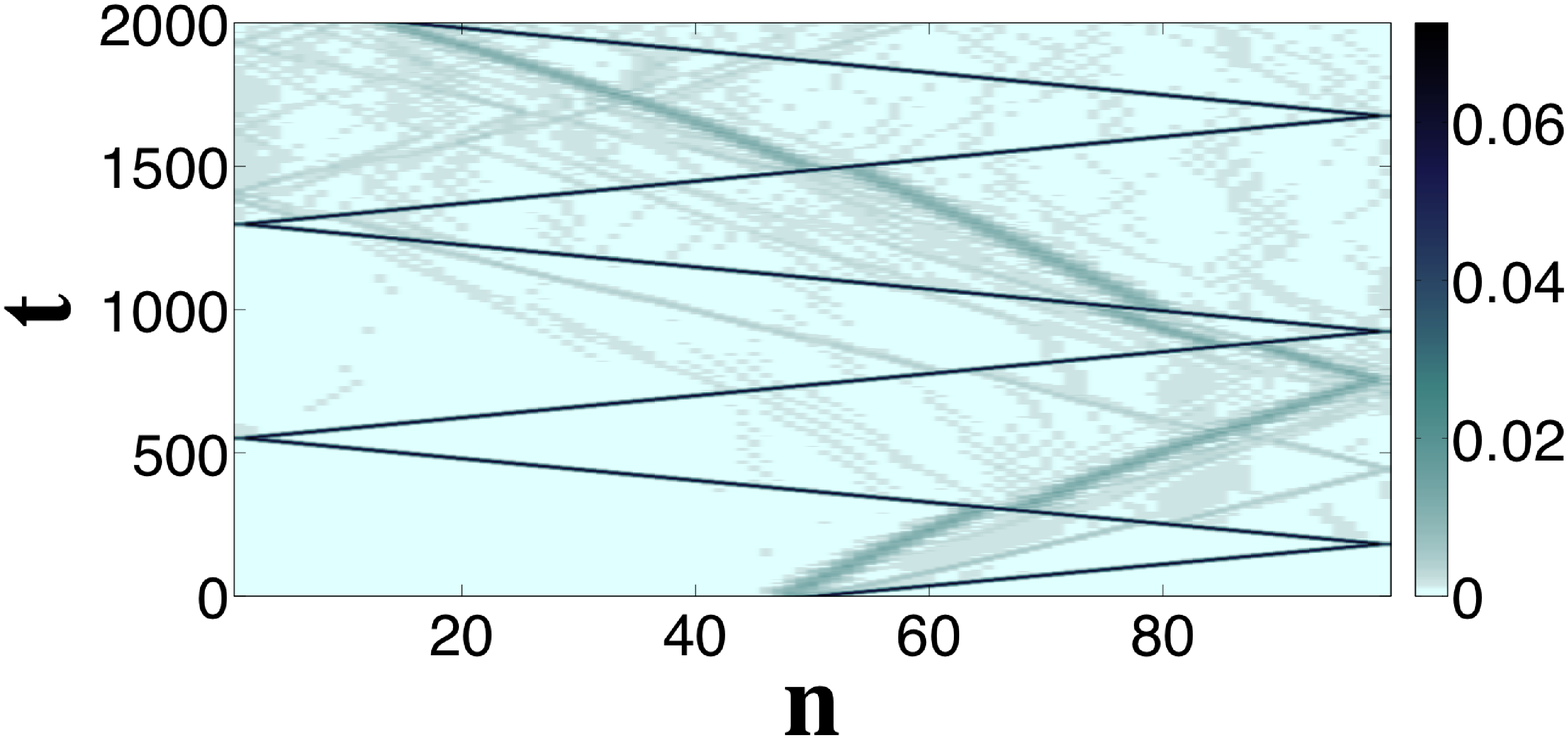}
\end{center}
\caption{\label{tbal1_5}
Space-time contour plot of the modulus of the field $|\phi_n(t)|$
for equation (\ref{dpsa}) with $\alpha=3/2$.
The initial condition is given by (\ref{defperturb}), and corresponds to
a site-centered breather (with frequency $\Omega=1$)
perturbed in the direction of the approximate marginal mode $W$.
Different levels of relative energy perturbation are considered~:
$\delta H_{\rm{rel}} \approx 0.001$ (top plot), 
$\delta H_{\rm{rel}} \approx 0.1$ (middle plot),
$\delta H_{\rm{rel}}\approx 0.5$ (bottom plot).
The perturbation leads to the manifestation of the
instability of the site-centered breather which, in turn, leads to its 
translational motion.
}
\end{figure}

\begin{figure}[h]
\begin{center}
\includegraphics[scale=0.2]{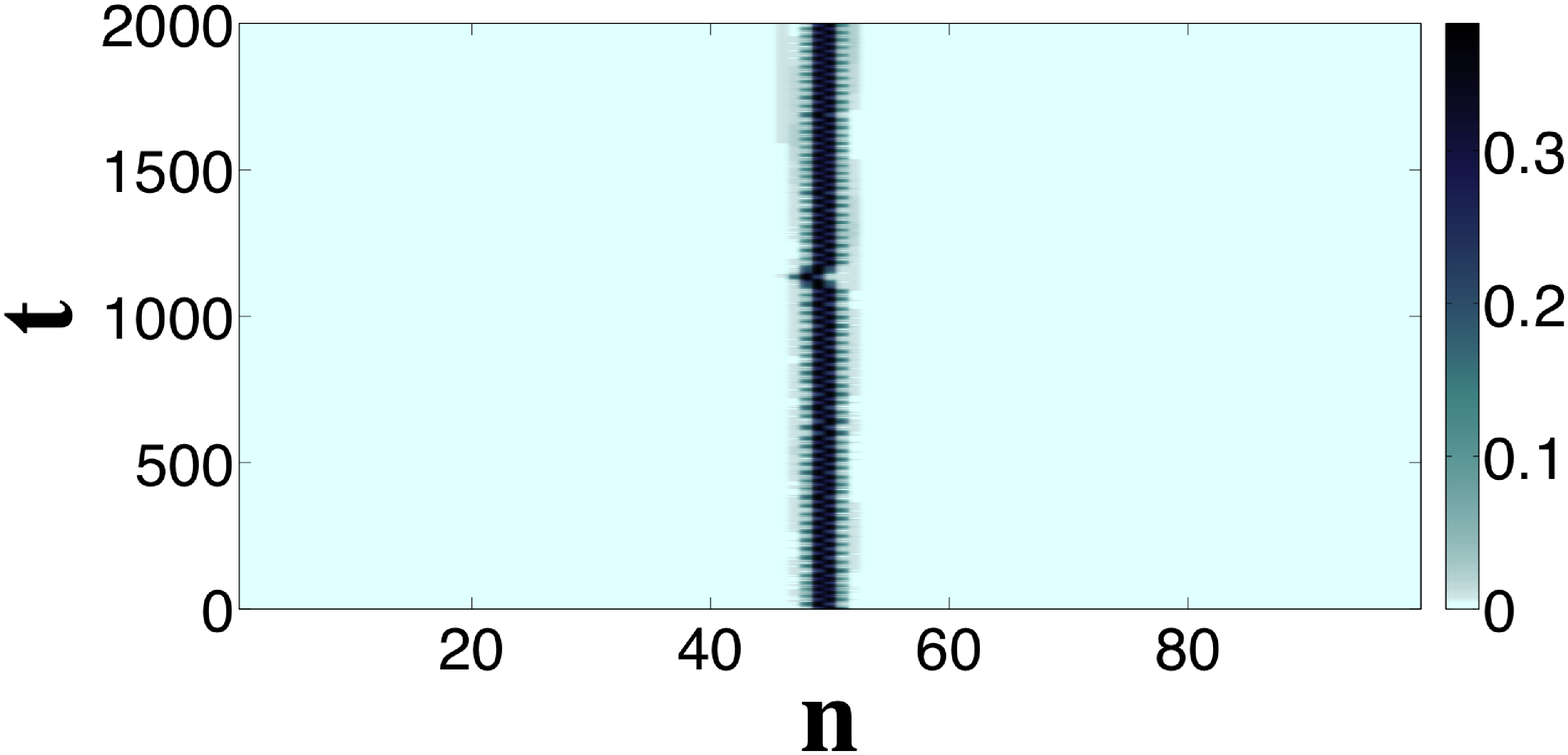}
\hspace*{1ex}
\includegraphics[scale=0.195]{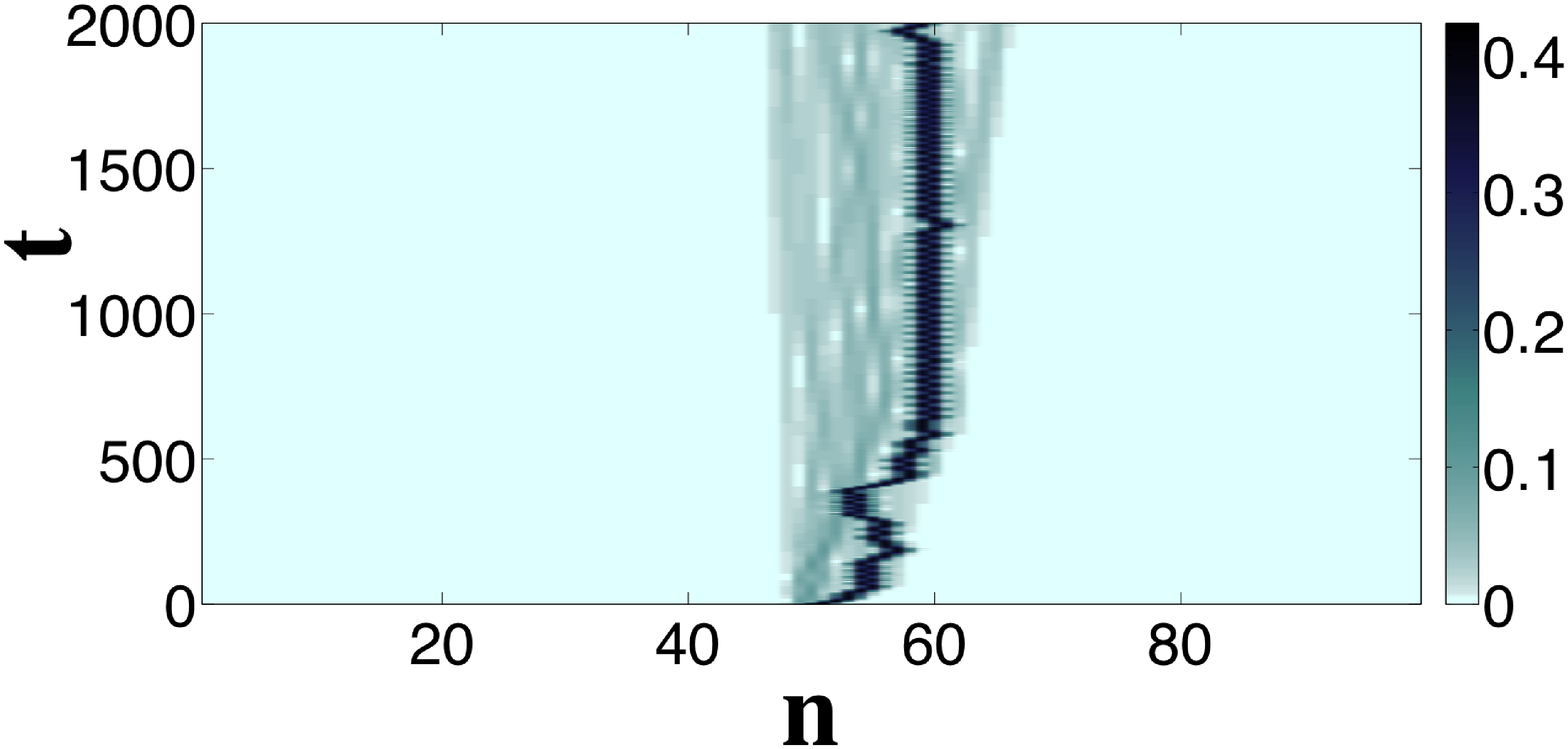}
\includegraphics[scale=0.2]{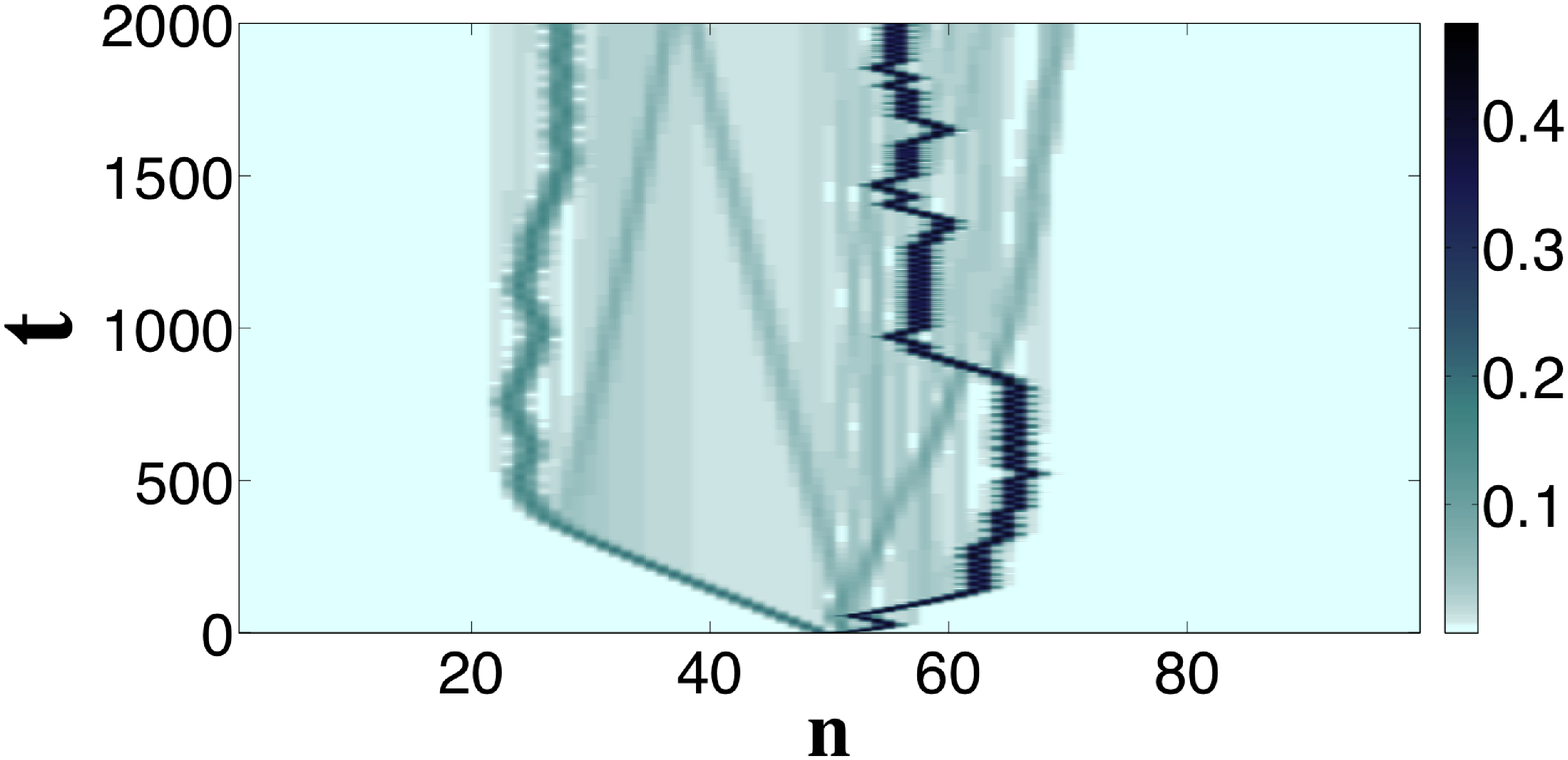}
\end{center}
\caption{\label{tbal3}
Same as in figure \ref{tbal1_5}, but in the case $\alpha=3$.
For $\delta H_{\rm{rel}} \approx 0.001$ (top plot),
the perturbed breather is wandering between two neighboring sites.
For $\delta H_{\rm{rel}} \approx 0.1$ (middle plot) and
$\delta H_{\rm{rel}}\approx 0.5$ (bottom plot), the breather
is able to propagate over a few sites but get pinned subsequently.
Breather trapping can be interrupted by 
phases of irregular motion which may 
originate from the interaction of the localized solution with
other waves. For the largest initial perturbation,
one can notice the emission of 
additional breathers with smaller amplitudes.}
\end{figure}

The above transition from breather mobility to 
pinning phenomena when $\alpha$ is increased
can be linked with a sharp increase of the Peierls-Nabarro (PN) 
energy barrier between site- and bond-centered breathers. 
The PN barrier \cite{kivc} approximates the amount of energy 
$H$ required
for the depinning of a stable bond-centered breather, resulting in a translational motion.
It is defined through $E_{\rm{PN}}=|H_{\rm{sc}}-H_{\rm{bc}}|$, where
$H_{\rm{sc}}$, $H_{\rm{bc}}$ denote the energies of site- and bond-centered breather
solutions (\ref{solper}) having the same squared $\ell_2$ norm $\sum_{n\in\mathbb{Z}}{|\phi_n|^2}$.
When $E_{\rm{PN}}$ is relatively high, a traveling breather can easily become
trapped after loosing some energy through dispersion or during interaction with other waves.
Figure \ref{pnb} displays $E_{\rm{PN}}$ for different values of $\alpha$, as well as
the relative energy difference $E_{\rm{rel}}=E_{\rm{PN}}/H_{\rm{sc}}$
(the frequency of the site-centered breather is set to unity in this computation).
One can observe a qualitative change in the energy curves around $\alpha =2.5$.
For $\alpha \leq 2.5$, the absolute and relative PN barriers are very small,
and they become much larger for $\alpha \in [3,12]$ 
(e.g. $E_{\rm{rel}}$ increases by $5$ orders of magnitude
between $\alpha = 3/2$ and $\alpha=3$). In this interval
the growth of $E_{\rm{PN}}$ becomes almost linear in $\alpha$.
The increase of $E_{\rm{PN}}$ explains why traveling breathers are easily
generated from static breathers when $\alpha$ is close to unity, 
whereas pinning dominates for sufficiently large values of $\alpha$.

\begin{figure}[h]
\psfrag{10.0}[1][Bl]{{\scriptsize $10$}}
\psfrag{0.1}[1][Bl]{{\scriptsize $10^{-1}$~~}}
\psfrag{1e-3}[1][Bl]{{\scriptsize $10^{-3}~~$}}
\psfrag{1e-5}[1][Bl]{{\scriptsize $10^{-5}~~$}}
\psfrag{1e-7}[1][Bl]{{\scriptsize $10^{-7}~~$}}
\psfrag{1e-9}[1][Bl]{{\scriptsize $10^{-9}~~$}}
\psfrag{1e-11}[1][Bl]{{\scriptsize $10^{-11}~~$}}
\psfrag{1e-13}[1][Bl]{{\scriptsize $10^{-13}~~$}}
\psfrag{1e-15}[1][Bl]{{\scriptsize $10^{-15}~~$}}
\psfrag{0.18}[1][Bl]{~}
\psfrag{0.16}[1][Bl]{~}
\psfrag{0.14}[1][Bl]{~}
\psfrag{0.12}[1][Bl]{~}
\psfrag{0.10}[1][Bl]{~}
\psfrag{0.08}[1][Bl]{~}
\psfrag{0.06}[1][Bl]{~}
\psfrag{0.04}[1][Bl]{~}
\psfrag{0.02}[1][Bl]{~}
\psfrag{0.00}[1][Bl]{{\small $0$}}
\psfrag{0.20}[1][Bl]{{\small $0.2$}}
\psfrag{0}[1][Bl]{{\small $0$}}
\psfrag{2}[1][Bl]{{\small $2$}}
\psfrag{4}[1][Bl]{{\small $4$}}
\psfrag{6}[1][Bl]{{\small $6$}}
\psfrag{8}[1][Bl]{{\small $8$}}
\psfrag{10}[1][Bl]{{\small $10$}}
\psfrag{12}[1][Bl]{{\small $12$}}
\psfrag{alpha}[1][Bl]{{~~~~~~~~~$\alpha$}}
\psfrag{E_PN}[1][Bl]{\begin{tabular}{c}$E_{\rm{PN}}$\\ ~\end{tabular}}
\psfrag{RE_PN}[1][Bl]{\begin{tabular}{c}$E_{\rm{rel}}$\end{tabular}}
\begin{center}
\includegraphics[scale=0.31]{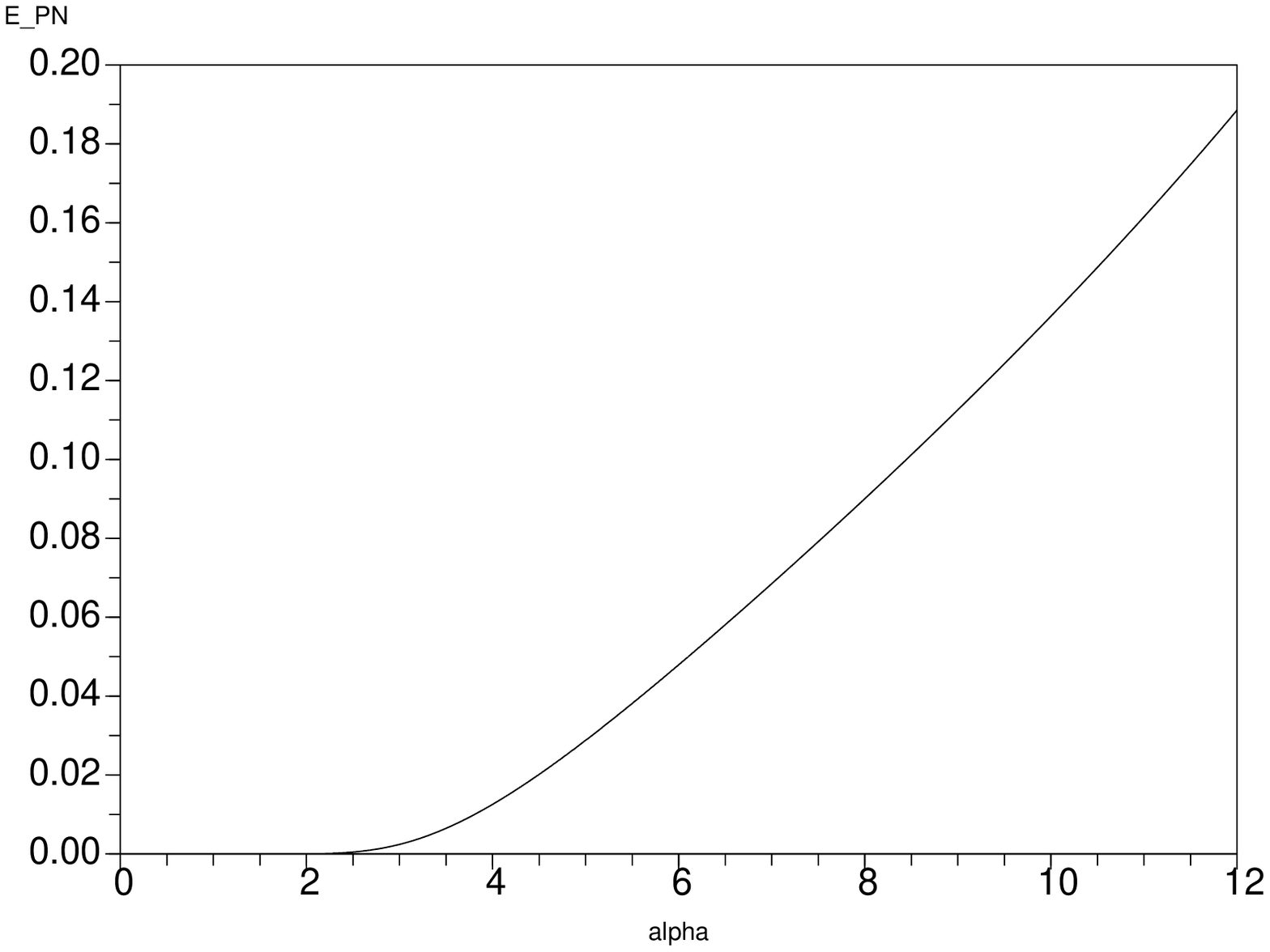}
\hspace*{-6ex}
\includegraphics[scale=0.31]{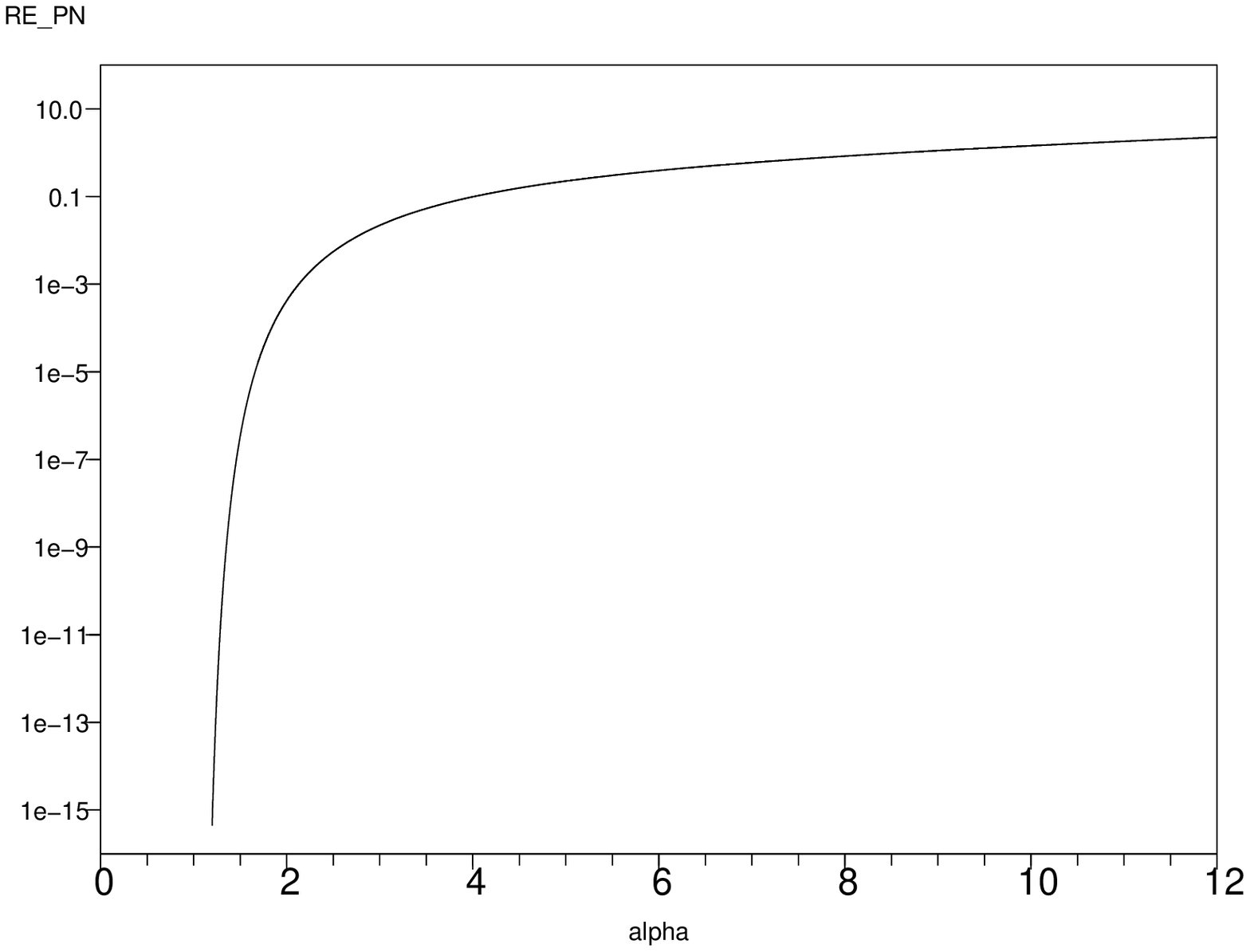}
\end{center}
\caption{\label{pnb}
Left~: Peierls-Nabarro energy barrier $E_{\rm{PN}}$ between the site-centered breather with frequency $\Omega=1$
and the bond-centered breather having the same $\ell_2$ norm. 
Right~: relative energy difference between the site- and bond-centered breathers
(note the semi-logarithmic scale).}
\end{figure}

\section{\label{sectdiscuss}Discussion}

In this paper, we have proved an existence theorem for
homoclinic orbits of the stationary DpS equation (\ref{dpsstat}), in agreement
with previous numerical results \cite{James11,sahvm,JKC11}.
These homoclinics correspond to breather solutions of the time-dependent
DpS equation (\ref{dpsa}). 
This result implies the existence of long-lived breather
solutions in infinite chains of oscillators with fully nonlinear nearest-neighbors interactions \cite{bdj},
in particular granular chains with Hertzian contact interactions \cite{James11,sahvm,JKC11}.
In addition, we have obtained analytical approximations of breather profiles and
of the corresponding intersecting stable and unstable manifolds. For $p\approx 2$,
this was achieved
by deriving a suitable continuum limit of the stationary DpS equation, consisting of
a logarithmic nonlinear Schr\"odinger equation. We have also numerically
determined the spectral stability of breather solutions depending on their
odd- or even-parity symmetry, and we have studied their mobility properties.

An interesting extension to this work would be to consider a DpS equation
including a localized defect
\begin{equation}
\label{dpsdefect}
i  \frac{d}{d t}\phi_n- \delta_{0,n}\, \nu\, \phi_n
=(\Delta_{p}\phi)_n, \ \
n\in \mathbb{Z},
\end{equation}
where $\nu$ is a real parameter defining the defect strength and
$\delta_{m,n}$ denotes the Kronecker delta.
For example,
system (\ref{dpsdefect}) can be derived in a similar way as in \cite{James11},
by considering a Newton's cradle including a small mass defect or
a small stiffness variation of one oscillator in the chain.
If $\nu \neq 0$, the Ansatz (\ref{solper}) defines a solution of (\ref{dpsdefect}) if
\begin{equation}
\label{dpsstatinhom}
-s\, {a_n}- \delta_{0,n}\, \lambda\, a_n= (\Delta_{p}a)_n, \ \
n\in \mathbb{Z},
\end{equation} 
where we have set $\nu = |\Omega |\, \lambda$ and $s=\mbox{Sign}(\Omega)$. 
An interesting problem is to analyze 
bifurcations of homoclinic
solutions of (\ref{dpsstatinhom}) when $\lambda$ is varied, or equivalently when
the defect strength and breather frequency are tuned. For example,
the associated solutions of (\ref{dpsdefect}) 
may correspond to localized modes excited by a traveling breather
reaching a defect \cite{JKC11}. In addition, it is also interesting to understand how the breathers
existing in a spatially homogeneous system are affected by spatial inhomogeneities,
since the latter are always present in real systems, or could be introduced 
dynamically to manipulate localized excitations.  

Similarly to what we have seen, one can 
reformulate (\ref{dpsstatinhom}) as a two-dimensional (nonautonomous) mapping.
The map possesses a particularly simple structure at the defect site when the mixed variables introduced
in section \ref{miv} are used. Indeed, the mapping (\ref{mappingn}) becomes
\begin{equation}
\label{mappinginhom}
\begin{pmatrix} v_{n+1} \\ y_{n+1}\end{pmatrix}
=\left(\mbox{Id}+ \delta_{0,n}\, \lambda \,
\begin{pmatrix} 0&0 \\ 1&0 \end{pmatrix} \, 
\right)
F
\begin{pmatrix} v_{n} \\ y_{n}\end{pmatrix} .
\end{equation}
This yields an interpretation of the defect as the composition of $F$ with 
the linear shear
\begin{equation}
\label{defal}
A(\lambda ) =\begin{pmatrix} 1&0 \\ \lambda&1 \end{pmatrix} ,
\end{equation}
which is quite convenient to analyze bifurcations of homoclinic solutions
(note that using the maps $M$ and $T$ introduced in section \ref{maps} would lead
to more complex perturbations induced by the defect). Indeed,
the mapping (\ref{mappinginhom}) admits 
an orbit homoclinic to $0$ if and only if 
$A(\lambda)\,  \mathcal{W}^{\text{u}}(0)$ and $\mathcal{W}^{\text{s}}(0)$ intersect
(see section \ref{defsumf} for the definition of the stable and unstable manifolds).
In that case, each intersection $(v_{1} , y_{1})$ determines an homoclinic orbit
of (\ref{mappinginhom}), hence 
an homoclinic solution of (\ref{dpsstatinhom}) with $s=1$.
Consequently,
the approximations of $\mathcal{W}^{\text{u}}(0)$ and $\mathcal{W}^{\text{s}}(0)$
introduced in section \ref{sectapprox} provide a theoretical tool to predict
defect-induced breather bifurcations, and approximate the shape and energy
of bifurcating solutions. This problem will be treated in detail in a forthcoming paper. 

\appendix
\begin{center}
{\Large\bf Appendix}
\end{center}

\section*{Analysis of the stationary DpS equation for $s=-1$}

In this appendix, we consider the stationary DpS equation with $s=-1$ and
prove the non-existence of nontrivial bounded solutions
stated in theorem \ref{mainthm}.

\begin{lemma}
\label{nonexist}
For $s=-1$, the only bounded solution of (\ref{dpsstat}) is $a_n=0$.
\end{lemma}
\begin{proof}
We restrict to the case $a_0 \geq 0$ due to the invariance $a_n \rightarrow -a_n $ of (\ref{dpsstat}).
Let us first assume
\begin{equation}
\label{case1}
a_1 >0, \ \ \ a_0 \geq 0, \ \ \ a_1 \geq a_0 .
\end{equation}
Using (\ref{dpsstat}), one can show by induction that $(a_n)_{n\geq 0}$ is a positive
non-decreasing sequence. Then (\ref{dpsstat}) implies 
$(a_{n+1}-a_n)^\alpha \geq a_n $ for all $n\geq 0$. It follows that
$a_{n+1}-a_n \geq a_1^{1/\alpha}$ for all $n\geq 1$, 
hence $\lim_{n\rightarrow +\infty}{a_n}=+\infty$.

If $0 \leq a_1 <a_0$, then $\tilde{a}_n = a_{-n+1}$ defines a solution of (\ref{dpsstat})
satisfying (\ref{case1}), and thus $\lim_{n\rightarrow -\infty}{a_n}=+\infty$.

If $a_1 \leq 0 \leq a_0$ and $(a_0 , a_1)\neq (0,0)$, the case $n=1$ of (\ref{dpsstat})
yields $a_2 < a_1 \leq 0$. Then $\tilde{a}_n = -a_{n+1}$ defines a solution of (\ref{dpsstat})
satisfying (\ref{case1}), and thus $\lim_{n\rightarrow +\infty}{a_n}=+\infty$.
{\hfill $\Box$}\end{proof}

\section*{Acknowledgments}

G.J. acknowledges financial support from the Rh\^one-Alpes Complex Systems Institute (IXXI).
Y.S. is grateful to Israel Science Foundation (Grant 484 /12)
for financial support.

\end{document}